\documentclass[11pt,letterpaper]{article}
\usepackage[margin=1in]{geometry}
\usepackage[T1]{fontenc}
\usepackage[utf8]{inputenc}
\usepackage{lmodern}
\usepackage{amsmath,amssymb,amsthm,mathtools,mathrsfs}
\usepackage{microtype}
\usepackage{needspace}
\usepackage{graphicx,xcolor}
\usepackage[section]{placeins}
\usepackage{tikz}
\usetikzlibrary{arrows.meta,positioning,calc,fit,backgrounds,decorations.pathreplacing}
\usepackage{enumitem}
\usepackage{hyperref}
\usepackage{bookmark}
\definecolor{linkblue}{RGB}{24,63,98}
\definecolor{figblue}{RGB}{30,86,122}
\definecolor{figteal}{RGB}{22,111,109}
\definecolor{figamber}{RGB}{169,102,25}
\hypersetup{colorlinks=true,linkcolor=linkblue,citecolor=linkblue,urlcolor=linkblue,
 pdftitle={An exponential strong converse for private communication over degradable quantum channels},
 pdfauthor={Mark M. Wilde}}
\numberwithin{equation}{section}
\newtheorem{theorem}{Theorem}[section]
\newtheorem{lemma}[theorem]{Lemma}
\newtheorem{proposition}[theorem]{Proposition}
\newtheorem{corollary}[theorem]{Corollary}
\theoremstyle{definition}

\theoremstyle{remark}
\newtheorem{remark}[theorem]{Remark}
\DeclareMathOperator{\Tr}{Tr}
\DeclareMathOperator{\id}{id}
\DeclareMathOperator{\Sym}{Sym}
\DeclareMathOperator{\supp}{supp}

\newcommand{\cN}{\mathcal N}
\newcommand{\cA}{\mathcal A}
\newcommand{\cD}{\mathcal D}
\newcommand{\cT}{\mathcal T}

\newcommand{\hmin}{H_{\min}}
\newcommand{\hmax}{H_{\max}}
\newcommand{\ket}[1]{\lvert #1\rangle}
\newcommand{\bra}[1]{\langle #1\rvert}
\newcommand{\op}[1]{\left\|#1\right\|}
\newcommand{\trn}[1]{\left\|#1\right\|_1}
\newcommand{\states}{\mathcal S}
\newcommand{\substates}{\mathcal S_{\le}}
\newcommand{\qN}{q_{\cN}}
\newcommand{\FS}{\mathbb F}
\newcommand{\hbin}{h_2}
\newcommand{\Capp}{C_{\mathrm{app}}}
\title{\Large\bfseries An exponential strong converse for private communication\\ over degradable quantum channels}
\author{Mark M. Wilde\thanks{Email: \href{mailto:wilde@cornell.edu}{wilde@cornell.edu}}\\[3pt]
\textit{School of Electrical and Computer Engineering, Cornell University}\\
\textit{Ithaca, New York 14850, USA}}
\date{September 26, 2026}
\begin{document}
\maketitle
\begin{abstract}
We prove an exponential strong converse for the unassisted private capacity of every finite-dimensional degradable quantum channel. We use a single error criterion: the infidelity between the actual message--estimate--environment state and an ideal state consisting of uniformly distributed, perfectly correlated messages that are independent of the environment. At every rate strictly above the channel's maximum coherent information, the optimized fidelity to this ideal state decays exponentially in the number of channel uses. The same bound holds when the ideal environment state is fixed to the actual environment marginal. The proof applies to arbitrary mixed-state encoders and arbitrary collective decoders. Its main operator ingredient is an overlap bound for privacy tests: two receiver placements with a common subsystem of dimension $d_C$ have overlap at most $d_C/M$, where $M$ is the key size, independently of both shield dimensions. We combine this bound with the signed averaging construction presented in arXiv:2608.01308 and a single-filter reduction to symmetric states. A Gaussian filter replaces the common receiver subsystem by a system of controlled dimension from which the original correlations can be approximately recovered using quantum side information. The filter preserves exact exchange symmetry, and a uniform smooth-entropy estimate bounds its output dimension for arbitrary correlated channel inputs. We also obtain an exponential strong converse at zero for antidegradable channels under the same joint criterion and, consequently, for the private capacity of the quantum erasure channel throughout its parameter range.
\end{abstract}

\tableofcontents

\section{Introduction}\label{sec:intro}
\subsection{Background and motivation}
A basic question in communication theory is how much information can be transmitted reliably through a noisy physical medium. For private communication, reliability is only one requirement: the information must also remain secret from an observer with access to the environment. Conveniently, it is possible to treat these two requirements within one mathematical description. To understand the private communication task, consider that an isometric extension of a quantum channel produces a receiver output and an environment output. Discarding the environment gives the channel to the receiver, whereas discarding the receiver output gives a \emph{complementary channel} to the environment, which describes the information available to an eavesdropper.

The coding theorems of Cai, Winter, and Yeung~\cite{CWY2004} and Devetak~\cite{Devetak2005} characterize the unassisted private capacity by a regularized difference between the information available to the receiver and to the environment. In general, this characterization requires an optimization over input states correlated across arbitrarily many channel uses. Degradable channels, introduced by Devetak and Shor~\cite{DS2005}, are an important exception. For such a channel, a further channel acting on the receiver's output reproduces a complementary channel. The maximum coherent information is additive, and both the quantum and private capacities equal this single-use quantity~\cite{DS2005,Smith2008}.

A capacity formula describes communication in the limit of vanishing
error. It does not by itself rule out transmission above capacity with
an error bounded away from its maximal value. A \emph{strong converse}
excludes this possibility: at every rate strictly above capacity, the
error converges to its maximal value. An exponential strong converse
specifies an exponential rate of convergence. For private communication,
the precise error criterion matters because a code can have different
decoding and secrecy errors. Ref.~\cite[Sec.~4.1]{Wilde2017Position}
used a single trace-distance criterion that combines average decoding
error and secrecy. Here we use a fidelity-based analogue, comparing
the actual message--estimate--environment state with an ideal state
imposing both correct decoding and independence from the environment.
Our strong converse therefore rules out simultaneous reliability and
secrecy above capacity. In particular, if the decoding error tends to
zero, the secrecy error must tend to its maximal value. The operational
significance of maximal secrecy error is explored in
Ref.~\cite[Sec.~6]{salzmann2022totalinsecurity}, which studies
``total insecurity'' for encryption using imperfect keys when the
eavesdropper has additional information about the possible plaintexts.

Morgan and Winter established a pretty strong converse for degradable channels, including a private-communication version, and related the quantum strong-converse problem to symmetric channels~\cite{MW2014}. Their reduction separates a symmetric pair of outputs from an additional system that is needed by the legitimate receiver. Ref.~\cite{WTB2017} developed converse bounds for private communication using privacy tests and entanglement measures. These bounds provide strong-converse rates, but need not give the private capacity of every degradable channel. Baghali Khanian and Hirche studied converse bounds for antidegradable channels using separate decoding and secrecy errors~\cite{BKH2025}. Appendix~B of Ref.~\cite{WTB2017} explains how separate errors compare with a joint fidelity criterion; we return to this comparison in Appendix~\ref{app:marginal}.

Ref.~\cite{KBK2026} recently reported exponential strong converses for the quantum capacities of finite-dimensional degradable and antidegradable channels. The argument presented therein places a decoder on each of exponentially many receiver subsystems of a symmetric extension. An entanglement test follows each decoder, and pairs of tests have small overlap when their receiver subsystems share few outputs. A signed averaging matrix, constructed using a low-degree approximation of a Boolean function, converts these pairwise estimates into an exponentially small bound on the average test. The underlying polynomial approximation theorem is due to de Wolf~\cite{deWolf2008}.

To extend this argument to private communication, entanglement tests must be replaced by privacy tests. This replacement is not immediate. A private state has local shield systems, and the twisting unitaries of different privacy tests can act noncommutatively on the same sender shield~\cite{HHHO2005,HHHO2009}. Moreover, equality of the quantum and private capacities is an asymptotic statement and does not identify their finite-error coding problems. These observations motivate the two ingredients developed below: a privacy-test overlap bound independent of shield dimensions, and a reduction that controls the dimension of the quantum system common to the receiver placements.

\subsection{Summary of results}
We prove an exponential strong converse for the unassisted private capacity of every finite-dimensional degradable channel. We use a joint infidelity criterion, following the private-communication framework of Refs.~\cite{WTB2017,Wilde2017Position} and the uniform-prior formulation in Ref.~\cite[Eq.~(16.1.12)]{KLW2026}. This coding task is called \emph{secret-key transmission} in that reference. The same bound applies to its maximal-infidelity formulation of \emph{private communication}. It also applies, with exactly the same exponent, when the ideal environment state is required to equal the actual environment marginal. We keep these finite-error criteria distinct and prove the necessary comparisons explicitly.

The first ingredient is an overlap bound for decoded privacy-test projections. Private states were introduced in Refs.~\cite{HHHO2005,HHHO2009}. The associated verification tests were developed in Refs.~\cite{HHHLO2008QKD,HHHLO2008Privacy}; we use the privacy-test formulation in Refs.~\cite{WTB2017,KLW2026}. Two receiver placements with disjoint receiver systems have test overlap exactly $1/M$, where $M$ is the key size. If the placements instead share a receiver subsystem of dimension $d_C$, their overlap is at most $d_C/M$. Neither shield dimension appears in this estimate.
The proof first introduces isometries \(J_1\) and \(J_2\) whose images are the ranges of the two privacy-test projections, and then describes their action on each key label. The shared-sender twisting operators remain in their original order, and their product is unitary. A partial-trace identity involving a swap then treats the shared receiver subsystem. We explain the enlarged Hilbert spaces and the decoded projections in detail in Section~\ref{sec:privacy}.

The second ingredient is a single-filter reduction. The symmetric dilation of Morgan and Winter~\cite{MW2014} produces two interchangeable outputs and an additional receiver system. Using the full dimension of this additional system in the overlap bound would generally give the wrong rate threshold. Instead, we prove that one Gaussian linear filter replaces it by a system of controlled dimension from which the original state can be recovered using quantum side information. The filter acts on the original state and preserves every exchange symmetry exactly. Its output dimension is controlled by the conditional smooth max-entropy. A postselection argument~\cite{CKR2009,MW2014} and a dimension-uniform asymptotic equipartition estimate~\cite{TCR2009,Tomamichel2012} then bound that dimension for arbitrary correlated inputs, including exponentially small recovery errors.

Only the existence of one normalized filtered state is needed. The construction does not assert a deterministic compression protocol or a lower bound on postselection probability. We derive the Gaussian moment estimate explicitly, while recognizing earlier uses of Gaussian ensembles in quantum coding~\cite{HSW2008}. Recovery follows from Uhlmann's theorem. The resulting proof retains the sender shield throughout and applies the symmetric-state test bound directly, without an intermediate operational simulation.

Combining these ingredients with the signed averaging construction of Ref.~\cite{KBK2026} gives an explicit finite-block converse. Choosing its parameters proves exponential decay of the joint fidelity above the coherent-information threshold. The same overlap argument gives an exponential strong converse at zero for antidegradable channels. Consequently, the private capacity of the quantum erasure channel has an exponential strong converse throughout its parameter range. Appendix~\ref{app:marginal} treats the actual-environment-marginal criterion, gives an example distinguishing it from the optimized criterion, and derives the corresponding constraint on the sum of decoding and secrecy errors.

\subsection{Paper organization}
The rest of our paper is organized as follows. Section~\ref{sec:setup} gives the preliminaries, coding criteria, and main theorem. Section~\ref{sec:privacy} recalls the coherent reduction to privacy tests, explains decoded projections, and proves their overlap bound. Section~\ref{sec:symmetric} derives the symmetric-state bound, including the signed averaging argument. Section~\ref{sec:filter} proves the single-filter lemma and its Gaussian estimate. Section~\ref{sec:reduction} applies this lemma to a degradable code and bounds the filter dimension uniformly over correlated inputs. Section~\ref{sec:converse} proves the finite-block and exponential converses, and Section~\ref{sec:conclusion} concludes. Appendices~\ref{app:dilation}, \ref{app:entropies}, and~\ref{app:boolean} give the symmetric dilation, smooth-entropy estimates, and polynomial approximation input. Appendix~\ref{app:marginal} treats the actual-marginal error criterion.

\section{Preliminaries and the main theorem}\label{sec:setup}
In this paper, all Hilbert spaces considered are taken to be finite dimensional. A system label also denotes its Hilbert space, and $|A|\coloneqq\dim A$. We write $\states(A)$ for the density operators on $A$, $I_A$ for the identity operator, and $\id_A$ for the identity channel. An isometry also denotes its action on a state when the meaning is unambiguous. A pure-state ket and its density operator are distinguished by writing, for example, $\phi\coloneqq\ket\phi\!\bra\phi$. System labels are omitted for single-system definitions and included when several systems are involved. We write $\op X$ for the operator norm (the largest singular value) and $\trn X\coloneqq\Tr\sqrt{X^\dagger X}$ for the trace norm. More generally, $\left\|X\right\|_p\coloneqq(\Tr[|X|^p])^{1/p}$ for $p\in\{1,2,4\}$, where $|X|\coloneqq\sqrt{X^\dagger X}$; in particular, $\left\|X\right\|_2$ is the Hilbert--Schmidt norm. For a vector, an unsubscripted norm denotes its Hilbert-space norm. All logarithms are base two, except $\ln$ and $\exp$. Rates are measured in bits per channel use.

For positive semidefinite operators, the fidelity is defined as~\cite{Uhlmann1976}
\begin{equation}\label{eq:fidelity}
 F(\rho,\sigma)\coloneqq\trn{\sqrt\rho\sqrt\sigma}^{2}.
\end{equation}
For normalized states, the sine distance is defined as~\cite{Rastegin2002,Rastegin2006}
\begin{equation}\label{eq:purified}
 P(\rho,\sigma)\coloneqq\sqrt{1-F(\rho,\sigma)}.
\end{equation}
Fidelity is nondecreasing under a channel, and sine distance is nonincreasing. In Appendix~\ref{app:entropies}, we use the purified-distance extension to subnormalized states for smooth entropies~\cite{TCR2010}; on normalized states it agrees with the sine distance. We use Uhlmann's theorem~\cite{Uhlmann1976} and the inequality~\cite{FvdG1999}
\begin{equation}\label{eq:trace-p}
 \frac12\trn{\rho-\sigma}\le P(\rho,\sigma).
\end{equation}
In particular, if $0\le T\le I$ is a measurement effect and $\rho,\sigma$ are normalized states, then
\begin{equation}\label{eq:test-probability}
 \bigl|\Tr[T\rho]-\Tr[T\sigma]\bigr|
 \le\frac12\trn{\rho-\sigma}\le P(\rho,\sigma).
\end{equation}
Thus the change in the acceptance probability of an arbitrary test $T$ is at most the sine distance.

The von Neumann entropy and conditional entropy are defined as
\begin{equation}\label{eq:entropy}
 H(\rho)\coloneqq-\Tr[\rho\log\rho],\qquad
 H(A|B)_\rho\coloneqq H(AB)_\rho-H(B)_\rho.
\end{equation}
A quantum channel $\cN\colon A\to B$ is a completely positive, trace-preserving linear transformation from operators on $A$ to operators on $B$. Stinespring's theorem~\cite{Stinespring1955} gives an environment system $E$ and an isometry $U\colon A\to BE$, meaning that $U^\dagger U=I_A$, that realize the channel by discarding $E$. The associated \emph{complementary channel} $\cN^c\colon A\to E$ instead discards the receiver output $B$:
\begin{equation}\label{eq:complement}
 \cN(\rho)=\Tr_E[U\rho U^\dagger],\qquad
 \cN^c(\rho)=\Tr_B[U\rho U^\dagger].
\end{equation}
The pair $(\cN,\cN^c)$ therefore describes the receiver and environment marginals of the same isometric evolution. Complementary channels associated with different Stinespring representations are equivalent up to an isometry on the environment, after restriction to the relevant support and enlargement when needed. We fix one representation throughout. In the private-communication task, the environment is assigned to the eavesdropper.

Devetak and Shor introduced the class of \emph{degradable} channels~\cite{DS2005}. A channel is degradable if there is a channel $\cD\colon B\to E$ for which $\cN^c=\cD\circ\cN$. It is \emph{antidegradable} if $\cN=\cA\circ\cN^c$ for a channel $\cA\colon E\to B$. Define
\begin{equation}\label{eq:qN}
 \qN\coloneqq\max_{\rho\in\states(A)}
 \left\{H(\cN(\rho))-H(\cN^c(\rho))\right\}.
\end{equation}
For a degradable channel, $P(\cN)=Q(\cN)=\qN$~\cite{DS2005,Smith2008}. Also, $\qN\ge0$, because each pure input produces a pure state on $BE$ and hence gives coherent information equal to zero.

\subsection{A single joint error criterion}\label{subsec:criterion}
We define the unassisted coding task and the single error criterion used in the converse. We then compare the optimized target environment state with two other finite-error formulations.

An unassisted $(n,M)$ private code consists of states $\left(\rho^m_{A^n}\right)_{m=1}^M$ and a positive operator-valued measure $\left(\Lambda^j_{B^n}\right)_{j=1}^M$. The sender chooses $m$ uniformly, prepares $\rho^m_{A^n}$, and sends it through $\cN^{\otimes n}$. The receiver measures $B^n$ and records the outcome in $\widehat K$. There is no initial shared entanglement or key and no auxiliary communication. Both the input states and the measurement may be collective across all $n$ uses.

For each message, let
\begin{equation}\label{eq:conditional-code}
 \omega^m_{\widehat K E^n}
 \coloneqq\sum_{j=1}^M\ket j\!\bra j_{\widehat K}\otimes
 \Tr_{B^n}\!\left[(\Lambda^j_{B^n}\otimes I_{E^n})
 U^{\otimes n}\rho^m_{A^n}(U^\dagger)^{\otimes n}\right].
\end{equation}
The state including the message register is
\begin{equation}\label{eq:code-state}
 \omega_{K\widehat K E^n}\coloneqq\frac1M\sum_{m=1}^M
 \ket m\!\bra m_K\otimes\omega^m_{\widehat K E^n}.
\end{equation}
The ideal correlations and the code's optimized joint fidelity are given by
\begin{equation}\label{eq:ideal}
 \overline\Phi^M_{K\widehat K}\coloneqq\frac1M\sum_{m=1}^M
 \ket{mm}\!\bra{mm}_{K\widehat K},
 \qquad
 f(\mathscr C_n)\coloneqq\max_{\sigma\in\states(E^n)}
 F\!\left(\omega_{K\widehat K E^n},
 \overline\Phi^M_{K\widehat K}\otimes\sigma_{E^n}\right).
\end{equation}
The joint error is
\begin{equation}\label{eq:error}
 \varepsilon(\mathscr C_n)\coloneqq1-f(\mathscr C_n).
\end{equation}
The maximum in~\eqref{eq:ideal} exists by compactness and continuity. The use of one error criterion combining decoding and secrecy already appears in Ref.~\cite[Sec.~4.1, Eq.~(4.3)]{Wilde2017Position}, which uses joint trace distance with a common ideal environment state. The infidelity criterion~\eqref{eq:error} is the uniform-message instance of Ref.~\cite[Eq.~(16.1.12)]{KLW2026}. Both compare the actual output with one ideal state imposing correct decoding and independence from the environment, but their numerical errors are not identical at finite blocklength.

For clarity, this criterion is not the arithmetic mean of the messagewise fidelities. The block-diagonal structure gives
\begin{equation}\label{eq:cq-fidelity}
 F\!\left(\omega_{K\widehat K E^n},\overline\Phi^M_{K\widehat K}\otimes\sigma_{E^n}\right)
 =\left(\frac1M\sum_{m=1}^M
 \sqrt{F\!\left(\omega^m_{\widehat K E^n},
 \ket m\!\bra m_{\widehat K}\otimes\sigma_{E^n}\right)}\right)^2.
\end{equation}
The same $\sigma_{E^n}$ is used for every message.

One can instead require the ideal environment state to be the actual marginal $\omega_{E^n}$. We denote this fidelity and error by
\begin{equation}\label{eq:marginal-definition-main}
 \begin{aligned}
 f_{\mathrm{marg}}(\mathscr C_n)
 &\coloneqq F\!\left(\omega_{K\widehat K E^n},
 \overline\Phi^M_{K\widehat K}\otimes\omega_{E^n}\right),\\
 \varepsilon_{\mathrm{marg}}(\mathscr C_n)&\coloneqq1-f_{\mathrm{marg}}(\mathscr C_n).
 \end{aligned}
\end{equation}
This is the joint criterion discussed in Ref.~\cite[Appendix~B]{WTB2017}. Since $\omega_{E^n}$ is one candidate in~\eqref{eq:ideal}, we have $f_{\mathrm{marg}}\le f$. Thus every upper bound on $f$ immediately bounds $f_{\mathrm{marg}}$, with no change of exponent. Appendix~\ref{app:marginal} proves the finite-error comparisons and explains why the two fidelities need not be equal.

For the maximal joint-infidelity formulation, we follow Ref.~\cite[Eq.~(16.1.15)]{KLW2026}, rather than the average trace-distance criterion of Ref.~\cite{Wilde2017Position}. It is defined as
\begin{equation}\label{eq:max-error}
 \varepsilon_{\max}(\mathscr C_n)\coloneqq
 \min_{\sigma\in\states(E^n)}\max_{1\le m\le M}
 \left\{1-F\!\left(\omega^m_{\widehat K E^n},
 \ket m\!\bra m_{\widehat K}\otimes\sigma_{E^n}\right)\right\}.
\end{equation}
For each fixed $\sigma$, the square of the average of the square roots in~\eqref{eq:cq-fidelity} is at least the smallest messagewise fidelity. Optimizing over $\sigma$ therefore gives
\begin{equation}\label{eq:error-comparison}
 \varepsilon(\mathscr C_n)\le\varepsilon_{\max}(\mathscr C_n).
\end{equation}
A converse proved for~\eqref{eq:error} thus also holds for~\eqref{eq:max-error}.

A rate $R\ge0$ is achievable under the uniform joint criterion if there is a sequence of unassisted codes $(\mathscr C_n)_n$, with message sizes $M_n$, such that
\begin{equation}\label{eq:achievable-definition}
 \liminf_{n\to\infty}\frac1n\log M_n\ge R,
 \qquad \lim_{n\to\infty}\varepsilon(\mathscr C_n)=0.
\end{equation}
The capacity is the supremum of achievable rates. Replacing $\varepsilon$ by $\varepsilon_{\max}$ gives the same vanishing-error capacity, denoted by $P(\cN)$~\cite{Devetak2005,KLW2026}. This agreement of capacities does not identify the two criteria at a fixed error; the comparison in~\eqref{eq:error-comparison} is the finite-error implication used here.

\subsection{Why coherent information is the private-capacity threshold}\label{subsec:capacity-background}
We recall the standard argument for the capacity identity, both to fix its meaning and to distinguish it from the strong converse proved here~\cite{CWY2004,Devetak2005,DS2005,Smith2008}. For a state $\rho$ and a channel $\cN$, write
\begin{equation}\label{eq:coherent-information-background}
 I_{\mathrm c}(\rho,\cN)\coloneqq H(\cN(\rho))-H(\cN^c(\rho)).
\end{equation}
Let $V\colon B\to FE'$ dilate a degrading channel, with $E'\cong E$. No symmetry of this dilation is needed for the present argument. Isometric invariance of entropy and degradability imply
\begin{equation}\label{eq:background-conditional}
 I_{\mathrm c}(\rho,\cN)=H(F|E')_{VU\rho U^\dagger V^\dagger}.
\end{equation}
Conditional entropy is concave in the state. Since the state on the right depends linearly on $\rho$, coherent information is concave in the input of a degradable channel. Every pure input has coherent information equal to zero. Decomposing an arbitrary mixed input into pure states therefore gives
\begin{equation}\label{eq:background-nonnegative}
 I_{\mathrm c}(\rho,\cN)\ge0
 \quad\text{for every }\rho\in\states(A).
\end{equation}

For an ensemble $(p_x,\rho^x)_x$, let $\overline\rho\coloneqq\sum_xp_x\rho^x$, and let $X$ record the label. The mutual information is defined as $I(A;B)_\rho\coloneqq H(A)_\rho+H(B)_\rho-H(AB)_\rho$. The equality of private information and coherent information for degradable channels is due to Smith~\cite{Smith2008}. To recall the argument, define $\theta_{XBE}\coloneqq\sum_xp_x\ket x\!\bra x_X\otimes U\rho^xU^\dagger$. Expanding the mutual informations gives
\begin{equation}\label{eq:background-ensemble}
 I(X;B)_\theta-I(X;E)_\theta
 =I_{\mathrm c}(\overline\rho,\cN)-\sum_xp_xI_{\mathrm c}(\rho^x,\cN)
 \le\qN.
\end{equation}
Conversely, choose an input attaining $\qN$ and any pure-state decomposition of that input. Each term in the sum in~\eqref{eq:background-ensemble} is then equal to zero, so the ensemble attains $\qN$. Thus the single-use private information equals $\qN$, as proved in Ref.~\cite{Smith2008}.

Devetak and Shor proved additivity of the maximum coherent information for degradable channels~\cite{DS2005}; see also Ref.~\cite{Smith2008} for its private-capacity consequence. The following conditional-entropy proof makes the use of degradability explicit. For an arbitrary joint input to two degradable channels, apply their degrading isometries and use the entropy chain rule followed by strong subadditivity~\cite{LiebRuskai1973}:
\begin{equation}\label{eq:background-additivity}
 \begin{aligned}
 H(F_1F_2|E'_1E'_2)
 &=H(F_1|E'_1E'_2)+H(F_2|F_1E'_1E'_2)\\
 &\le H(F_1|E'_1)+H(F_2|E'_2)\\
 &\le q_{\cN_1}+q_{\cN_2}.
 \end{aligned}
\end{equation}
The first inequality removes conditioning systems in each term; both applications have the stated direction by strong subadditivity. Product inputs attaining the two individual maxima give the opposite inequality. Iteration proves additivity for every tensor power. The regularized private coding theorem therefore gives $P(\cN)=\qN$. Likewise the quantum coding theorem and the same additivity give $Q(\cN)=\qN$. These coding theorems characterize vanishing error; the remainder of this paper establishes what happens at every rate strictly larger than this threshold.

\subsection{Main result}
The following theorem strengthens the vanishing-error capacity formula to a bound that holds uniformly over every code above capacity, including codes with large decoding or secrecy errors.

\begin{theorem}[Exponential strong converse]\label{thm:main}
Let $\cN\colon A\to B$ be a finite-dimensional degradable quantum channel, and let $\qN$ be defined in~\eqref{eq:qN}. For every $\Delta>0$, there exist $\gamma>0$ and $n_0\in\mathbb N$, depending only on $\cN$ and $\Delta$, such that every unassisted $(n,M)$ private code with
\begin{equation}\label{eq:rate-gap}
 n\ge n_0,\qquad \log M\ge n(\qN+\Delta)
\end{equation}
satisfies
\begin{equation}\label{eq:main}
 f(\mathscr C_n)\le2^{-\gamma n},\qquad
 \varepsilon(\mathscr C_n)\ge1-2^{-\gamma n}.
\end{equation}
The same lower bound holds for $\varepsilon_{\max}(\mathscr C_n)$ and $\varepsilon_{\mathrm{marg}}(\mathscr C_n)$. Consequently, the unassisted private capacity has an exponential strong converse at $P(\cN)=\qN$ for the uniform optimized, maximal optimized, and uniform actual-marginal joint criteria.
\end{theorem}

The theorem is uniform over codes. In particular, neither $n_0$ nor $\gamma$ depends on the message states, on their ranks, or on the dimensions of local purifying systems. 

The proof has three steps. First, a coherent implementation of a private code gives a privacy test accepted with probability at least $f(\mathscr C_n)$. Second, a symmetric output state with a common receiver system of dimension $L$ satisfies a test bound involving $L/M$, independently of the shields. Third, a single filter converts a degradable code to this state-level setting with $\log L\le n\qN$ plus a controlled recovery-error correction. Section~\ref{sec:converse} makes the correction small enough while keeping all approximation errors exponentially small.

\section{Privacy tests and their overlaps}\label{sec:privacy}
This section isolates the operator statement that permits a private-communication version of the symmetric-channel argument. The sender's shield will be common to all tests. It is essential that no dimension factor is assigned to that shield.

\subsection{Private states and coherent implementations}
We first convert the joint fidelity of a private code into the acceptance probability of a privacy test. Retaining the local systems that purify the encoder and decoder produces the shield systems of that test. Figure~\ref{fig:coherent} illustrates the coherent implementation used below.

\begin{figure}
\centering
\resizebox{\linewidth}{!}{%
\begin{tikzpicture}[x=1cm,y=1cm,>=Latex,font=\small,
 qwire/.style={draw=black!75,line width=.65pt},
 box/.style={draw=figblue,line width=.8pt,fill=figblue!4,rounded corners=2pt},
 lab/.style={fill=white,inner sep=2pt}]
 \node[box,minimum width=1.45cm,minimum height=3cm] (prep) at (1.0,1.3) {$\ket\phi$};
 \node[above,font=\small\bfseries] at (1,3.0) {Sender};
 \node[box,minimum width=1.15cm,minimum height=1.3cm] (U) at (4.3,.2) {$U^{\otimes n}$};
 \node[box,minimum width=1.4cm,minimum height=1.6cm] (D) at (8,0) {$U_{\mathrm{dec}}$};
 \node[font=\small\bfseries] at (8,1.12) {Receiver};
 \node[draw=figteal,fill=figteal!5,line width=.9pt,rounded corners=2pt,
       minimum width=1.1cm,minimum height=3.8cm] (P) at (12.55,.95) {$\Pi^\gamma$};
 \draw[qwire] (1.725,2.4)--(12,2.4) node[pos=.15,lab] {$K$};
 \draw[qwire] (1.725,1.4)--(12,1.4) node[pos=.15,lab] {$S$};
 \draw[qwire,->] (1.725,.2)--(3.725,.2) node[midway,lab] {$A^n$};
 \draw[qwire,->] (4.875,.45)--(7.3,.45) node[midway,lab] {$B^n$};
 \draw[qwire] (4.875,-.1)--(5.35,-.1)--(5.35,-1.65)--(13.15,-1.65);
 \node[lab] at (7,-1.65) {$E^n$};
 \node[right] at (13.15,-1.65) {Eve};
 \draw[qwire] (8.7,.5)--(12,.5) node[midway,lab] {$\widehat K$};
 \draw[qwire] (8.7,-.5)--(12,-.5) node[midway,lab] {$T$};
 \draw[qwire,->] (13.1,.95)--(14.0,.95);
 \node[align=center,right] at (14.0,.95) {accept\\or reject};
 \node[align=center,text=figteal] at (12.55,3.4) {Privacy test\\(converse only)};
 \node[align=center,text=black!65,font=\footnotesize] at (7.0,-2.25)
 {Tracing out $S$ and $T$ recovers the original message--estimate--environment state.};
\end{tikzpicture}%
}\par
\caption{A coherent implementation of a private code. The sender retains separate key and shield systems $K$ and $S$. The channel dilation $U^{\otimes n}$ produces the receiver output $B^n$ and the environment~$E^n$. The receiver retains the environment $T$ of its coherent decoder $U_{\mathrm{dec}}$. The privacy test acts jointly on $K,S,\widehat K,T$; it is a converse test, not an allowed communication resource. Its acceptance probability is at least the optimized joint fidelity of the original code.}
\label{fig:coherent}
\end{figure}

Let $K$ and $\widehat K$ be $M$-dimensional key systems, and let $S$ and $T$ be the sender and receiver shields. Define
\begin{equation}\label{eq:bell}
 \ket{\Phi^M}_{K\widehat K}\coloneqq\frac1{\sqrt M}\sum_{m=1}^M\ket{mm}_{K\widehat K}.
\end{equation}
A twisting unitary has the form
\begin{equation}\label{eq:twisting}
 V_{\mathrm{tw}}\coloneqq\sum_{i,j=1}^M
 \ket i\!\bra i_K\otimes\ket j\!\bra j_{\widehat K}\otimes V^{ij}_{ST},
\end{equation}
where each $V^{ij}_{ST}$ is unitary. Let $\tau_{ST}\in\states(ST)$ be an arbitrary normalized density operator on the two shield systems; it may be entangled across $S:T$. A private state and its privacy-test projection are, respectively,
\begin{equation}\label{eq:private-test}
 \gamma_{KS\widehat K T}
 =V_{\mathrm{tw}}(\Phi^M_{K\widehat K}\otimes\tau_{ST})V_{\mathrm{tw}}^\dagger,
 \qquad
 \Pi^\gamma_{KS\widehat K T}
 =V_{\mathrm{tw}}(\Phi^M_{K\widehat K}\otimes I_{ST})V_{\mathrm{tw}}^\dagger.
\end{equation}
Tensor factors in this display are reordered canonically. The operator $\tau_{ST}$ specifies the shield state before twisting; it need not equal the marginal $\gamma_{ST}$ after twisting. We fix a twisting representation of $\gamma$ and write $\Pi^\gamma$ for its associated test; this notation does not assert that the test is uniquely determined by the density operator $\gamma$. The test accepts $\gamma$ with probability one. The private-state definition follows Refs.~\cite{HHHO2005,HHHO2009}. Verification by untwisting and testing the key systems appears in Refs.~\cite{HHHLO2008QKD,HHHLO2008Privacy}; Ref.~\cite[Sec.~4.2]{WTB2017} explicitly credits these works when defining the privacy test. We use that definition and Ref.~\cite[Definition~15.12]{KLW2026}. A privacy test is generally a nonlocal measurement. We use it only as a mathematical test in a converse, not as a free operation in the communication protocol.

The next lemma is the standard coherent correspondence between approximate secret correlations and private states, followed by the privacy-test acceptance bound; see Ref.~\cite[Chapter~15]{KLW2026} and Ref.~\cite[Sec.~2 and Lemma~9]{WTB2017}. We retain the proof to identify precisely the systems and the fidelity used in the converse.

\begin{lemma}[Coherent lift of the joint criterion]\label{lem:lift}
For every unassisted code $\mathscr C_n$, there is a coherent implementation with sender systems $KS$, receiver systems $\widehat K T$, and environment $E^n$, and there is a size-$M$ privacy-test projection $\Pi^\gamma$, such that its final state $\omega_{KS\widehat K T}$ satisfies
\begin{equation}\label{eq:lift}
 \Tr[\Pi^\gamma\omega_{KS\widehat K T}]\ge f(\mathscr C_n).
\end{equation}
The coherent implementation uses no initial shared resource.
\end{lemma}
\begin{proof}
Choose a purification $\ket{\phi_m}_{S_1A^n}$ of each message state $\rho^m_{A^n}$. The sender prepares
\begin{equation}\label{eq:coherent-encoder}
 \ket\phi_{KS_0S_1A^n}\coloneqq\frac1{\sqrt M}\sum_{m=1}^M
 \ket m_K\ket m_{S_0}\ket{\phi_m}_{S_1A^n},\qquad S\coloneqq S_0S_1.
\end{equation}
Tracing out $S_0S_1$ gives the original classical message--input state. The receiver implements the measurement by the isometry
\begin{equation}\label{eq:coherent-decoder}
 U_{\mathrm{dec}}\colon B^n\to\widehat K T_0T_1,\qquad
 U_{\mathrm{dec}}\coloneqq\sum_{j=1}^M\ket j_{\widehat K}\ket j_{T_0}\otimes\sqrt{\Lambda^j},
\end{equation}
where $T_1\cong B^n$ and $T\coloneqq T_0T_1$. The identity $\sum_j\Lambda^j=I_{B^n}$ implies $U_{\mathrm{dec}}^\dagger U_{\mathrm{dec}}=I_{B^n}$. Apply the channel isometry to the input share $A^n$ of $\ket\phi$ and then apply $U_{\mathrm{dec}}$ to its receiver output $B^n$. The resulting state is
\begin{equation}\label{eq:coherent-output}
 \ket\omega_{KS\widehat K T E^n}
 \coloneqq (I_{KS}\otimes U_{\mathrm{dec}}\otimes I_{E^n})
 (I_{KS}\otimes U^{\otimes n})\ket\phi_{KSA^n}.
\end{equation}
Both transformations are isometries, so this state is pure and normalized. Tracing out $S_0$ removes the off-diagonal message terms, and tracing out $T_0$ removes the off-diagonal decoder outcomes. Tracing out the remaining shield factors $S_1T_1$ therefore gives
\begin{equation}\label{eq:coherent-output-marginal}
 \Tr_{ST}[\omega_{KS\widehat K T E^n}]=\omega_{K\widehat K E^n},
\end{equation}
with the right-hand side exactly as in~\eqref{eq:code-state}. In particular, the decoder's local purification has not changed the environment marginal.

Let $\sigma_{E^n}$ achieve the maximum in~\eqref{eq:ideal}. By Uhlmann's theorem, after enlarging local shields by pure auxiliary systems when necessary, there is a purification $\ket\gamma_{KS\widehat K T E^n}$ of $\overline\Phi^M_{K\widehat K}\otimes\sigma_{E^n}$ such that
\begin{equation}\label{eq:uhlmann-lift}
 |\langle\gamma|\omega\rangle|^2=f(\mathscr C_n).
\end{equation}
This purification has the form
\begin{equation}\label{eq:gamma-conditional}
 \ket\gamma=\frac1{\sqrt M}\sum_{m=1}^M\ket{mm}_{K\widehat K}\ket{\chi_m}_{STE^n},
 \qquad \Tr_{ST}[\chi_m]=\sigma_{E^n}.
\end{equation}
Every $\ket{\chi_m}$ purifies the same state. Purification uniqueness therefore gives unitaries $V^{mm}_{ST}$ for which $\ket{\chi_m}=V^{mm}_{ST}\ket{\chi_1}$. Choose arbitrary unitary blocks for $i\ne j$ in~\eqref{eq:twisting}. It follows that the marginal $\gamma_{KS\widehat K T}$ is a private state of the form~\eqref{eq:private-test}. In particular,
\begin{equation}\label{eq:test-range}
 (\Pi^\gamma\otimes I_{E^n})\ket\gamma=\ket\gamma.
\end{equation}
The Cauchy--Schwarz inequality now gives
\begin{equation}\label{eq:lift-cs}
 |\langle\gamma|\omega\rangle|^2
 =|\langle\gamma|(\Pi^\gamma\otimes I_{E^n})|\omega\rangle|^2
 \le\langle\omega|(\Pi^\gamma\otimes I_{E^n})|\omega\rangle,
\end{equation}
which proves~\eqref{eq:lift}.
\end{proof}

\subsection{Decoded test projections}\label{subsec:decoded-test}
This subsection defines decoded test projections and explains why we retain the decoder's auxiliary input system when comparing their operator norms. The overlap argument will compare different possible receivers before their decoders are applied. We therefore represent ``decode, then test'' as one projection on the decoder's input space. This change of representation is an operator identity, not a new step in the communication protocol.

Let $R$ be a receiver input system. Implement a receiver decoder by an isometry from $R$ to a key system $\widehat K$ and a retained shield. Introduce a receiver-only auxiliary system $G$, initially in a fixed pure state $\ket0_G$, and enlarge the retained receiver shield $T$ if needed. The isometry then extends to a unitary
\begin{equation}\label{eq:decoder-unitary}
 D\colon R G\longrightarrow\widehat K T.
\end{equation}
Such a unitary realization is always possible: one can add pure auxiliary systems on the output side until its dimension is a multiple of $|R|$, choose $G$ to make the input and output dimensions equal, and extend an orthonormal set to an orthonormal basis. All added output systems are included in $T$. The decoder acts trivially on the sender systems $KS$.

For a fixed privacy-test projection $\Pi^\gamma$ on $KS\widehat K T$, define its \emph{decoded test projection} by
\begin{equation}\label{eq:decoded-projection}
 P\coloneqq(I_{KS}\otimes D^\dagger)\Pi^\gamma(I_{KS}\otimes D).
\end{equation}
It acts on $KSRG$. Unitarity and $(\Pi^\gamma)^2=\Pi^\gamma$ imply $P^\dagger=P$ and $P^2=P$. For an arbitrary input state $\theta_{KSR}$, set $\widetilde\theta\coloneqq\theta\otimes\ket0\!\bra0_G$. Cyclicity of the trace gives
\begin{equation}\label{eq:decoded-acceptance}
 \Tr[P\widetilde\theta]
 =\Tr\!\left[\Pi^\gamma(I_{KS}\otimes D)\widetilde\theta
 (I_{KS}\otimes D^\dagger)\right].
\end{equation}
Thus the two descriptions in Figure~\ref{fig:decoded-test} have exactly the same acceptance probability.

There is a reason for retaining $G$ in~\eqref{eq:decoded-projection}. Compressing it to its initial vector would instead give the measurement effect
\begin{equation}\label{eq:compressed-effect}
 \Lambda_{\mathrm{test}}\coloneqq(I_{KSR}\otimes\bra0_G)P(I_{KSR}\otimes\ket0_G),
 \qquad 0\le\Lambda_{\mathrm{test}}\le I_{KSR}.
\end{equation}
Although $\Tr[\Lambda_{\mathrm{test}}\theta]=\Tr[P\widetilde\theta]$, the operator $\Lambda_{\mathrm{test}}$ need not be a projection. For example, compressing $\ket\Phi\!\bra\Phi_{RG}$ for a two-qubit maximally entangled state against $\ket0_G$ gives $\frac12\ket0\!\bra0_R$. Our argument needs projections, so all operator norms below are taken on the enlarged space with $G$ retained. When comparing different decoder placements, we give each placement its own auxiliary system. These systems are disjoint and do not increase the dimension shared by the receivers. A closely related use of Naimark dilation occurs in sequential decoding of general measurements: Ref.~\cite[Sec.~3, Lemma~3.1]{Wilde2013Sequential} introduces a separate auxiliary system for each measurement effect and applies a projection bound on the enlarged space. Here the enlarged-space projections represent alternative decoder placements, not a sequence of measurements. The same unitary-dilation convention is used for entanglement tests in Ref.~\cite[Supplemental Material, Sec.~S1]{KBK2026}.

\begin{figure}
\centering
\begin{tikzpicture}[x=1cm,y=1cm,font=\small,>=Latex,
  wire/.style={line width=.65pt},
  gate/.style={draw=figblue,line width=.8pt,fill=figblue!5,rounded corners=2pt},
  test/.style={draw=figteal,line width=.8pt,fill=figteal!5,rounded corners=2pt}]
\node[anchor=west,font=\small\bfseries] at (0,3.65) {(a) Decode, then test};
\foreach \y/\lab in {2.9/K,2.15/S,1.25/R}{
  \node[anchor=east] at (.5,\y) {$\lab$};
}
\node[anchor=east] at (.5,.50) {$\ket0_G$};
\draw[gate] (1.5,.15) rectangle (2.5,1.6);
\node at (2,.875) {$D$};
\draw[test] (3.95,.15) rectangle (4.95,3.25);
\node at (4.45,1.70) {$\Pi^\gamma$};
\draw[wire] (.65,2.9)--(3.95,2.9);
\draw[wire] (.65,2.15)--(3.95,2.15);
\draw[wire] (.65,1.25)--(1.5,1.25);
\draw[wire] (.65,.50)--(1.5,.50);
\draw[wire] (2.5,1.25)--(3.95,1.25);
\draw[wire] (2.5,.50)--(3.95,.50);
\node[above,font=\footnotesize] at (3.20,1.25) {$\widehat K$};
\node[above,font=\footnotesize] at (3.20,.50) {$T$};
\draw[->,wire] (4.95,1.70)--(5.95,1.70);
\node[above,font=\footnotesize] at (5.48,1.70) {accept};
\node[align=center,font=\footnotesize] at (2.95,-.35)
  {Sender systems are unchanged by $D$.};
\node[font=\Large] at (6.60,1.70) {$\equiv$};
\node[anchor=west,font=\small\bfseries] at (7.25,3.65) {(b) Test on the enlarged input};
\foreach \y/\lab in {2.9/K,2.15/S,1.25/R}{
  \node[anchor=east] at (8.2,\y) {$\lab$};
  \draw[wire] (8.35,\y)--(10.0,\y);
}
\node[anchor=east] at (8.2,.50) {$\ket0_G$};
\draw[wire] (8.35,.50)--(10.0,.50);
\draw[test] (10.0,.15) rectangle (11.0,3.25);
\node at (10.5,1.70) {$P$};
\draw[->,wire] (11.0,1.70)--(12.2,1.70);
\node[above,font=\footnotesize] at (11.6,1.70) {accept};
\node[align=center,font=\footnotesize] at (10.1,-.35)
  {$P$ is a projection on $KSRG$, not just $KSR$.};
\node[font=\small] at (6.15,-1.0)
  {$P=(I_{KS}\otimes D^\dagger)\,\Pi^\gamma\,(I_{KS}\otimes D),
    \qquad\widetilde\theta=\theta_{KSR}\otimes\ket0\!\bra0_G$.};
\end{tikzpicture}\par
\caption{Two descriptions of the same test probability. In (a), the receiver appends $\ket0_G$, applies the unitary decoder $D$, and the joint privacy test $\Pi^\gamma$ is applied to $KS\widehat K T$. In (b), the decoded projection $P=(I_{KS}\otimes D^\dagger)\Pi^\gamma(I_{KS}\otimes D)$ is applied directly to the enlarged input state on $KSRG$. Sender key and shield wires are shown separately. The equality concerns acceptance probabilities; neither joint test is an additional allowed communication operation. The system $G$ remains part of the operator space even though its input state is fixed.}
\label{fig:decoded-test}
\end{figure}
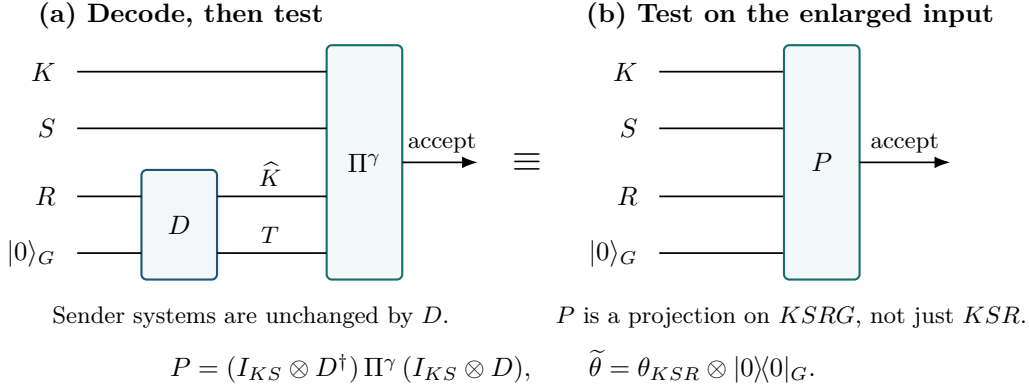

\subsection{A shield-independent overlap bound}\label{subsec:overlap}
We now bound the overlap of two decoded privacy-test projections by the dimension shared by their receiver inputs, with a bound independent of both shield dimensions. A \emph{receiver placement} specifies which tensor factors are supplied to a decoder. Consider two such choices: the first receiver uses $XC$ and the second uses $CY$. Here $X$ is available only to the first receiver, $Y$ only to the second, and $C$ is the subsystem present in both choices. These are alternative tests on one Hilbert space, not two simultaneously executed decoders. Both tests also use the same sender key $K$ and sender shield $S$.

There are two parts to the argument. With disjoint receivers, matching the common key label produces a factor $1/M$; all shield operators combine into unitaries and therefore contribute norm one. With a shared receiver system $C$, an operator identity reduces the problem to disjoint receivers with a multiplicative factor at most $|C|$. Figure~\ref{fig:overlap} displays the systems in these two steps.

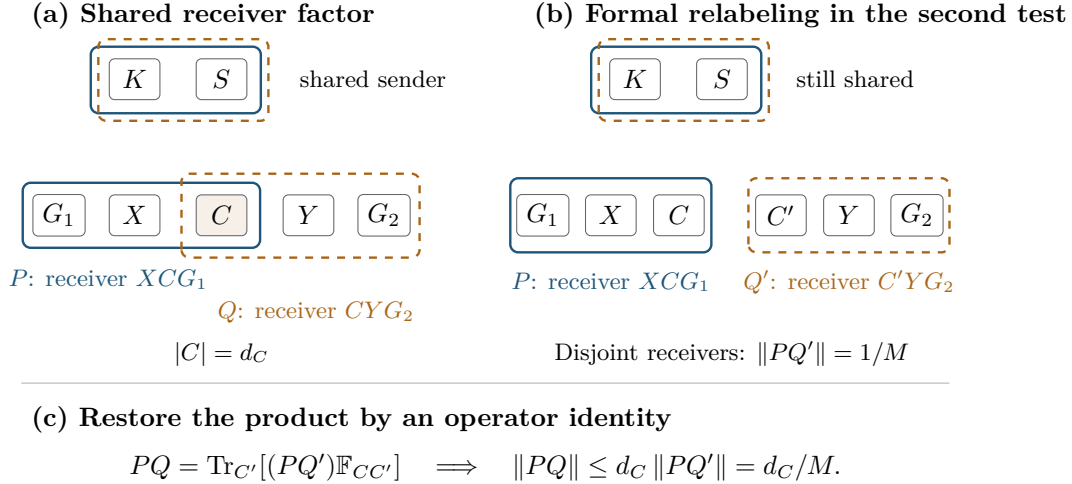
\begin{figure}
\centering
\begin{tikzpicture}[x=1cm,y=1cm,font=\small,>=Latex,
  sys/.style={draw=black!60,fill=white,rounded corners=2pt,minimum width=.67cm,minimum height=.54cm},
  first/.style={draw=figblue,line width=.9pt,rounded corners=3pt},
  second/.style={draw=figamber,dashed,line width=.9pt,rounded corners=3pt}]
\node[anchor=west,font=\small\bfseries] at (0,3.7) {(a) Shared receiver factor};
\node[sys] at (1.50,2.85) {$K$};
\node[sys] at (2.65,2.85) {$S$};
\draw[first] (.92,2.40) rectangle (3.17,3.28);
\draw[second] (1.02,2.30) rectangle (3.27,3.38);
\node[font=\footnotesize,anchor=west] at (3.55,2.85) {shared sender};
\node[sys] at (.5,1.05) {$G_1$};
\node[sys] at (1.50,1.05) {$X$};
\node[sys,fill=figamber!10] at (2.65,1.05) {$C$};
\node[sys] at (3.80,1.05) {$Y$};
\node[sys] at (4.80,1.05) {$G_2$};
\draw[first] (.02,.62) rectangle (3.16,1.48);
\draw[second] (2.12,.49) rectangle (5.27,1.60);
\node[text=figblue,font=\footnotesize] at (1.15,.20) {$P$: receiver $XCG_1$};
\node[text=figamber,font=\footnotesize] at (3.90,-.23) {$Q$: receiver $CYG_2$};
\node[font=\footnotesize] at (2.65,-.78) {$|C|=d_C$};
\node[anchor=west,font=\small\bfseries] at (6.65,3.7) {(b) Formal relabeling in the second test};
\node[sys] at (8.12,2.85) {$K$};
\node[sys] at (9.27,2.85) {$S$};
\draw[first] (7.54,2.40) rectangle (9.79,3.28);
\draw[second] (7.64,2.30) rectangle (9.89,3.38);
\node[font=\footnotesize,anchor=west] at (10.12,2.85) {still shared};
\node[sys] at (6.90,1.05) {$G_1$};
\node[sys] at (7.80,1.05) {$X$};
\node[sys] at (8.70,1.05) {$C$};
\node[sys] at (10.05,1.05) {$C'$};
\node[sys] at (10.95,1.05) {$Y$};
\node[sys] at (11.85,1.05) {$G_2$};
\draw[first] (6.47,.56) rectangle (9.13,1.54);
\draw[second] (9.62,.56) rectangle (12.28,1.54);
\node[text=figblue,font=\footnotesize] at (7.80,.17) {$P$: receiver $XCG_1$};
\node[text=figamber,font=\footnotesize] at (10.95,.17) {$Q'$: receiver $C'YG_2$};
\node[font=\footnotesize] at (9.40,-.78) {Disjoint receivers: $\left\| PQ'\right\|=1/M$};
\draw[black!30] (0,-1.17)--(12.28,-1.17);
\node[anchor=west,font=\small\bfseries] at (0,-1.65) {(c) Restore the product by an operator identity};
\node at (6.14,-2.28)
  {$PQ=\operatorname{Tr}_{C'}[(PQ')\mathbb F_{CC'}]
    \quad\Longrightarrow\quad
    \left\| PQ\right\|\le d_C\left\| PQ'\right\|=d_C/M.$};
\end{tikzpicture}\par
\caption{System supports in Lemma~\ref{lem:overlap}. The sender systems $K,S$ belong to both tests in every panel. In (a), the first receiver uses $X,C,G_1$ and the second uses $C,Y,G_2$, so only $C$ is a shared receiver factor. In (b), $Q'$ acts on a separate formal tensor factor $C'$ in place of $C$. The receiver supports are now disjoint, and their overlap is $1/M$ by~\eqref{eq:cross-action}. Panel (c) records the operator identity restoring the original product. The partial trace increases the norm by at most the factor $d_C=|C'|$. The auxiliary systems $G_1,G_2$ are included to make both decoded tests projections; they are never shared. The panels compare mathematical tests, not simultaneous decoding operations.}
\label{fig:overlap}
\end{figure}
\FloatBarrier

\begin{lemma}[Overlap of decoded privacy tests]\label{lem:overlap}
Let $K,S,X,C,Y,G_1,G_2$ be distinct tensor factors, let $|K|=M$ and $|C|=d_C$, and fix an orthonormal basis of $K$. The sender systems are $KS$. The two receiver inputs are $XC$ and $CY$, with separate pure auxiliary systems $G_1$ and $G_2$. Let
\begin{equation}\label{eq:two-decoders}
 D_1\colon XCG_1\longrightarrow K_1T_1,
 \qquad D_2\colon CYG_2\longrightarrow K_2T_2
\end{equation}
be unitary decoder realizations, where $|K_1|=|K_2|=M$ and $T_1,T_2$ contain all retained receiver outputs. When $C$ is nontrivial, $K_1T_1$ and $K_2T_2$ are alternative output spaces; they are not simultaneous subsystems of the common input space. For $i=1,2$, let $\Pi^{\gamma_i}$ be an arbitrary size-$M$ privacy-test projection on $KSK_iT_i$, with the fixed basis of $K$. Define
\begin{equation}\label{eq:two-pulled-tests}
 P\coloneqq(I_{KS}\otimes D_1^\dagger)\Pi^{\gamma_1}(I_{KS}\otimes D_1),
 \quad
 Q\coloneqq(I_{KS}\otimes D_2^\dagger)\Pi^{\gamma_2}(I_{KS}\otimes D_2),
\end{equation}
where $P$ and $Q$ are extended by identities on $YG_2$ and $XG_1$, respectively. They act on the common input space $KSXCYG_1G_2$ and satisfy
\begin{equation}\label{eq:overlap}
 \op{PQ}\le\min\left\{1,\frac{d_C}{M}\right\}.
\end{equation}
In particular, neither the sender shield dimension nor a receiver shield or auxiliary dimension appears in the bound.
\end{lemma}
\begin{proof}
\emph{Disjoint receiver systems.} First take $d_C=1$, so $C$ is trivial. The decoder unitaries act on $XG_1$ and $YG_2$, respectively, and hence on disjoint tensor factors. Define the unitary $ \mathsf U$ as follows:
\begin{equation}\label{eq:joint-decoder-unitary}
 \mathsf U\coloneqq I_{KS}\otimes D_1\otimes D_2
 \colon KSXG_1YG_2\longrightarrow KSK_1T_1K_2T_2.
\end{equation}
With identities on unused output factors understood, set
\begin{equation}\label{eq:undecoded-tests}
 \widehat P\coloneqq\Pi^{\gamma_1}\otimes I_{K_2T_2},
 \qquad \widehat Q\coloneqq\Pi^{\gamma_2}\otimes I_{K_1T_1}.
\end{equation}
The operator $P$ is the identity on the input of $D_2$, and $Q$ is the identity on the input of $D_1$. Consequently, the definitions in~\eqref{eq:two-pulled-tests} and the identities $D_iD_i^\dagger=I$ give
\begin{equation}\label{eq:decoder-conjugation}
 \mathsf U P\mathsf U^\dagger=\widehat P,
 \qquad \mathsf U Q\mathsf U^\dagger=\widehat Q,
 \qquad \mathsf U(PQ)\mathsf U^\dagger=\widehat P\widehat Q.
\end{equation}
The last equality follows by inserting $\mathsf U^\dagger\mathsf U=I$ between $P$ and $Q$. Unitary invariance of the operator norm therefore implies
\begin{equation}\label{eq:decoder-norm-invariance}
 \op{PQ}=\op{\mathsf U(PQ)\mathsf U^\dagger}
 =\op{\widehat P\widehat Q}.
\end{equation}
Thus both decoders have been removed from the calculation by one unitary change of coordinates; no commutation of the two privacy tests is assumed. Reorder the output factors as $KK_1K_2ST_1T_2$ for the remaining calculation.

Write $U_i$ for the diagonal twisting block of the first test, acting on $ST_1$, and $V_i$ for that of the second, acting on $ST_2$. Extend both by identities to the full shield system $\mathsf S\coloneqq ST_1T_2$. Only diagonal blocks enter a privacy test because $\ket{\Phi^M}$ has equal key labels. Define two isometries with explicit domains and the same codomain:
\begin{equation}\label{eq:range-domains}
 J_1\colon K_2\mathsf S\longrightarrow KK_1K_2\mathsf S,
 \qquad
 J_2\colon K_1\mathsf S\longrightarrow KK_1K_2\mathsf S.
\end{equation}
Their action on a key-basis vector and an arbitrary shield vector $\ket\xi_{\mathsf S}$ is
\begin{equation}\label{eq:range-isometries}
 \begin{aligned}
 J_1(\ket j_{K_2}\otimes\ket\xi)
 &\coloneqq\frac1{\sqrt M}\sum_{i=1}^M
 \ket i_K\ket i_{K_1}\ket j_{K_2}\otimes U_i\ket\xi,\\
 J_2(\ket j_{K_1}\otimes\ket\xi)
 &\coloneqq\frac1{\sqrt M}\sum_{i=1}^M
 \ket i_K\ket j_{K_1}\ket i_{K_2}\otimes V_i\ket\xi.
 \end{aligned}
\end{equation}
Orthogonality of the key basis and unitarity of the blocks give $J_1^\dagger J_1=I$ and $J_2^\dagger J_2=I$. Expanding the corresponding projections gives $\widehat P=J_1J_1^\dagger$ and $\widehat Q=J_2J_2^\dagger$, including the identity on the unused receiver key. Equivalently, the image of $J_1$ is the range of $\widehat P$, and the image of $J_2$ is the range of $\widehat Q$.

We next compute $J_1^\dagger J_2$. In the second line of~\eqref{eq:range-isometries}, $J_1^\dagger$ requires the labels on $K$ and $K_1$ to agree. Thus only the summand $i=j$ survives, and
\begin{equation}\label{eq:cross-action}
 J_1^\dagger J_2(\ket j_{K_1}\otimes\ket\xi)
 =\frac1M\ket j_{K_2}\otimes U_j^\dagger V_j\ket\xi.
\end{equation}
Equivalently,
\begin{equation}\label{eq:cross-isometry}
 J_1^\dagger J_2
 =\frac1M\sum_{j=1}^M\ket j_{K_2}\bra j_{K_1}\otimes U_j^\dagger V_j.
\end{equation}
The order $U_j^\dagger V_j$ is retained. Both factors can act on $S$ and need not commute, but their product is unitary. For an arbitrary vector $\sum_j\ket j_{K_1}\ket{\xi_j}_{\mathsf S}$, orthogonality on $K_2$ therefore gives
\begin{equation}\label{eq:cross-norm}
 \begin{aligned}
 \left\|J_1^\dagger J_2\sum_j\ket j_{K_1}\otimes\ket{\xi_j}\right\|^2
 &=\left\|\sum_j\frac1M\ket j_{K_2}\otimes U_j^\dagger V_j\ket{\xi_j}\right\|^2\\
 &=\frac1{M^2}\sum_j\left\|U_j^\dagger V_j\ket{\xi_j}\right\|^2
 =\frac1{M^2}\sum_j\left\|\ket{\xi_j}\right\|^2.
 \end{aligned}
\end{equation}
The first equality uses~\eqref{eq:cross-action} and linearity. In the second equality, every cross term with two distinct labels vanishes because $\langle j|k\rangle_{K_2}=0$ for $j\ne k$. The final equality uses unitarity of $U_j^\dagger V_j$. Since the squared norm of the input vector is $\sum_j\left\|\ket{\xi_j}\right\|^2$, this calculation proves $\op{J_1^\dagger J_2}=1/M$.

To pass from this norm to the projection product, use $\widehat P=J_1J_1^\dagger$ and $\widehat Q=J_2J_2^\dagger$. The isometries and their adjoints all have operator norm one. Submultiplicativity therefore gives
\begin{equation}\label{eq:projection-norm-upper}
 \begin{aligned}
 \op{\widehat P\widehat Q}
 &=\op{J_1(J_1^\dagger J_2)J_2^\dagger}\\
 &\le\op{J_1}\,\op{J_1^\dagger J_2}\,\op{J_2^\dagger}
 =\op{J_1^\dagger J_2}.
 \end{aligned}
\end{equation}
Conversely, $J_1^\dagger\widehat P\widehat QJ_2=J_1^\dagger J_2$, because $J_i^\dagger J_i=I$. Applying submultiplicativity a second time gives
\begin{equation}\label{eq:projection-norm-lower}
 \begin{aligned}
 \op{J_1^\dagger J_2}
 &=\op{J_1^\dagger(\widehat P\widehat Q)J_2}\\
 &\le\op{J_1^\dagger}\,\op{\widehat P\widehat Q}\,\op{J_2}
 =\op{\widehat P\widehat Q}.
 \end{aligned}
\end{equation}
Combining the two inequalities with~\eqref{eq:decoder-norm-invariance} proves
\begin{equation}\label{eq:disjoint-overlap}
 \op{PQ}=\op{\widehat P\widehat Q}
 =\op{J_1^\dagger J_2}=\frac1M.
\end{equation}
This proves the assertion for disjoint receiver systems, including arbitrary and possibly very large shields.

\emph{A shared receiver system.} Now let $C$ have dimension $d_C$. Work again with the decoded projections on the input space in~\eqref{eq:two-pulled-tests}. Introduce an additional tensor factor $C'\cong C$, and define $Q'$ by replacing $C$ with $C'$ in the definition of $Q$. All other tensor factors, in particular the common sender systems $KS$, are unchanged. This is a relabeling of an operator onto a larger Hilbert space; no state is copied, and no cloning operation is being asserted.

The receiver supports of $P$ and $Q'$ are $XCG_1$ and $C'YG_2$, which are disjoint. Thus the result just proved gives $\op{PQ'}=1/M$, even though $P$ and $Q'$ still share $KS$. Let $\FS_{CC'}$ be the unitary swap on $CC'$. The link to the original product is the identity
\begin{equation}\label{eq:swap-contraction}
 PQ=\Tr_{C'}\!\left[(PQ')\FS_{CC'}\right].
\end{equation}
We verify it explicitly, since the order of the operators on $KS$ matters. Set $\mathsf R\coloneqq KSXYG_1G_2$ and write all tensors in the order $\mathsf R,C,C'$. For a fixed basis of $C$ and the corresponding basis of $C'$, expand
\begin{equation}\label{eq:matrix-unit-expansion}
 \begin{aligned}
 P&=\sum_{i,j=1}^{d_C} P_{ij}\otimes\ket i\!\bra j_C\otimes I_{C'},\\
 Q'&=\sum_{k,l=1}^{d_C} Q_{kl}\otimes I_C\otimes\ket k\!\bra l_{C'},\\
 \FS_{CC'}&=\sum_{r,s=1}^{d_C}\ket r\!\bra s_C\otimes\ket s\!\bra r_{C'}.
 \end{aligned}
\end{equation}
The coefficients $P_{ij}$ and $Q_{kl}$ act on $\mathsf R$. Multiplication of $P$, $Q'$, and $ \FS_{CC'}$ gives $P_{ij}Q_{kl}$ in that order. On $C$ we obtain $\delta_{jr}\ket i\!\bra s$, while tracing the $C'$ factor gives $\delta_{ls}\delta_{kr}$. Summing these constraints yields
\begin{equation}\label{eq:swap-expanded}
 \Tr_{C'}[(PQ')\FS_{CC'}]
 =\sum_{i,j,l=1}^{d_C} P_{ij}Q_{jl}\otimes\ket i\!\bra l_C=PQ.
\end{equation}
In particular, no commutation of the coefficients is used.

For an arbitrary operator $Z$ on a system tensored with $C'$, writing the partial trace as $ \Tr_{C'}Z = \sum_{i=1}^{d_C}(I\otimes\bra i)Z(I\otimes\ket i)$ and using triangle inequality for the operator norm gives
\begin{equation}\label{eq:partial-trace-norm}
 \op{\Tr_{C'}Z}
 \le\sum_{i=1}^{d_C}\op{(I\otimes\bra i)Z(I\otimes\ket i)}
 \le d_C\op Z.
\end{equation}
Combining~\eqref{eq:swap-contraction}, this inequality, and unitarity of the swap, we conclude that
\begin{equation}\label{eq:shared-overlap}
 \op{PQ}\le d_C\op{(PQ')\FS_{CC'}}
 =d_C\op{PQ'}=\frac{d_C}{M}.
\end{equation}
Since each of $P$ and $Q$ is a projection, its norm is at most one and $\op{PQ}\le1$ as well. This proves~\eqref{eq:overlap}.
\end{proof}

The proof separates the two kinds of shared systems. The sender shield $S$ is shared even in the disjoint-receiver case, but its contribution is an ordered product of unitaries and introduces no dimension factor. Only the shared \emph{receiver} subsystem $C$ is removed by a partial trace in~\eqref{eq:swap-contraction}; this is the source of the factor $d_C$. The construction of $C'$ is only a device for proving an operator identity, not an extension of the physical code.

\section{Symmetric outputs with a common receiver system}\label{sec:symmetric}
We now use the overlap bound to control a whole family of decoder placements with identical acceptance probabilities. The goal is a bound that depends on the key size, the dimensions of the output pairs, and the dimension of one common receiver system, but not on the shields. We first prove this bound for a symmetric state; Section~\ref{sec:reduction} will construct the required state from a degradable code.

The construction used later in Section~\ref{sec:reduction} need not preserve a uniform key marginal or be a deterministic coding operation. This is why we formulate the bound for an arbitrary normalized symmetric state. A recoverable filter will produce one such state, and a separate fidelity comparison will relate its test probability to the original code.

\subsection{Decoder placements and their common Hilbert space}\label{subsec:placements}
This subsection constructs the projections whose average will bound the test probability. Each bit string selects one member of every symmetric output pair. The symmetry makes all selections statistically equivalent, while Lemma~\ref{lem:overlap} quantifies how much two selections can overlap. Figure~\ref{fig:placements} illustrates this distinction.

We use receiver-output labels $B_{j,0},B_{j,1}$ to emphasize that the argument is about alternative decoder inputs. Define
\begin{equation}\label{eq:pair-space}
 \mathbf B\coloneqq\bigotimes_{j=1}^n(B_{j,0}\otimes B_{j,1}),
 \qquad |B_{j,0}|=|B_{j,1}|=d.
\end{equation}
Let $\theta_{KSC\mathbf B}$ be an arbitrary state invariant under swapping $B_{j,0}$ and $B_{j,1}$ for each $j$ separately. Suppose that $|K|=M$ and $|C|=L$. There is no assumption on the marginal on $K$, or on correlations involving $C$. For $x\in\{0,1\}^n$, let $\bar x$ be its bitwise complement and set
\begin{equation}\label{eq:placement}
 B_x\coloneqq\bigotimes_{j=1}^nB_{j,x_j},\qquad
 s(x,y)\coloneqq\bigl|\{j:x_j=y_j\}\bigr|.
\end{equation}
Thus $\mathbf B$ is canonically identified with $B_xB_{\bar x}$ for every $x$. A fixed decoder on $CB_{0^n}$, followed by a fixed size-$M$ privacy test with sender systems $KS$, defines the actual test probability $p$. Identical copies of this decoder and test can be placed on $CB_x$. Independent exchange symmetry implies that each placement has the same probability $p$.

The labels in~\eqref{eq:pair-space} do not assign both outputs to the physical receiver. In the degradable-channel application we will identify $B_{j,0}$ with $E'_j$ and $B_{j,1}$ with $E_j$, as specified in~\eqref{eq:output-dictionary}. The constructed state's decoder uses $CE'^n$, whereas the original code's decoder uses $B^n$. Other selections are mathematical tests on the constructed state. Retaining their tensor factors in a converse does not make them accessible in the coding protocol. The privacy-test reduction avoids a separate estimate of secrecy, but the environment remains part of the isometric description that exposes the exchange symmetry.

We next define these tests on a common Hilbert space, using the convention of Section~\ref{subsec:decoded-test}. Give each placement a distinct auxiliary system $G_x$, initialized in a fixed pure state, and a unitary decoder
\begin{equation}\label{eq:placement-decoders}
 D_x\colon CB_xG_x\longrightarrow\widehat K_xT_x,
 \qquad |\widehat K_x|=M.
\end{equation}
Let $\Pi^{\gamma_x}_{KS\widehat K_xT_x}$ be the relabeled copy of the fixed privacy test. With $G_{\ne x}\coloneqq\bigotimes_{y\ne x}G_y$, define
\begin{equation}\label{eq:placement-projection}
 P_x\coloneqq
 \bigl[(I_{KS}\otimes D_x^\dagger)\Pi^{\gamma_x}
 (I_{KS}\otimes D_x)\bigr]\otimes I_{B_{\bar x}}\otimes I_{G_{\ne x}}.
\end{equation}
Canonical permutations place these tensor factors in the order
\begin{equation}\label{eq:common-hilbert}
 \mathcal H_{\mathrm{tot}}\coloneqq
 K\otimes S\otimes C\otimes\mathbf B
 \otimes\bigotimes_{x\in\{0,1\}^n}G_x.
\end{equation}
Every $P_x$ is an orthogonal projection on this single space. For the normalized state
$\widetilde\theta\coloneqq\theta\otimes\bigotimes_x\ket0\!\bra0_{G_x}$,
\eqref{eq:decoded-acceptance} and symmetry give
\begin{equation}\label{eq:average-acceptance}
 p=\Tr[P_x\widetilde\theta]=\Tr[A\widetilde\theta]\le\op A,
 \qquad A\coloneqq2^{-n}\sum_xP_x.
\end{equation}
Here $A\ge0$, so its expectation in a normalized state is at most its largest eigenvalue.

For distinct $x,y$, define their common receiver factor by
\begin{equation}\label{eq:shared-placement-factor}
 C_{x,y}\coloneqq C\otimes\bigotimes_{j:x_j=y_j}B_{j,x_j}.
\end{equation}
At a coordinate where $x_j\ne y_j$, the selections use different tensor factors. More explicitly, in Lemma~\ref{lem:overlap} take
\begin{equation}\label{eq:placement-lemma-identification}
 X=\bigotimes_{j:x_j\ne y_j}B_{j,x_j},\qquad
 Y=\bigotimes_{j:x_j\ne y_j}B_{j,y_j},\qquad
 G_1=G_x,\quad G_2=G_y,
\end{equation}
and take the shared receiver system to be $C_{x,y}$. The unselected output factors and all other auxiliary systems carry identities, which do not change the operator norm. Lemma~\ref{lem:overlap} gives
\begin{equation}\label{eq:placement-overlap}
 \op{P_xP_y}\le\min\left\{1,\frac{Ld^{s(x,y)}}M\right\}.
\end{equation}
The numerator is exactly the dimension of the shared receiver factor:
\begin{equation}\label{eq:shared-placement-dimension}
 |C_{x,y}|=|C|\prod_{j:x_j=y_j}|B_{j,x_j}|
 =L\,d^{s(x,y)}.
\end{equation}
Thus $L$ comes from the common system $C$, and each of the $s(x,y)$ agreeing coordinates contributes one factor $d$. Neither the common sender shield $S$ nor a dedicated $G_x$ contributes. The restriction $x\ne y$ matters because a placement shares its own auxiliary system with itself. The signed matrix used below avoids the diagonal altogether, so no diagonal overlap estimate is needed.

Only one decoder is used in the actual protocol. The family $(P_x)_x$ consists of alternative mathematical tests, not a prescription for executing all decoders or copying their inputs. Also, $\op A$ is a maximum over the full space~\eqref{eq:common-hilbert}, whereas $\widetilde\theta$ has specified pure auxiliary inputs. Enlarging the set of vectors in this maximum can only weaken the upper bound in~\eqref{eq:average-acceptance}. The unitary extensions away from the initialized auxiliary subspaces therefore need not represent additional operational resources.

\begin{figure}
\centering
\resizebox{\linewidth}{!}{%
\begin{tikzpicture}[x=1cm,y=1cm,font=\small,>=Latex,
 outputsystem/.style={draw=black!60,fill=white,rounded corners=2pt,minimum width=1.2cm,minimum height=.7cm},
 xs/.style={draw=figblue,line width=1.15pt,rounded corners=4pt},
 ys/.style={draw=figamber,dashed,line width=1.05pt,rounded corners=5pt}]
 \node[outputsystem,fill=figamber!8,minimum width=.85cm] at (.15,.85) {$C$};
 \node[font=\footnotesize,align=center] at (.15,-.15) {shared\\$|C|=L$};
 \foreach \j in {1,...,5}{
  \pgfmathsetmacro{\xx}{1.8+1.85*(\j-1)}
  \node[font=\footnotesize] at (\xx,2.65) {use $\j$};
  \node[outputsystem] (z\j) at (\xx,1.6) {$B_{\j,0}$};
  \node[outputsystem] (o\j) at (\xx,.1) {$B_{\j,1}$};
  \draw[<->,draw=black!40] (\xx,1.15)--(\xx,.55);
 }
 \foreach \j/\r in {1/z,2/z,3/o,4/z,5/o}{
  \node[xs,fit=(\r\j),inner sep=3pt] {};
 }
 \foreach \j/\r in {1/o,2/o,3/z,4/z,5/o}{
  \node[ys,fit=(\r\j),inner sep=6.5pt] {};
 }
 \node[text=figblue,anchor=west] at (10.5,1.8) {$x=00101$};
 \draw[xs] (10.5,1.3)--(11.4,1.3);
 \node[anchor=west,font=\footnotesize] at (11.55,1.3) {solid};
 \node[text=figamber,anchor=west] at (10.5,.45) {$y=11001$};
 \draw[ys] (10.5,-.05)--(11.4,-.05);
 \node[anchor=west,font=\footnotesize] at (11.55,-.05) {dashed};
 \node[font=\footnotesize,align=center] at (5.5,-1.0)
 {Two shared selected outputs: $s(x,y)=2$.\quad Common receiver dimension: $Ld^2$.};
\end{tikzpicture}%
}\par
\caption{Two of the $2^n$ receiver placements, illustrated for $n=5$. The solid selection is $x=00101$ and the dashed selection is $y=11001$; they agree at positions $4$ and $5$. Both also use $C$, so the common receiver factor is $C B_{4,0}B_{5,1}$ and has dimension $Ld^2$. The sender key and shield, which belong to every test, are not shown. Each placement has its own auxiliary input system. The two rows represent interchangeable output factors, not two outputs available simultaneously to the actual receiver. In Section~\ref{sec:reduction} they are identified with $E'^n$ and $E^n$.}
\label{fig:placements}
\end{figure}

\subsection{Signed averaging}\label{subsec:signed-averaging}
We now turn the pairwise bound~\eqref{eq:placement-overlap} into a small norm for the uniform average $A$. A two-test argument alone gives only a constant restriction on the common acceptance probability. Averaging the norms of all products is also insufficient for the rate threshold we need: two uniformly chosen strings typically agree on about half of their coordinates, whereas we need to use pairs with only a small fraction of common outputs.

The construction of Ref.~\cite{KBK2026} resolves this problem by working on the space of \emph{placement labels}. It builds a real matrix $W_t$ close in operator norm to the uniform matrix $U_0$, but with entries supported only on pairs satisfying $s(x,y)\le t$. The entries may have either sign; $W_t$ is an auxiliary operator, not a stochastic matrix or a physical measurement. Its absolute row sums bound the contribution of the signed entries.

The two lemmas below have distinct roles. Lemma~\ref{lem:weights} constructs $W_t$ using a low-degree approximation of a Boolean function. The Fourier transform on the Boolean cube is just the tensor-product Hadamard transform, and we give the required identities explicitly. Lemma~\ref{lem:averaging} then uses only elementary operator inequalities to combine the approximation error $\zeta$, the absolute row-sum bound $\ell$, and the overlap bound $\eta$ into
$\op A\le\zeta+\sqrt{\ell\eta}$.
In our application, $\eta=\min\{1,Ld^t/M\}$. For $t$ a sufficiently small positive fraction of $n$, the exponentially small ratio $L/M$ offsets both $d^t$ and the row-sum factor. Section~\ref{sec:converse} makes this choice precise.

\begin{lemma}[Signed Boolean weights; Ref.~\cite{KBK2026}, Lemma S3]\label{lem:weights}
There are universal constants $\Capp\ge1$ and $a>0$ such that, for every positive integer $n$ and integer $t$ satisfying
\begin{equation}\label{eq:t-range}
 4\Capp\sqrt n\le t\le n/2,
\end{equation}
there is a real symmetric matrix $W_t$ indexed by $\{0,1\}^n$ for which, with $\mathbf1$ the column vector of $2^n$ ones and $U_0\coloneqq2^{-n}\mathbf1\mathbf1^{\mathsf T}$,
\begin{equation}\label{eq:weight-properties}
 \begin{aligned}
 &W_t\mathbf1=\mathbf1,\qquad \op{W_t-U_0}\le2\exp(-at^2/n),\\
 &\max_x\sum_y|(W_t)_{xy}|\le\left(\sum_{j=0}^t\binom nj\right)^{1/2},\\
 &(W_t)_{xy}\ne0\quad\Longrightarrow\quad s(x,y)\le t.
 \end{aligned}
\end{equation}
One may take $a=1/(4\Capp^2)$ after increasing $\Capp$ if necessary.
\end{lemma}
\begin{proof}
We include a proof for completeness.

\emph{Step 1: a polynomial selecting the zero string.}
For $z\in\{0,1\}^n$, let $|z|$ be its number of nonzero entries and define
\begin{equation}\label{eq:NOR}
 \operatorname{NOR}_n(z)\coloneqq
 \begin{cases}1,&z=0^n,\\0,&z\ne0^n.\end{cases}
\end{equation}
The polynomial approximation theorem of de Wolf~\cite[Theorem~1]{deWolf2008}, specialized in Appendix~\ref{app:boolean}, gives a real multilinear polynomial $p$ such that
\begin{equation}\label{eq:approx-degree-main}
 \deg p\le\Capp\bigl(\sqrt n+\sqrt{n\ln(1/\epsilon)}\bigr),
 \qquad |p(z)-\operatorname{NOR}_n(z)|\le\epsilon
\end{equation}
for every $z$, whenever $2^{-n}\le\epsilon\le1/3$. This approximation theorem is the external input; the matrix construction and its bounds are proved here.

Set $a\coloneqq1/(4\Capp^2)$ and
\begin{equation}\label{eq:boolean-epsilon}
 \epsilon\coloneqq\exp\!\left(-\frac{t^2}{4\Capp^2n}\right).
\end{equation}
The lower bound in~\eqref{eq:t-range} gives $\epsilon\le e^{-4}<1/3$. The upper bound gives $\ln(1/\epsilon)\le n/(16\Capp^2)\le n\ln2$, so $\epsilon\ge2^{-n}$. Thus~\eqref{eq:approx-degree-main} applies and gives
\begin{equation}\label{eq:boolean-degree-check}
 \deg p\le\Capp\sqrt n+t/2\le3t/4\le t.
\end{equation}
Since $p(0^n)\ge1-\epsilon>0$, normalize the polynomial by setting $q(z)\coloneqq p(z)/p(0^n)$. Then
\begin{equation}\label{eq:normalized-poly}
 q(0^n)=1,\qquad |q(z)|\le\frac\epsilon{1-\epsilon}\le2\epsilon
 \quad(z\ne0^n),\qquad \deg q\le t.
\end{equation}
In particular, $|q(z)|\le1$ for every $z$.

\emph{Step 2: low degree becomes a restriction on shared outputs.}
For bit strings $e,z$, let $e\cdot z\coloneqq\sum_j e_jz_j$ modulo two. Define the Walsh coefficients by
\begin{equation}\label{eq:walsh}
 \widehat q(e)\coloneqq2^{-n}\sum_zq(z)(-1)^{e\cdot z},
 \qquad q(z)=\sum_e\widehat q(e)(-1)^{e\cdot z}.
\end{equation}
The inverse formula follows from character orthogonality:
\begin{equation}\label{eq:character-orthogonality}
 2^{-n}\sum_z(-1)^{e\cdot z}(-1)^{f\cdot z}
 =\prod_{j=1}^n\frac{1+(-1)^{e_j+f_j}}2
 =\delta_{e,f}.
\end{equation}
Each monomial in $q$ involves at most $t$ coordinates. Replacing a coordinate by $z_j=[1-(-1)^{z_j}]/2$ expresses such a monomial as a sum of characters involving only those coordinates. Consequently,
\begin{equation}\label{eq:walsh-support}
 \widehat q(e)=0\quad\text{whenever }|e|>t.
\end{equation}
Let $\oplus$ denote addition modulo two in each coordinate, and recall that $\bar x$ is the bitwise complement of $x$. Define
\begin{equation}\label{eq:W-construction}
 (W_t)_{xy}\coloneqq\widehat q(\bar x\oplus y).
\end{equation}
The string $\bar x\oplus y$ has a one exactly where $x$ and $y$ agree. Therefore $|\bar x\oplus y|=s(x,y)$, and~\eqref{eq:walsh-support} proves the required support restriction. Also $\bar x\oplus y=\bar y\oplus x$, so $W_t$ is real and symmetric. For a fixed $x$, the substitution $e=\bar x\oplus y$ is a bijection, giving
\begin{equation}\label{eq:row-sum}
 \sum_y(W_t)_{xy}=\sum_e\widehat q(e)=q(0^n)=1.
\end{equation}

\emph{Step 3: diagonalize the matrix in the Hadamard basis.}
For each $z$, define the vector $\chi_z$ on the label space by
$\chi_z(x)\coloneqq2^{-n/2}(-1)^{z\cdot x}$.
Equation~\eqref{eq:character-orthogonality} shows that these vectors form an orthonormal basis. Substituting $y=\bar x\oplus e$ in the matrix product gives
\begin{equation}\label{eq:W-spectrum}
 \begin{aligned}
 (W_t\chi_z)(x)
 &=2^{-n/2}\sum_e\widehat q(e)(-1)^{z\cdot(\bar x\oplus e)}\\
 &=(-1)^{|z|}\chi_z(x)\sum_e\widehat q(e)(-1)^{z\cdot e}\\
 &=(-1)^{|z|}q(z)\chi_z(x).
 \end{aligned}
\end{equation}
The sign $(-1)^{|z|}$ arises from complementing every bit of $x$. Thus $W_t$ has eigenvalue one on the uniform vector $\chi_{0^n}$ and eigenvalues of magnitude at most $2\epsilon$ on all orthogonal characters. Since $U_0=\chi_{0^n}\chi_{0^n}^\dagger$,
\begin{equation}\label{eq:W-approx}
 \op{W_t-U_0}=\max_{z\ne0^n}|q(z)|
 \le2\epsilon=2\exp(-at^2/n).
\end{equation}

\emph{Step 4: bound the absolute row sums.}
By~\eqref{eq:character-orthogonality}, Parseval's identity takes the form
\begin{equation}\label{eq:parseval-main}
 \sum_e|\widehat q(e)|^2=2^{-n}\sum_z|q(z)|^2\le1.
\end{equation}
For any row, the support restriction and Cauchy--Schwarz now imply
\begin{equation}\label{eq:W-row-l1}
 \begin{aligned}
 \sum_y|(W_t)_{xy}|
 &=\sum_{|e|\le t}|\widehat q(e)|\\
 &\le\left(\sum_{j=0}^t\binom nj\right)^{1/2}
       \left(\sum_{|e|\le t}|\widehat q(e)|^2\right)^{1/2}
 \le\left(\sum_{j=0}^t\binom nj\right)^{1/2}.
 \end{aligned}
\end{equation}
Equations~\eqref{eq:row-sum}, \eqref{eq:W-approx}, and~\eqref{eq:W-row-l1}, together with the support restriction, prove all claims.
\end{proof}

The construction separates two properties that might otherwise appear contradictory. In the coordinate basis, $W_t$ connects only almost complementary strings. In the Hadamard basis, it acts almost as a projection onto the uniform vector. Signed entries make these properties compatible. Because $t<n$, $s(x,x)=n$ also gives $(W_t)_{xx}=0$, so the diagonal projection overlaps never enter.

The next lemma explains how to use this matrix without interpreting its entries as probabilities. It compares two quadratic forms on an auxiliary label space tensored with the physical Hilbert space of the tests. In this representation, a matrix entry naturally multiplies a product $P_xP_y$, exactly the quantity controlled by the overlap lemma (Lemma~\ref{lem:overlap}).

\begin{lemma}[Average of projections; cf. Ref.~\cite{KBK2026}, Lemma S1]\label{lem:averaging}
Let $(P_x)_{x\in\{0,1\}^n}$ be orthogonal projections on a common finite-dimensional Hilbert space $\mathcal H$, and set $A\coloneqq2^{-n}\sum_xP_x$. Let $W$ be a real symmetric matrix and let $\zeta,\ell,\eta\ge0$ satisfy
\begin{equation}\label{eq:general-weights}
 \op{W-U_0}\le\zeta,\qquad
 \max_x\sum_y|W_{xy}|\le\ell,\qquad
 W_{xy}\ne0\ \Longrightarrow\ \op{P_xP_y}\le\eta.
\end{equation}
Then
\begin{equation}\label{eq:averaging-bound}
 \op A\le\zeta+\sqrt{\ell\eta}.
\end{equation}
\end{lemma}
\begin{proof}
Since $A\ge0$, its operator norm $\lambda\coloneqq\op A$ is its largest eigenvalue. Choose a unit vector $\ket v\in\mathcal H$ satisfying $A\ket v=\lambda\ket v$. Introduce the label space $\mathcal H_{\mathrm{lab}}\coloneqq\mathbb C^{2^n}$, with orthonormal basis $(\ket x)_x$, and the vector
\begin{equation}\label{eq:block-vector}
 \ket w\coloneqq\sum_x\ket x_{\mathrm{lab}}\otimes P_x\ket v
 \quad\in\mathcal H_{\mathrm{lab}}\otimes\mathcal H.
\end{equation}
Equivalently, $\ket w$ is the direct sum of the vectors $P_x\ket v$. The first tensor factor records a mathematical placement label. It is not an extra register available in the communication protocol. Matrices $W,U_0$ act on this first factor, whereas $I_{\mathcal H}$ is the identity on the full space of the decoded tests, including their auxiliary systems.

Orthogonality of the label basis and the identities $P_x^\dagger P_x=P_x$ give
\begin{equation}\label{eq:block-norm}
 \left\|\ket w\right\|^2=\sum_x\bra vP_x\ket v
 =2^n\bra vA\ket v=2^n\lambda.
\end{equation}
For the uniform matrix, expanding on the label factor gives
\begin{equation}\label{eq:block-uniform}
 \begin{aligned}
 \bra w(U_0\otimes I_{\mathcal H})\ket w
 &=2^{-n}\sum_{x,y}\bra vP_xP_y\ket v\\
 &=2^{-n}\left\|\sum_xP_x\ket v\right\|^2
 =2^{-n}\left\|2^n\lambda\ket v\right\|^2=2^n\lambda^2.
 \end{aligned}
\end{equation}
This is why the tensor factor structure is useful: the uniform label matrix collects all cross terms of the physical sum $\sum_xP_x\ket v$.

For $W$, the same expansion yields
$\bra w(W\otimes I_{\mathcal H})\ket w=\sum_{x,y}W_{xy}\bra vP_xP_y\ket v$.
When $W_{xy}\ne0$, the third assumption in~\eqref{eq:general-weights} and $\left\|\ket v\right\|=1$ imply $|\bra vP_xP_y\ket v|\le\op{P_xP_y}\le\eta$. The absolute row-sum bound in the same equation therefore gives
\begin{equation}\label{eq:block-signed}
 \begin{aligned}
 |\bra w(W\otimes I_{\mathcal H})\ket w|
 &\le\sum_{x,y}|W_{xy}|\,|\bra vP_xP_y\ket v|\\
 &\le\eta\sum_x\sum_y|W_{xy}|
 \le2^n\ell\eta.
 \end{aligned}
\end{equation}
The products $P_xP_y$ need not be Hermitian; the inequality
$|\bra vT\ket v|\le\left\|\ket v\right\|^2\op T$ holds for every operator $T$.

To compare the two quadratic forms, use the first assumption in~\eqref{eq:general-weights}, not just the overlap and row-sum estimates. Since tensoring with an identity preserves the operator norm,
\begin{equation}\label{eq:block-approximation}
 \begin{aligned}
 |\bra w[(U_0-W)\otimes I_{\mathcal H}]\ket w|
 &\le\op{(U_0-W)\otimes I_{\mathcal H}}\,\left\|\ket w\right\|^2\\
 &=\op{U_0-W}\,\left\|\ket w\right\|^2
 \le\zeta\,2^n\lambda.
 \end{aligned}
\end{equation}
Now write $U_0=W+(U_0-W)$ in~\eqref{eq:block-uniform}. Equations~\eqref{eq:block-signed} and~\eqref{eq:block-approximation}, which use all three assumptions in~\eqref{eq:general-weights}, imply
\begin{equation}\label{eq:quadratic-average}
 2^n\lambda^2\le2^n\ell\eta+2^n\zeta\lambda,
 \qquad \lambda^2-\zeta\lambda-\ell\eta\le0.
\end{equation}
The nonnegative $\lambda$ is at most the larger root of this quadratic. Since $\sqrt{a+b}\le\sqrt a+\sqrt b$ for $a,b\ge0$,
\begin{equation}\label{eq:quadratic-root}
 \lambda\le\frac{\zeta+\sqrt{\zeta^2+4\ell\eta}}2
 \le\zeta+\sqrt{\ell\eta}.
\end{equation}
This proves the lemma, including the case $\lambda=0$; no division by $\lambda$ was used.
\end{proof}

\subsection{The symmetric-output bound}\label{subsec:symmetric-bound}
We combine the weight construction and the averaging lemma to obtain the bound used in the rest of the proof. It holds for every state and decoder described in Section~\ref{subsec:placements}, including states with a nonuniform key marginal.

\begin{proposition}[Symmetric-output privacy bound]\label{prop:symmetric}
In the state-level setting of Section~\ref{subsec:placements}, every integer $t$ satisfying~\eqref{eq:t-range} gives
\begin{equation}\label{eq:symmetric-bound}
 p\le2\exp(-at^2/n)
 +\left(\sum_{j=0}^t\binom nj\right)^{1/4}
 \min\left\{1,\frac{Ld^t}M\right\}^{1/2}.
\end{equation}
\end{proposition}
\begin{proof}
Use the projections~\eqref{eq:placement-projection} on~\eqref{eq:common-hilbert} and the matrix $W_t$ from Lemma~\ref{lem:weights}. Its support excludes $x=y$. On that support, $s(x,y)\le t$, so~\eqref{eq:placement-overlap} gives $\eta\coloneqq\min\{1,Ld^t/M\}$. Set
$\zeta\coloneqq2\exp(-at^2/n)$ and
$\ell\coloneqq(\sum_{j=0}^t\binom nj)^{1/2}$.
All assumptions of Lemma~\ref{lem:averaging} now follow from~\eqref{eq:weight-properties}. Applying that lemma and~\eqref{eq:average-acceptance} proves~\eqref{eq:symmetric-bound}.

\end{proof}

No condition on the key marginal or on the preparation probability of the state was used. This state-uniformity will allow us to apply the bound after a single filter in Section~\ref{sec:reduction}.

\section{Compression by a single filter}\label{sec:filter}
This section constructs a replacement of controlled dimension for one part of a pure state, while allowing recovery using quantum side information. The construction is the step that connects the dimension of the common receiver system in Proposition~\ref{prop:symmetric} to a conditional entropy. It uses one linear filter and Uhlmann's theorem. Figure~\ref{fig:single-filter} illustrates the construction and the reference system that it preserves.

\subsection{Why a filter is sufficient}
We first explain the statement needed for the converse. Suppose a pure state $\psi_{ABR}$ has a system $A$ that is to be replaced by a system $C$ of controlled dimension. A receiver who has $CB$ should be able to reconstruct the original $AB$ state, including its correlations with $R$. For this purpose, it suffices that the filtered state's marginal on $R$ remain close to $\psi_R$: Uhlmann's theorem then provides the recovery on its purifying systems $CB$.

The filter is not required to be a trace-preserving encoder. Proposition~\ref{prop:symmetric} bounds a test on every normalized symmetric state, not only states prepared deterministically by a communication protocol. We may therefore apply a linear operator $Z\colon A\to C$ and normalize its nonzero output. We will prove directly that the resulting state is recoverable. No estimate of the probability of physically implementing the filter is used.

This distinction makes the reduction simpler than employing an approach based on operational state merging~\cite{HOW2007,DBWR2014}. There is no shared entangled seed, forward message, or sum over measurement outcomes. The proof uses the same general principle that a nearly unchanged reference marginal permits recovery, but only the existence of one suitable filtered state is needed. Gaussian constructions have previously been used in quantum coding~\cite{HSW2008}; below we derive the particular second-moment estimate used here in full.

\Needspace{21\baselineskip}
\subsection{The one-shot statement}
The next lemma gives a dimension bound in terms of conditional smooth max-entropy. This is the quantity that will reduce to the coherent information in the channel application. Smooth min- and max-entropies, including the subnormalized smoothing convention, are defined in Appendix~\ref{app:entropies}. We use the Schatten norm conventions of Section~\ref{sec:setup} and set $[x]_+\coloneqq\max\{x,0\}$.

\begin{lemma}[A recoverable single filter]\label{lem:single-filter}
Let $\psi_{ABR}$ be a normalized pure state, and let $0<\delta<1/2$. There are a finite-dimensional system $C$, a linear operator $Z\colon A\to C$, and a recovery channel $\mathcal R\colon CB\to AB$ such that
\begin{equation}\label{eq:filter-state}
 \ket\chi_{CBR}\coloneqq\frac{(Z\otimes I_{BR})\ket\psi_{ABR}}{\sqrt\nu},
 \qquad \nu\coloneqq\bra\psi(Z^\dagger Z\otimes I_{BR})\ket\psi>0,
\end{equation}
satisfies
\begin{equation}\label{eq:filter-recovery}
 P\!\left((\mathcal R\otimes\id_R)(\chi),\psi\right)\le\delta,
\end{equation}
and
\begin{equation}\label{eq:filter-dimension}
 \log|C|\le\left[\hmax^{\delta^2/16}(A|B)_\psi+4\log(1/\delta)+4\right]_++1.
\end{equation}
Smoothing in~\eqref{eq:filter-dimension} is over subnormalized states using purified distance, while~\eqref{eq:filter-recovery} uses sine distance between normalized states. No lower bound on the probability of an implementation of $Z$ as a postselected operation is asserted or required.
\end{lemma}

\begin{figure}
\centering
\begin{tikzpicture}[x=1cm,y=1cm,>=Latex,font=\small,
 wire/.style={draw=black!75,line width=.7pt},
 box/.style={draw=figblue,fill=figblue!4,rounded corners=2pt,line width=.8pt},
 lab/.style={fill=white,inner sep=2pt}]
 \node[box,minimum width=1.05cm,minimum height=2.95cm] at (.6,1.5) {$\ket\psi$};
 \node[box,draw=figamber,fill=figamber!5,minimum width=1.15cm,minimum height=.72cm] (Z) at (3.3,2.5) {$Z$};
 \node[box,minimum width=1.8cm,minimum height=1.9cm,align=center] (rec) at (9.0,2.0) {recovery\\$\mathcal R$};
 \draw[wire,->] (1.125,2.5)--(Z.west) node[midway,lab] {$A$};
 \draw[wire,->] (Z.east)--(8.1,2.5) node[pos=.55,lab] {$C$};
 \draw[wire,->] (1.125,1.5)--(8.1,1.5) node[pos=.2,lab] {$B$};
 \draw[wire,->] (1.125,.5)--(13.4,.5) node[pos=.11,lab] {$R$};
 \draw[wire,->] (9.9,2.5)--(13.4,2.5) node[midway,lab] {$A$};
 \draw[wire,->] (9.9,1.5)--(13.4,1.5) node[midway,lab] {$B$};
 \draw[densely dashed,black!40] (6.0,.05)--(6.0,3.08);
 \node[font=\footnotesize,align=center] at (6.0,3.6) {normalized filtered state\\$\chi_{CBR}$};
 \draw[decorate,decoration={brace,amplitude=4pt}] (13.6,2.65)--(13.6,.35);
 \node[anchor=west,align=center] at (13.85,1.5) {$\widehat\chi$};
 \node[font=\footnotesize,align=center,text=figamber] at (3.3,-.35) {filter and normalize\\by $\sqrt\nu$};
 \node[font=\footnotesize,align=center] at (10.7,-.35)
 {$P(\chi_R,\psi_R)\le\delta$\\$\Longrightarrow\ P(\widehat\chi_{ABR},\psi_{ABR})\le\delta$};
\end{tikzpicture}\par
\caption{The single-filter construction for a pure state $\psi_{ABR}$. A linear operator $Z$ replaces $A$ by $C$, and its nonzero output is normalized to obtain $\chi_{CBR}$. The system $B$ remains available as quantum side information. Controlling the change in the reference marginal, $P(\chi_R,\psi_R)\le\delta$, gives a recovery channel on $CB$ with error at most $\delta$ for the entire state $\psi_{ABR}$. Both $A$ and $B$ are reconstructed. No operation acts on $R$, but its marginal can change upon normalization; the proof controls this change. The filter and normalization are a mathematical construction, not a deterministic coding operation.}
\label{fig:single-filter}
\end{figure}
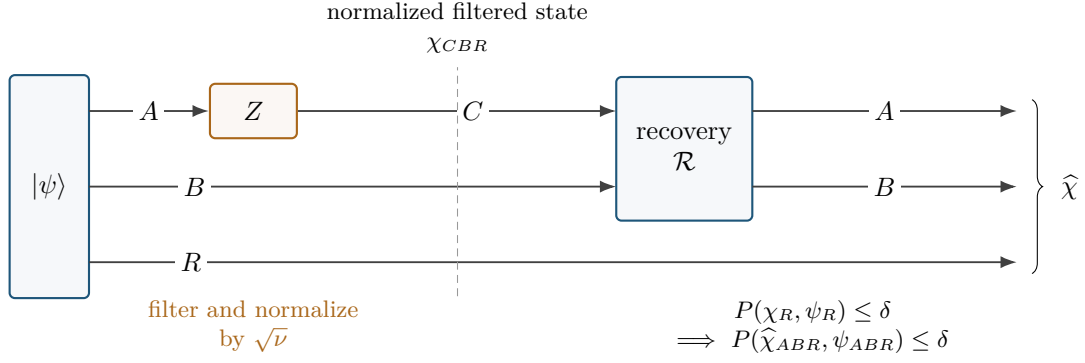

We prove the lemma in the next two subsections. First, a Gaussian calculation bounds how much an unnormalized filter changes the reference marginal. Second, smoothing gives the dimension in~\eqref{eq:filter-dimension}, and normalization and Uhlmann's theorem give recovery.

\subsection{A Gaussian estimate for the reference marginal}\label{subsec:gaussian}
We now establish the elementary random-matrix bound underlying Lemma~\ref{lem:single-filter}. It is convenient to prove it for an arbitrary positive semidefinite operator $\tau_{AR}$ with $\Tr\tau\le1$. Suppose that, for a normalized state $\sigma_R$ and a number $\lambda>0$,
\begin{equation}\label{eq:gaussian-domination}
 \tau_{AR}\le\lambda I_A\otimes\sigma_R.
\end{equation}
For an integer $L\ge1$, set $|C|=L$ and choose the matrix entries of $Z\colon A\to C$ as
\begin{equation}\label{eq:gaussian-entries}
 Z_{ai}\coloneqq\frac{\xi_{ai}+\mathrm i\eta_{ai}}{\sqrt{2L}},
 \qquad 1\le a\le L,\quad 1\le i\le|A|,
\end{equation}
where all $\xi_{ai},\eta_{ai}$ are independent real standard Gaussian random variables. Thus $\mathbb E[Z_{ai}]=0$ and $\mathbb E[|Z_{ai}|^2]=1/L$.

Define a completely positive linear transformation, not necessarily trace preserving, by
\begin{equation}\label{eq:gaussian-transform}
 \Gamma_Z(X_{AR})\coloneqq
 \Tr_C[(Z\otimes I_R)X_{AR}(Z^\dagger\otimes I_R)],
 \qquad \tau_R^Z\coloneqq\Gamma_Z(\tau_{AR}).
\end{equation}
We will prove
\begin{equation}\label{eq:gaussian-bound}
 \mathbb E\,\trn{\tau_R^Z-\tau_R}\le\sqrt{\lambda/L}.
\end{equation}

\emph{Step 1: compute the second moments of the filter.}
Set $Y\coloneqq Z^\dagger Z$. Independence gives $\mathbb E[Y]=I_A$. More explicitly, the fourth moments are
\begin{equation}\label{eq:gaussian-fourth}
 \mathbb E[\overline{Z_{ai}}Z_{aj}\overline{Z_{bk}}Z_{b\ell}]
 =L^{-2}\bigl(\delta_{ij}\delta_{k\ell}+\delta_{ab}\delta_{i\ell}\delta_{jk}\bigr).
\end{equation}
To verify this formula, independence and invariance of each complex Gaussian under multiplication by a phase force conjugated and unconjugated entries to occur in matching pairs. Each pairing contributes $(1/L)^2$. When all four entries coincide, both pairings contribute, consistently with $\mathbb E[|Z_{ai}|^4]=2/L^2$. These observations cover all index choices. Summing~\eqref{eq:gaussian-fourth} over $a,b$ and subtracting the product of the means gives
\begin{equation}\label{eq:wishart-covariance}
 \mathbb E[Y_{ij}]=\delta_{ij},\qquad
 \mathbb E[(Y_{ji}-\delta_{ji})(Y_{\ell k}-\delta_{\ell k})]
 =L^{-1}\delta_{jk}\delta_{i\ell}.
\end{equation}
In particular, $\mathbb E[\Gamma_Z(X)]=\Tr_A X$ for every operator $X_{AR}$.

\emph{Step 2: compute a weighted squared norm.}
Write $\tau_{AR}=\sum_{i,j}\ket i\!\bra j_A\otimes\tau_{ij}$. Then
\begin{equation}\label{eq:gaussian-blocks}
 \tau_R^Z-\tau_R=\sum_{i,j}(Y_{ji}-\delta_{ji})\tau_{ij}.
\end{equation}
All inverse powers of $\sigma_R$ below are taken on its support. This is sufficient:~\eqref{eq:gaussian-domination} implies that $\tau_{AR}$ is supported on $A\otimes\supp\sigma_R$, and hence so are its reference marginals before and after filtering. Set
\begin{equation}\label{eq:omega-weighted}
 \Omega_{AR}\coloneqq(I_A\otimes\sigma_R^{-1/4})\tau_{AR}(I_A\otimes\sigma_R^{-1/4}),
 \qquad \Omega_{ij}\coloneqq\sigma_R^{-1/4}\tau_{ij}\sigma_R^{-1/4}.
\end{equation}
The difference in~\eqref{eq:gaussian-blocks} is Hermitian. Expanding its weighted squared Hilbert--Schmidt norm and applying~\eqref{eq:wishart-covariance} yields
\begin{equation}\label{eq:gaussian-exact-moment}
 \begin{aligned}
 \mathbb E\left\|\sigma_R^{-1/4}(\tau_R^Z-\tau_R)\sigma_R^{-1/4}\right\|_2^2
 &=\sum_{i,j,k,\ell}\frac{\delta_{jk}\delta_{i\ell}}{L}\Tr[\Omega_{ij}\Omega_{k\ell}]\\
 &=\frac1L\sum_{i,j}\Tr[\Omega_{ij}\Omega_{ji}]
 =\frac1L\Tr[\Omega_{AR}^2].
 \end{aligned}
\end{equation}
Equation~\eqref{eq:gaussian-domination} implies $0\le\Omega_{AR}\le\lambda I_A\otimes\sigma_R^{1/2}$. The trace of the product of two positive semidefinite operators is nonnegative, so
\begin{equation}\label{eq:gaussian-weight-bound}
 \begin{aligned}
 \Tr[\Omega_{AR}^2]
 &\le\lambda\Tr[\Omega_{AR}(I_A\otimes\sigma_R^{1/2})]\\
 &=\lambda\Tr[\tau_{AR}]\le\lambda.
 \end{aligned}
\end{equation}
The first inequality uses $\Tr[\Omega(\lambda I_A\otimes\sigma_R^{1/2}-\Omega)]\ge0$; it does not square an operator inequality.

\emph{Step 3: pass to trace norm.}
For an operator $X$ supported on $\supp\sigma_R$, Schatten H\"older gives
\begin{equation}\label{eq:weighted-holder}
 \begin{aligned}
 \trn X
 &=\left\|\sigma_R^{1/4}(\sigma_R^{-1/4}X\sigma_R^{-1/4})\sigma_R^{1/4}\right\|_1\\
 &\le\left\|\sigma_R^{1/4}\right\|_4^2\,
       \left\|\sigma_R^{-1/4}X\sigma_R^{-1/4}\right\|_2
 =\left\|\sigma_R^{-1/4}X\sigma_R^{-1/4}\right\|_2.
 \end{aligned}
\end{equation}
Here $\left\|\sigma_R^{1/4}\right\|_4^2=(\Tr\sigma_R)^{1/2}=1$. Apply~\eqref{eq:weighted-holder} to $X=\tau_R^Z-\tau_R$, and use Cauchy--Schwarz for the expectation, followed by~\eqref{eq:gaussian-exact-moment} and~\eqref{eq:gaussian-weight-bound}. This proves~\eqref{eq:gaussian-bound}. Neither the dimension of $R$ nor the smallest nonzero eigenvalue of $\sigma_R$ appears in this estimate.

\subsection{Smoothing, normalization, and recovery}\label{subsec:filter-recovery}
We now apply the Gaussian estimate to a nearby operator satisfying the smooth-min-entropy bound, but apply the final filter to the original pure state. This distinction will preserve its exact symmetry in the channel application.

\begin{proof}[Proof of Lemma~\ref{lem:single-filter}]
Set $u\coloneqq\delta^2/16$, $h\coloneqq\hmax^u(A|B)_\psi$, and $\rho_{AR}\coloneqq\psi_{AR}$. Smooth entropy duality~\cite{TCR2010} gives $\hmin^u(A|R)_\psi=-h$. By the definitions in Appendix~\ref{app:entropies}, there are $\widetilde\rho_{AR}\in\substates(AR)$ and $\sigma_R\in\states(R)$ satisfying
\begin{equation}\label{eq:filter-smoothing}
 \widetilde\rho_{AR}\le2^hI_A\otimes\sigma_R,
 \qquad P(\widetilde\rho_{AR},\rho_{AR})\le u,
 \qquad \trn{\widetilde\rho_{AR}-\rho_{AR}}\le2u.
\end{equation}
The last inequality follows from the generalized trace-distance bound for purified distance~\cite{TCR2010}: since the center is normalized, $\frac12\trn{\widetilde\rho-\rho}+\frac12(1-\Tr\widetilde\rho)\le u$. Dropping the nonnegative trace-deficit term gives the claimed trace-norm bound. Also, $F(\widetilde\rho,\rho)\ge1-u^2$ and $F(\widetilde\rho,\rho)\le\Tr\widetilde\rho$, so $\Tr\widetilde\rho\ge1-u^2>0$. The compact smoothing ball therefore excludes zero, and the relevant finite-dimensional entropy optimizations attain their extrema. Set
\begin{equation}\label{eq:filter-rank-choice}
 L\coloneqq\left\lceil\max\{1,16\,2^h/\delta^4\}\right\rceil.
\end{equation}
Choose the Gaussian filter of Section~\ref{subsec:gaussian} with this $L$.

To transfer the estimate from $\widetilde\rho$ to $\rho$, let $\Delta$ be an arbitrary Hermitian operator on $AR$ and write $\Delta=\Delta_+-\Delta_-$ for its positive and negative parts. Positivity of $\Gamma_Z$ and the triangle inequality give
\begin{equation}\label{eq:filter-positive-parts}
 \left\|\Gamma_Z(\Delta)\right\|_1
 \le\Tr[\Gamma_Z(\Delta_+)]+\Tr[\Gamma_Z(\Delta_-)]
 =\Tr[\Gamma_Z(|\Delta|)].
\end{equation}
Taking the expectation and using $\mathbb E[Z^\dagger Z]=I_A$ yields
\begin{equation}\label{eq:filter-expectation-contraction}
 \mathbb E\left\|\Gamma_Z(\Delta)\right\|_1\le\left\|\Delta\right\|_1.
\end{equation}
This is an estimate after averaging; an individual Gaussian filter need not contract trace norm.

With $\rho_R^Z\coloneqq\Gamma_Z(\rho_{AR})$ and $\widetilde\rho_R^Z\coloneqq\Gamma_Z(\widetilde\rho_{AR})$, the triangle inequality now gives
\begin{equation}\label{eq:filter-total-error}
 \begin{aligned}
 \mathbb E\left\|\rho_R^Z-\rho_R\right\|_1
 &\le\mathbb E\left\|\rho_R^Z-\widetilde\rho_R^Z\right\|_1
     +\mathbb E\left\|\widetilde\rho_R^Z-\widetilde\rho_R\right\|_1
     +\left\|\widetilde\rho_R-\rho_R\right\|_1\\
 &\le2u+\sqrt{2^h/L}+2u\le\delta^2/2.
 \end{aligned}
\end{equation}
Here the middle term uses~\eqref{eq:gaussian-bound}; the last term uses trace-norm contraction under the partial trace for Hermitian operators. The choices of $u,L$ give $4u=\delta^2/4$ and $\sqrt{2^h/L}\le\delta^2/4$.

There is consequently a fixed realization of $Z$ for which $e\coloneqq\left\|\rho_R^Z-\rho_R\right\|_1\le\delta^2/2$. For this realization set $\nu\coloneqq\Tr\rho_R^Z$. Since $|\nu-1|\le e<1$, the number $\nu$ is positive and~\eqref{eq:filter-state} defines a normalized pure state with marginal $\chi_R=\rho_R^Z/\nu$. Its normalization error is explicit:
\begin{equation}\label{eq:filter-normalization}
 \begin{aligned}
 \left\|\chi_R-\rho_R\right\|_1
 &\le\left\|\rho_R^Z/\nu-\rho_R^Z\right\|_1+\left\|\rho_R^Z-\rho_R\right\|_1\\
 &=|1-\nu|+e\le2e\le\delta^2.
 \end{aligned}
\end{equation}
In particular, no inverse success-probability factor remains in this estimate. For normalized states, the Fuchs--van de Graaf inequalities~\cite{FvdG1999} imply
\begin{equation}\label{eq:filter-sine}
 P(\chi_R,\rho_R)^2=1-F(\chi_R,\rho_R)
 \le\left\|\chi_R-\rho_R\right\|_1\le\delta^2.
\end{equation}
Indeed, if $r=\sqrt{F(\chi_R,\rho_R)}$, then $1-r\le\frac12\left\|\chi_R-\rho_R\right\|_1$ and $1-r^2\le2(1-r)$.

The states $\chi_{CBR}$ and $\psi_{ABR}$ are purifications of the two reference marginals in~\eqref{eq:filter-sine}. Uhlmann's theorem~\cite{Uhlmann1976}, with a sufficiently large auxiliary output $T_{\mathrm{aux}}$, gives an isometry $V_{\mathrm{rec}}\colon CB\to ABT_{\mathrm{aux}}$ such that
\begin{equation}\label{eq:filter-uhlmann}
 F\!\left((V_{\mathrm{rec}}\otimes I_R)\chi(V_{\mathrm{rec}}^\dagger\otimes I_R),
          \psi_{ABR}\otimes\ket0\!\bra0_{T_{\mathrm{aux}}}\right)
 =F(\chi_R,\rho_R).
\end{equation}
Canonical permutations of tensor factors are understood. Define $\mathcal R(X)\coloneqq\Tr_{T_{\mathrm{aux}}}[V_{\mathrm{rec}}XV_{\mathrm{rec}}^\dagger]$. Fidelity is nondecreasing under a partial trace, so~\eqref{eq:filter-sine} and~\eqref{eq:filter-uhlmann} prove~\eqref{eq:filter-recovery}.

Finally, if $b\coloneqq h+4\log(1/\delta)+4$, then~\eqref{eq:filter-rank-choice} gives $L=\lceil\max\{1,2^b\}\rceil\le2\max\{1,2^b\}$. Thus $\log L\le[b]_++1$, proving~\eqref{eq:filter-dimension}.
\end{proof}

The filter $Z$ in the lemma is one fixed operator, not an average over Gaussian choices. Multiplication by $1/\max\{1,\op Z\}$ makes it a contraction and hence a possible measurement outcome, without changing the normalized state~\eqref{eq:filter-state}. Its implementation probability after this rescaling is not estimated. Recovery, not that probability, is what compares its output with the original state.

\begin{remark}[No enlargement is needed]
The bound~\eqref{eq:filter-dimension} need not be strictly smaller than $\log|A|$ for every state and error. If the integer chosen in~\eqref{eq:filter-rank-choice} is at least $|A|$, take $C=A$, $Z=I_A$, and the identity recovery instead. This choice gives zero recovery error and satisfies the same dimension bound. Thus the filter can always be chosen with $|C|\le|A|$, although its useful assertion is the conditional-entropy bound, not a strict dimension reduction in every instance.
\end{remark}

\section{From a degradable code to symmetric outputs}\label{sec:reduction}
We now combine the single-filter lemma with the symmetric dilation of a degradable channel. The dilation identifies an additional receiver system $F$ alongside two interchangeable outputs $E',E$. We replace $F^n$ by a system $C$ of controlled dimension, preserve all exchange symmetries exactly, and use recovery to compare the original privacy test with a test on the filtered state. No communication protocol is added to the original unassisted code.

\subsection{The symmetric factorization}
We first expose the symmetry needed by Proposition~\ref{prop:symmetric}. Figure~\ref{fig:factorization} displays two factorizations of the same isometry. The following statement is the symmetric-dilation construction of Morgan and Winter~\cite{MW2014}, including their harmless maximally mixed-qubit extension. Appendix~\ref{app:dilation} proves it and verifies that the extension preserves the joint error criterion. Figure~\ref{fig:symmetrization} in that appendix explains how an added qubit converts an intermediate exchange symmetry into an exactly symmetric range.

\begin{lemma}[Symmetric factorization]\label{lem:factorization}
Let $\cN\colon A\to B$ be finite dimensional and degradable. After appending an independent maximally mixed qubit to the receiver output, if necessary, there are isometries
\begin{equation}\label{eq:factorization-maps}
 U\colon A\to BE,\qquad V\colon B\to FE',\qquad
 J\colon A\to FG,\qquad W\colon G\to E'E,
\end{equation}
with $E'\cong E$, such that
\begin{equation}\label{eq:factorization}
 (V\otimes I_E)U=(I_F\otimes W)J,\qquad \FS_{E'E}W=W.
\end{equation}
The isometry $U$ dilates the channel, $V$ dilates a degrading channel, and $W$ can be chosen as the inclusion of $G=\Sym^2(E)$ into $E'E$. The appended qubit leaves $\qN$ unchanged. Every original private code has a corresponding code with exactly the same optimized joint fidelity and the same actual-marginal joint fidelity.
\end{lemma}

\begin{figure}
\centering
\begin{tikzpicture}[x=1cm,y=1cm,>=Latex,font=\small,
 wire/.style={draw=black!75,line width=.7pt},
 box/.style={draw=figblue,fill=figblue!4,rounded corners=2pt,line width=.8pt},
 lab/.style={fill=white,inner sep=1.7pt}]
 \node[anchor=west,font=\small\bfseries] at (0,5.35)
 {(a) Channel dilation followed by a degrading isometry};
 \node[box,minimum width=1.15cm,minimum height=1.25cm] (U) at (2.1,3.9) {$U$};
 \node[box,minimum width=1.15cm,minimum height=1.25cm] (V) at (5.1,4.2) {$V$};
 \draw[wire,->] (.15,3.9)--(U.west) node[midway,lab] {$A$};
 \draw[wire,->] (2.675,4.2)--(4.525,4.2) node[midway,lab] {$B$};
 \draw[wire,->] (5.675,4.5)--(10.05,4.5) node[pos=.48,lab] {$F$};
 \draw[wire,->] (5.675,3.9)--(10.05,3.9) node[pos=.48,lab] {$E'$};
 \draw[wire,->] (2.675,3.6)--(3.2,3.6)--(3.2,2.7)--(10.05,2.7)
 node[pos=.68,lab] {$E$};
 \draw[decorate,decoration={brace,amplitude=4pt}] (10.2,4.58)--(10.2,3.82);
 \node[anchor=west] at (10.55,4.2) {receiver};
 \node[anchor=west] at (10.55,2.7) {environment};
 \draw[black!25] (0,2.02)--(13.0,2.02);
 \node[anchor=west,font=\small\bfseries] at (0,1.48)
 {(b) Same isometry, with an explicitly symmetric output pair};
 \node[box,minimum width=1.15cm,minimum height=1.25cm] (J) at (2.1,-.05) {$J$};
 \node[box,draw=figteal,fill=figteal!4,minimum width=1.15cm,minimum height=1.25cm]
 (W) at (5.1,-.9) {$W$};
 \draw[wire,->] (.15,-.05)--(J.west) node[midway,lab] {$A$};
 \draw[wire,->] (2.675,.25)--(10.05,.25) node[pos=.56,lab] {$F$};
 \draw[wire,->] (2.675,-.35)--(3.4,-.35)--(3.4,-.9)--(4.525,-.9)
 node[pos=.81,lab] {$G$};
 \draw[wire,->] (5.675,-.6)--(10.05,-.6) node[pos=.48,lab] {$E'$};
 \draw[wire,->] (5.675,-1.2)--(10.05,-1.2) node[pos=.48,lab] {$E$};
 \draw[decorate,decoration={brace,amplitude=4pt}] (10.2,.33)--(10.2,-.68);
 \node[anchor=west] at (10.55,-.18) {receiver};
 \node[anchor=west] at (10.55,-1.2) {environment};
 \draw[<->,figteal] (8.65,-.72)--(8.65,-1.08);
 \node[text=figteal,font=\footnotesize] at (6.9,-1.78)
 {interchangeable outputs: $\mathbb F_{E'E}W=W$};
 \node at (6.5,-2.48) {$(V\otimes I_E)U=(I_F\otimes W)J$};
\end{tikzpicture}\par
\caption{The two factorizations in Lemma~\ref{lem:factorization}. Panel (a) first dilates the original channel and then the degrading channel, producing $F,E',E$. Panel (b) represents the same isometry as $J$ followed by the symmetric inclusion $W$. The output state is identical for every input, including inputs entangled with a reference. The system $F$ is unchanged by exchange of $E'$ and $E$ and will occur in every alternative receiver selection until it is replaced by $C$. No ownership of $F$ is transferred to the sender. The lower factorization is used to exhibit exact exchange symmetry, not to introduce an auxiliary communication protocol.}
\label{fig:factorization}
\end{figure}

We henceforth use $B,E$ for the possibly enlarged systems in Lemma~\ref{lem:factorization}, and set
\begin{equation}\label{eq:channel-dimensions}
 a_A\coloneqq|A|,\qquad d\coloneqq|E|,\qquad d_F\coloneqq|F|.
\end{equation}
These dimensions depend only on the channel factorization. Define the channel to the degraded receiver systems by
\begin{equation}\label{eq:T-channel}
 \cT(\rho)\coloneqq V\cN(\rho)V^\dagger.
\end{equation}
Isometric invariance of entropy and degradability give
\begin{equation}\label{eq:conditional-coherent}
 \begin{aligned}
 H(F|E')_{\cT(\rho)}
 &=H(FE')_{\cT(\rho)}-H(E')_{\cT(\rho)}\\
 &=H(\cN(\rho))-H(\cN^c(\rho))\le\qN.
 \end{aligned}
\end{equation}
This conditional entropy is the asymptotic value of the logarithmic dimension bound in Lemma~\ref{lem:single-filter}.

For the coherent encoder $\ket\phi_{RA^n}$ of Lemma~\ref{lem:lift}, set $R\coloneqq KS$. Define
\begin{equation}\label{eq:symmetric-output-state}
 \ket\psi_{RF^nE'^nE^n}\coloneqq
 (I_R\otimes[(V\otimes I_E)U]^{\otimes n})\ket\phi
 =(I_{RF^n}\otimes W^{\otimes n})(I_R\otimes J^{\otimes n})\ket\phi.
\end{equation}
This is the original coherent channel output after $V^{\otimes n}$, before decoding. For every $j$, exchange of $E'_j$ and $E_j$ leaves this ket unchanged, since $\FS_{E'E}W=W$.

\subsection{Filtering while preserving the exact symmetry}\label{subsec:filter-symmetry}
We next replace the common system $F^n$ without breaking the symmetry. Applying Proposition~\ref{prop:symmetric} directly with $C=F^n$ would introduce the factor $d_F^n$ and generally give a threshold involving $\log d_F$, rather than $\qN$. The single-filter lemma replaces this dimension by a quantity governed by $\hmax^u(F^n|E'^n)_\psi$. Figure~\ref{fig:filter-reduction} shows the filter and recovery for the coherent code.

Apply Lemma~\ref{lem:single-filter} with its systems $(A,B,R)$ identified as $(F^n,E'^n,KSE^n)$. The reference thus includes the sender key, the sender shield, and the original environment. We obtain $Z\colon F^n\to C$ and a normalized pure state
\begin{equation}\label{eq:filtered-code-state}
 \ket\chi_{KSC E'^nE^n}\coloneqq
 \frac{(I_{KS}\otimes Z\otimes I_{E'^nE^n})\ket\psi}{\sqrt\nu},
 \qquad \nu\coloneqq\left\|(I_{KS}\otimes Z\otimes I_{E'^nE^n})\ket\psi\right\|^2>0.
\end{equation}
The dimension $L\coloneqq|C|$ obeys~\eqref{eq:filter-dimension}, and a recovery channel $\mathcal R_{\mathrm{rec}}\colon CE'^n\to F^nE'^n$ satisfies
\begin{equation}\label{eq:code-filter-recovery}
 P\!\left((\id_{KS E^n}\otimes\mathcal R_{\mathrm{rec}})(\chi),\psi\right)\le\delta,
\end{equation}
with tensor factors ordered as in~\eqref{eq:symmetric-output-state}.

Let $\mathsf F_j$ denote the swap on $E'_jE_j$, extended by identities. It commutes with the filter on $F^n$, so
\begin{equation}\label{eq:filtered-symmetry}
 \mathsf F_j\ket\chi
 =\frac{(I_{KS}\otimes Z\otimes I_{E'^nE^n})\mathsf F_j\ket\psi}{\sqrt\nu}
 =\ket\chi.
\end{equation}
This symmetry is exact. The smoothing optimizer used in the proof of Lemma~\ref{lem:single-filter} need not be symmetric: it is used only to prove existence of a good filter. The actual filter in~\eqref{eq:filtered-code-state} acts on the original state $\psi$.

In the notation of Section~\ref{sec:symmetric}, identify
\begin{equation}\label{eq:output-dictionary}
 B_{j,0}\equiv E'_j,\qquad B_{j,1}\equiv E_j,\qquad B_{0^n}\equiv E'^n.
\end{equation}
The decoder on the constructed state will use $CE'^n$. Other bit strings select alternative tensor factors for mathematical tests. They do not give the physical receiver access to $E^n$. The filtered state need not have a uniform key marginal, and need not be the output of a deterministic private code. Neither property is assumed in Proposition~\ref{prop:symmetric}.

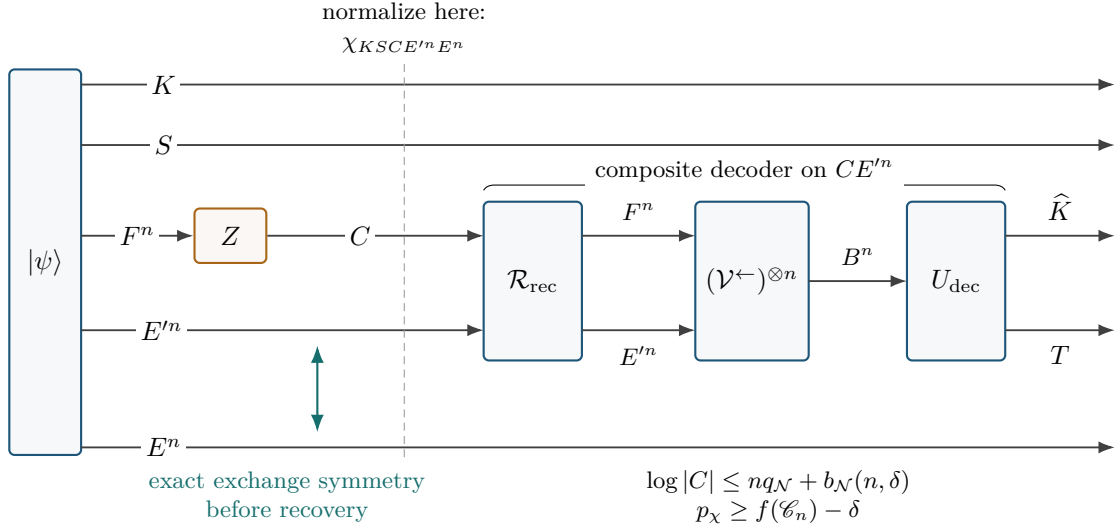
\begin{figure}
\centering
\begin{tikzpicture}[x=1cm,y=1cm,>=Latex,font=\small,
 wire/.style={draw=black!75,line width=.7pt},
 box/.style={draw=figblue,fill=figblue!4,rounded corners=2pt,line width=.8pt},
 lab/.style={fill=white,inner sep=1.6pt}]
 \node[box,minimum width=.95cm,minimum height=5.1cm] at (.55,2.45) {$\ket\psi$};
 \node[box,draw=figamber,fill=figamber!5,minimum width=.95cm,minimum height=.72cm] (Z) at (3.0,2.8) {$Z$};
 \node[box,minimum width=1.3cm,minimum height=2.1cm,align=center] (rec) at (7.0,2.2) {$\mathcal R_{\mathrm{rec}}$};
 \node[box,minimum width=1.5cm,minimum height=2.1cm,align=center] (inv) at (9.9,2.2) {$(\mathcal V^{\leftarrow})^{\otimes n}$};
 \node[box,minimum width=1.3cm,minimum height=2.1cm,align=center] (dec) at (12.6,2.2) {$U_{\mathrm{dec}}$};
 \draw[wire,->] (1.025,4.8)--(14.7,4.8) node[pos=.08,lab] {$K$};
 \draw[wire,->] (1.025,4.0)--(14.7,4.0) node[pos=.08,lab] {$S$};
 \draw[wire,->] (1.025,2.8)--(Z.west) node[midway,lab] {$F^n$};
 \draw[wire,->] (Z.east)--(6.35,2.8) node[pos=.43,lab] {$C$};
 \draw[wire,->] (1.025,1.55)--(6.35,1.55) node[pos=.2,lab] {$E'^n$};
 \draw[wire,->] (1.025,0)--(14.7,0) node[pos=.08,lab] {$E^n$};
 \draw[wire,->] (7.65,2.8)--(9.15,2.8) node[midway,above=2pt,font=\footnotesize] {$F^n$};
 \draw[wire,->] (7.65,1.55)--(9.15,1.55) node[midway,below=2pt,font=\footnotesize] {$E'^n$};
 \draw[wire,->] (10.65,2.2)--(11.95,2.2) node[midway,above=2pt,font=\footnotesize] {$B^n$};
 \draw[wire,->] (13.25,2.8)--(14.7,2.8) node[midway,above=2pt] {$\widehat K$};
 \draw[wire,->] (13.25,1.55)--(14.7,1.55) node[midway,below=2pt] {$T$};
 \draw[<->,figteal,line width=.8pt] (4.15,.2)--(4.15,1.35);
 \node[text=figteal,font=\footnotesize,align=center] at (3.75,-.65) {exact exchange symmetry\\before recovery};
 \draw[densely dashed,black!40] (5.3,-.13)--(5.3,5.08);
 \node[fill=white,font=\footnotesize,align=center] at (5.3,5.55) {normalize here:\\$\chi_{KSC E'^nE^n}$};
 \draw[decorate,decoration={brace,amplitude=4pt}] (6.35,3.4)--(13.25,3.4);
 \node[font=\footnotesize,fill=white] at (9.8,3.67) {composite decoder on $CE'^n$};
 \node[font=\footnotesize,align=center] at (10.25,-.65)
 {$\log|C|\le n\qN+b_{\cN}(n,\delta)$\\$p_\chi\ge f(\mathscr C_n)-\delta$};
\end{tikzpicture}\par
\caption{The filter reduction for a degradable code. Starting from $\psi_{KSF^nE'^nE^n}$, apply one fixed filter $Z$ only to $F^n$ and normalize to obtain $\chi$. Before recovery, the state retains every exchange symmetry $E'_j\leftrightarrow E_j$. The composite decoder acts on $CE'^n$: it first recovers $F^nE'^n$, then reverses the degrading isometry on its range, and finally applies the original coherent decoder. The systems $K,S,E^n$ are references for recovery. The final state on the original systems is within sine distance~$\delta$ of the original coherent output. The filter is not a shared entangled resource or an additional allowed encoder; it constructs a state to which the state-level test bound applies.}
\label{fig:filter-reduction}
\end{figure}

\subsection{A uniform dimension estimate for correlated inputs}
The original encoder may correlate all channel uses. We therefore need a bound on $\hmax^u(F^n|E'^n)_\psi$ uniform over its input, including when $u$ is exponentially small. The next proposition combines the postselection estimate of Morgan and Winter~\cite{MW2014} with a dimension-uniform finite-block asymptotic equipartition estimate~\cite{Tomamichel2012}.

\begin{proposition}[Uniform filter dimension]\label{prop:overhead}
Fix the factorization in Lemma~\ref{lem:factorization}. For $n\ge1$ and $0<\delta<1/2$, define
\begin{equation}\label{eq:overhead-parameters}
 g_n\coloneqq(n+1)^{a_A^2},\qquad c_F\coloneqq8\log(2d_F+1),\qquad
 u\coloneqq\frac{\delta^2}{16},\qquad v\coloneqq\frac{u^2}{32g_n},
\end{equation}
and set
\begin{equation}\label{eq:b-definition}
 b_{\cN}(n,\delta)\coloneqq c_F\sqrt{n\log(2/v^2)}+3\log g_n+8\log(1/\delta)+19.
\end{equation}
For every coherent encoder, the filter and recovery in~\eqref{eq:filtered-code-state}--\eqref{eq:code-filter-recovery} can be chosen with
\begin{equation}\label{eq:L-overhead}
 \log L\le n\qN+b_{\cN}(n,\delta).
\end{equation}
The estimate is uniform over all correlated input states on $A^n$.
\end{proposition}
\begin{proof}
Appendix~\ref{app:entropies} proves the following estimates. For an arbitrary input $\rho_{A^n}$ and $0<u<1$,
\begin{equation}\label{eq:postselection-main}
 \begin{aligned}
 \hmax^u(F^n|E'^n)_{\cT^{\otimes n}(\rho)}
 \le{}&\max_{\sigma\in\states(A)}
 \hmax^{u^2/(32g_n)}(F^n|E'^n)_{\cT(\sigma)^{\otimes n}}\\
 &+3\log g_n+6+2\log(1/u).
 \end{aligned}
\end{equation}
For every state $\theta_{FE'}$, every $n\ge1$, and every $0<v<1$,
\begin{equation}\label{eq:aep-main}
 \hmax^v(F^n|E'^n)_{\theta^{\otimes n}}
 \le nH(F|E')_\theta+c_F\sqrt{n\log(2/v^2)}.
\end{equation}
The constant depends only on $d_F$, not on the spectrum of $\theta$. Permutation averaging in~\eqref{eq:postselection-main} is an entropy-estimation device; it does not change the original code or the state to which the filter is applied.

Combining~\eqref{eq:conditional-coherent}, \eqref{eq:postselection-main}, and~\eqref{eq:aep-main} yields
\begin{equation}\label{eq:hmax-combined}
 \hmax^u(F^n|E'^n)_\psi
 \le n\qN+c_F\sqrt{n\log(2/v^2)}+3\log g_n+6+2\log(1/u).
\end{equation}
In Lemma~\ref{lem:single-filter}, use $u=\delta^2/16$. Then
\begin{equation}\label{eq:smoothing-algebra}
 2\log(1/u)=4\log(1/\delta)+8.
\end{equation}
Adding the filter term $4\log(1/\delta)+4$ and the rounding term $1$ gives the constant $6+8+4+1=19$ in~\eqref{eq:b-definition}. Since $\qN\ge0$ and the remaining upper-bound terms are positive, the positive part in~\eqref{eq:filter-dimension} is bounded by this same expression. This proves~\eqref{eq:L-overhead}.
\end{proof}

The filter and recovery may depend on the particular code and on $\delta$. Uniformity means that the same dimension bound holds for every code; it does not assert one filter that works for all inputs. This is sufficient for a converse, which bounds each code separately.

The dependence on exponentially small error is explicit:
\begin{equation}\label{eq:logv}
 \log(2/v^2)=8\log(1/\delta)+2\log g_n+27.
\end{equation}
Consequently, for each fixed $c>0$,
\begin{equation}\label{eq:overhead-exponent}
 \limsup_{n\to\infty}\frac{b_{\cN}(n,2^{-cn})}{n}\le c_F\sqrt{8c}+8c.
\end{equation}
For fixed positive $c$, the correction is generally linear in $n$, not $o(n)$. Its linear coefficient can be made arbitrarily small by choosing $c$ small. This suffices for the exponential strong converse at every fixed positive rate gap.

\subsection{Returning to the original privacy test}
We now compose recovery with the original decoder and compare test probabilities. This comparison is the reason that a filtered state is sufficient: it relates the original code directly to a normalized symmetric state, without interpreting the filter as a successful run of a communication protocol.

Let $\mathcal V^{\leftarrow}\colon FE'\to B$ be an inverse of the degrading isometry on its range, extended to a channel by
\begin{equation}\label{eq:inverse-channel}
 \mathcal V^{\leftarrow}(X)\coloneqq V^\dagger XV+\Tr[(I_{FE'}-VV^\dagger)X]\,\tau_B,
\end{equation}
where $\tau_B$ is an arbitrary fixed state. This channel is completely positive and trace preserving, and $\mathcal V^{\leftarrow}(V\rho V^\dagger)=\rho$. Apply the composite receiver channel
\begin{equation}\label{eq:composite-filter-decoder}
 \mathcal D_{\chi}\coloneqq\mathcal D_{\mathrm{coh}}\circ
 (\mathcal V^{\leftarrow})^{\otimes n}\circ\mathcal R_{\mathrm{rec}},
\end{equation}
where $\mathcal D_{\mathrm{coh}}(X)\coloneqq U_{\mathrm{dec}}XU_{\mathrm{dec}}^\dagger$ uses the original coherent decoder~\eqref{eq:coherent-decoder}. Denote the resulting state on the original key and shield systems by $\widehat\chi_{KS\widehat K T}$. Equation~\eqref{eq:code-filter-recovery} and data processing imply
\begin{equation}\label{eq:original-output-recovery}
 P(\widehat\chi_{KS\widehat K T},\omega_{KS\widehat K T})\le\delta.
\end{equation}
Retaining local dilations of the composite channel simply enlarges the receiver shield. Extend the original twisting blocks by identities on these additional systems, so the privacy-test probability remains unchanged.

Let $p_\chi$ be the acceptance probability of this test on the recovered filtered state. Lemma~\ref{lem:lift} and the test-probability estimate~\eqref{eq:test-probability} give
\begin{equation}\label{eq:aux-acceptance}
 p_\chi=\Tr[\Pi^\gamma\widehat\chi]\ge\Tr[\Pi^\gamma\omega]-\delta
 \ge f(\mathscr C_n)-\delta.
\end{equation}
Before the composite decoder, $\chi$ has exactly the symmetry in~\eqref{eq:filtered-symmetry}. Proposition~\ref{prop:symmetric} therefore applies with common receiver system $C$, using~\eqref{eq:output-dictionary}. The filter has already been fixed, so every test placement acts on this one state $\chi$. The alternative tests are relabelings of the same composite decoder on $CB_x$; exact exchange symmetry makes their acceptance probabilities equal. No filter is chosen anew for an alternative placement.

The logical implication is
\begin{equation}\label{eq:filter-logic}
 f(\mathscr C_n)\le p_\chi+\delta
 \le\bigl[\text{bound for every normalized symmetric state}\bigr]+\delta.
\end{equation}
The first inequality is a consequence of recovery, and the second is state-uniform. Neither multiplies nor divides by a probability of realizing $Z$. This distinguishes the argument from an unjustified replacement of an unconditional code fidelity by a postselected success probability.

\section{The finite-block converse and exponential decay}\label{sec:converse}
We can now state a finite-block inequality from which the main theorem follows by explicit parameter choices. Keeping this inequality separate also identifies the three contributions: recovery error, approximation of the uniform averaging matrix, and the common-receiver overlap.

\begin{theorem}[Finite-block private converse]\label{thm:finite}
Let $\cN\colon A\to B$ be finite dimensional and degradable. Fix a symmetric factorization as in Lemma~\ref{lem:factorization}, let $d=|E|$, and define $b_{\cN}$ by~\eqref{eq:overhead-parameters}--\eqref{eq:b-definition}. For every unassisted $(n,M)$ code, every $0<\delta<1/2$, and every integer $t$ satisfying~\eqref{eq:t-range},
\begin{equation}\label{eq:finite-block}
 f(\mathscr C_n)\le\delta+2e^{-at^2/n}
 +\left(\sum_{j=0}^t\binom nj\right)^{1/4}
 \min\left\{1,\frac{2^{n\qN+b_{\cN}(n,\delta)}d^t}{M}\right\}^{1/2}.
\end{equation}
All constants in this inequality depend only on the fixed channel factorization and the universal approximation constants.
\end{theorem}
\begin{proof}
Apply Lemma~\ref{lem:lift} to the original code. Use Lemma~\ref{lem:factorization} and Proposition~\ref{prop:overhead} to construct the normalized filtered state~\eqref{eq:filtered-code-state}. Its common receiver system has dimension $L\le2^{n\qN+b_{\cN}(n,\delta)}$. As explained after~\eqref{eq:aux-acceptance}, Proposition~\ref{prop:symmetric} bounds its test acceptance by~\eqref{eq:symmetric-bound}. Combining that bound with~\eqref{eq:aux-acceptance}, and using monotonicity of the right-hand side of~\eqref{eq:symmetric-bound} in $L$, gives~\eqref{eq:finite-block}.
\end{proof}

\begin{proof}[Proof of Theorem~\ref{thm:main}]
Fix $\Delta>0$. Choose $c>0$ sufficiently small that
\begin{equation}\label{eq:choose-c}
 c_F\sqrt{8c}+8c<\Delta/4.
\end{equation}
Set $\delta_n\coloneqq2^{-cn}$. By~\eqref{eq:overhead-exponent}, there is a channel- and gap-dependent threshold after which
\begin{equation}\label{eq:small-overhead}
 b_{\cN}(n,\delta_n)\le n\Delta/4.
\end{equation}
Next, choose $0<\tau<1/2$ such that
\begin{equation}\label{eq:choose-tau}
 \tau\log d+\tfrac12\hbin(\tau)<\Delta/4,
 \qquad
 \hbin(\tau)\coloneqq-\tau\log\tau-(1-\tau)\log(1-\tau).
\end{equation}
Such a choice exists because both terms on the left tend to zero as $\tau\downarrow0$. Set $t_n\coloneqq\lfloor\tau n\rfloor$. For sufficiently large $n$, this integer obeys~\eqref{eq:t-range}, and $t_n\ge\tau n/2$.

We use the binomial estimate
\begin{equation}\label{eq:binomial}
 \sum_{j=0}^{\lfloor\tau n\rfloor}\binom nj\le2^{n\hbin(\tau)}.
\end{equation}
For completeness, in the binomial expansion of $1=(\tau+1-\tau)^n$, each term with $j\le\tau n$ has $\tau^j(1-\tau)^{n-j}\ge2^{-n\hbin(\tau)}$, because $\tau/(1-\tau)<1$. Summing those terms proves~\eqref{eq:binomial}.

Assume $\log M\ge n(\qN+\Delta)$. Dropping the minimum in~\eqref{eq:finite-block} in favor of its second argument only increases the bound. Using~\eqref{eq:small-overhead}--\eqref{eq:binomial}, the final term in~\eqref{eq:finite-block} is at most
\begin{equation}\label{eq:third-term}
 2^{n\hbin(\tau)/4}
 2^{[-n\Delta+b_{\cN}(n,\delta_n)+t_n\log d]/2}
 \le2^{n[-3\Delta/8+\tau\log d/2+\hbin(\tau)/4]}
 \le2^{-n\Delta/4}.
\end{equation}
Consequently,
\begin{equation}\label{eq:three-exponentials}
 f(\mathscr C_n)\le2^{-cn}+2\exp(-a\tau^2n/4)+2^{-n\Delta/4}
\end{equation}
for all sufficiently large $n$, uniformly over codes. Choose
\begin{equation}\label{eq:gamma-choice}
 0<\gamma<\min\left\{c,\frac{a\tau^2}{4\ln2},\frac\Delta4\right\}.
\end{equation}
Increasing the threshold $n_0$ absorbs the constant prefactor in~\eqref{eq:three-exponentials}, proving~\eqref{eq:main}. Equation~\eqref{eq:error-comparison} gives the same conclusion for maximal joint infidelity. The actual marginal is a feasible choice in~\eqref{eq:ideal}, so $f_{\mathrm{marg}}\le f$ gives the actual-marginal conclusion with the same exponent. Finally, the known capacity formula $P(\cN)=\qN$ identifies this threshold with the unassisted private capacity.
\end{proof}

\subsection{Antidegradable channels}
No filtering is needed for antidegradable channels. Their outputs have a symmetric extension, and the state-level form of Proposition~\ref{prop:symmetric} is sufficient even if this extension is mixed.

\begin{corollary}[Antidegradable channels]\label{cor:anti}
For every finite-dimensional antidegradable channel $\cN\colon A\to B$ and every $R>0$, there are $\gamma>0$ and $n_0$ such that every unassisted $(n,M)$ code with $n\ge n_0$ and $\log M\ge nR$ satisfies $f(\mathscr C_n)\le2^{-\gamma n}$. Thus the unassisted private capacity has an exponential strong converse at zero under the joint criteria~\eqref{eq:error}, \eqref{eq:max-error}, and~\eqref{eq:marginal-definition-main}.
\end{corollary}
\begin{proof}
Let $\cA\colon E\to B_1$ be an antidegrading channel. Apply $\cA$ to the environment of a Stinespring dilation of $\cN$, retaining the original output $B_0$. Both marginal channels of this extension equal $\cN$, including on inputs entangled with a reference. Average the extension with its output-swapped version to make it swap invariant. Its $n$-fold tensor power is invariant under independent swaps of each pair.

Use the coherent encoder and privacy test from Lemma~\ref{lem:lift}. The marginal on the sender's systems and $B_0^n$ is unchanged by replacing the original channel with this symmetric extension. Proposition~\ref{prop:symmetric} therefore applies with $d=|B|$ and $L=1$, and gives
\begin{equation}\label{eq:anti-finite}
 f(\mathscr C_n)\le2\exp(-at^2/n)
 +\left(\sum_{j=0}^t\binom nj\right)^{1/4}
 \min\left\{1,\frac{d^t}M\right\}^{1/2}.
\end{equation}
Choose $t=\lfloor\tau n\rfloor$, with $\tau>0$ small enough that $\tau\log d+\hbin(\tau)/2<R/2$. The argument used in~\eqref{eq:third-term} then makes the second term exponentially small; the first term is also exponentially small. This proves the assertion and the zero-capacity converse directly.
\end{proof}

The criteria in Corollary~\ref{cor:anti} are joint infidelities. Prior results with separate decoding and secrecy parameters, including Ref.~\cite{BKH2025}, concern a different finite-error formulation. Appendix~\ref{app:marginal} gives the precise implication for that formulation: the sum of the decoding error and the trace-distance secrecy error is at least $1-2^{-\gamma n/2}$. The coherent lift permits the optimized joint fidelity itself to pass to the test without a loss.

\subsection{The quantum erasure channel}
The erasure channel provides a useful example in which the two cases cover the full parameter range. For $p\in[0,1]$ and input dimension $d_0$, let
\begin{equation}\label{eq:erasure}
 \mathcal E_{p,d_0}(\rho)\coloneqq(1-p)\rho\oplus p\ket e\!\bra e,
\end{equation}
where the erasure flag is orthogonal to the input space. Its complement is equivalent to $\mathcal E_{1-p,d_0}$. The erasure-channel capacities were studied in Ref.~\cite{BDS1997}. In the terminology of Ref.~\cite{DS2005}, the channel is degradable for $p\le1/2$ and antidegradable for $p\ge1/2$. For an arbitrary input,
\begin{equation}\label{eq:erasure-entropies}
 \begin{aligned}
 H(\mathcal E_{p,d_0}(\rho))&=\hbin(p)+(1-p)H(\rho),\\
 H(\mathcal E_{p,d_0}^{\,c}(\rho))&=\hbin(p)+pH(\rho),\\
 I_{\mathrm c}(\rho,\mathcal E_{p,d_0})&=(1-2p)H(\rho).
 \end{aligned}
\end{equation}
The capacity formula and the preceding results therefore give the following consequence.

\begin{corollary}[All erasure probabilities]\label{cor:erasure}
For every $p\in[0,1]$ and finite $d_0$, the unassisted private capacity of the erasure channel is
\begin{equation}\label{eq:erasure-capacity}
 P(\mathcal E_{p,d_0})=[1-2p]_+\log d_0,
\end{equation}
and it has an exponential strong converse under the joint criteria~\eqref{eq:error}, \eqref{eq:max-error}, and~\eqref{eq:marginal-definition-main}. The notation $[x]_+$ denotes $\max\{x,0\}$.
\end{corollary}

\section{Conclusion}\label{sec:conclusion}
\subsection{Summary}
We have proved an exponential strong converse for unassisted private communication over every finite-dimensional degradable channel, using a single joint infidelity criterion. Above the maximum coherent information, the fidelity of every code decays exponentially, uniformly over correlated input states and collective decoders. The conclusion applies both when the ideal environment state is optimized and when it is fixed to the actual marginal. The same method gives an exponential strong converse at zero for antidegradable channels and hence covers the quantum erasure channel throughout its parameter range.

The private and quantum coding problems have the same asymptotic capacity for a degradable channel, but their finite-error converses require different tests. Lemma~\ref{lem:overlap} supplies the private-state analogue of the entanglement-test estimate. The ordered products of the twisting unitaries have norm one even on a common sender shield. Only the dimension of the common receiver factor enters the overlap bound.

The single-filter lemma connects the common receiver dimension to the coherent-information threshold. A Gaussian moment estimate produces a filter that nearly preserves the reference marginal. Normalization is controlled explicitly, and Uhlmann's theorem then recovers the original state from the output of the filter together with quantum side information. In the channel application, the filter preserves exact exchange symmetry and recovery preserves the original privacy-test probability up to a small additive error. The correlated-input entropy estimate remains uniform for exponentially small error. No state-merging protocol or auxiliary shared entanglement is used. The signed Boolean construction is due to Ref.~\cite{KBK2026}, based on Ref.~\cite{deWolf2008}; our extension uses privacy tests and the recoverable single-filter reduction.

\subsection{Future directions}
The proof raises several questions concerning the quantitative bounds and the scope of the result.
\begin{enumerate}[label=(\roman*),leftmargin=*,itemsep=5pt]
\item \emph{Strong-converse exponents.} The exponent obtained in~\eqref{eq:gamma-choice} is not asserted to be optimal. It would be useful to sharpen the averaging and filter estimates, and to determine whether a variational expression characterizes the optimal exponent for degradable channels. Even special families such as erasure channels provide concrete cases in which to compare the present bound with tighter rate-dependent estimates.

\item \emph{Bosonic private communication.} Ref.~\cite{Wilde2026PureLoss} proves a strong converse for the unconstrained quantum capacity of the pure-loss bosonic channel, with an inverse-blocklength bound on entanglement-generation fidelity. A natural next question is whether its unassisted private capacity satisfies a strong converse under the joint fidelity criteria studied here. One could also consider appropriately specified photon-number constraints. The finite-dimensional theorem does not settle these questions by a formal cutoff: both the common-output factor $d^t$ and the uniform entropy constants depend on dimension, and the truncation error must be controlled for the entire class of allowed inputs. A proof adapted to the bosonic channel, possibly using privacy tests in place of entanglement tests in Ref.~\cite{Wilde2026PureLoss}, would have to address these issues directly.

\item \emph{Channels beyond degradability.} The overlap lemma is independent of degradability, and Proposition~\ref{prop:symmetric} is a statement about states. It is therefore natural to ask which other channels admit useful approximate symmetric reductions with a controlled common receiver dimension. Such a result would need quantitative control of the approximation for arbitrary correlated code inputs, not merely for single-use states.

\item \emph{Other assistance models.} Strong converses with public classical communication assistance are already known for several channels. Ref.~\cite[Theorem~12, Propositions~18 and~22, and Corollary~25]{WTB2017} established strong-converse bounds for two-way public-assisted private communication using the relative entropy of entanglement. Its general bound for teleportation-simulable channels need not be tight for every channel. For erasure and qubit dephasing channels, and for unconstrained pure-loss and quantum-limited amplifier channels, the bounds coincide with the corresponding assisted private capacities. A remaining question is whether the present method gives strong converses at the capacity for further channels or assistance models. This requires bounds appropriate to the resources and communication allowed in the task. The normalized filtered state used here is a mathematical reduction, not an interactive protocol; the present proof does not establish a converse for arbitrary feedback-assisted communication.
\end{enumerate}

\section*{Acknowledgements}
\addcontentsline{toc}{section}{Acknowledgements}
I acknowledge support from the Cornell University School of Electrical and Computer Engineering.

\section*{Statement on AI-assisted preparation}
\addcontentsline{toc}{section}{Statement on AI-assisted preparation}
ChatGPT (OpenAI) Pro 6 (Astra) was used extensively in developing this paper, including for mathematical exploration, formulation and checking of proof arguments, drafting and revision, bibliographic verification, and preparation of LaTeX, TikZ, and numerical-checking code. The author proposed the research problem, directed the successive revisions, and supplied detailed comments on both the mathematical arguments and their presentation.

Internal AI-assisted mathematical and editorial review passes were followed by further revisions. These checks were self-reviews, not independent external peer review or formal verification. Numerical calculations were used as finite-dimensional checks of identities and estimates, not as substitutes for proofs. The author is responsible for the mathematical statements, citations, and presentation in the manuscript.

\begingroup
\raggedright
\bibliographystyle{alphaurl}
\bibliography{references}

@book{KLW2026,
  author = {Khatri, Sumeet and Lami, Ludovico and Wilde, Mark M.},
  title = {Principles of Quantum Communication Theory: A Modern Approach},
  publisher = {Zenodo},
  year = {2026},
  doi = {10.5281/zenodo.21763149},
  note = {Version 1, August 2, 2026.}
}

@article{CWY2004,
  author = {Cai, Ning and Winter, Andreas and Yeung, Raymond W.},
  title = {Quantum privacy and quantum wiretap channels},
  journal = {Problems of Information Transmission},
  volume = {40}, number = {4}, pages = {318--336}, year = {2004},
  doi = {10.1007/s11122-005-0002-x}
}

@article{Devetak2005,
  author = {Devetak, Igor},
  title = {The private classical capacity and quantum capacity of a quantum channel},
  journal = {IEEE Transactions on Information Theory},
  volume = {51}, number = {1}, pages = {44--55}, year = {2005},
  doi = {10.1109/TIT.2004.839515},
  url = {https://arxiv.org/abs/quant-ph/0304127}
}

@article{DS2005,
  author = {Devetak, Igor and Shor, Peter W.},
  title = {The capacity of a quantum channel for simultaneous transmission of classical and quantum information},
  journal = {Communications in Mathematical Physics},
  volume = {256}, number = {2}, pages = {287--303}, year = {2005},
  doi = {10.1007/s00220-005-1317-6},
  url = {https://arxiv.org/abs/quant-ph/0311131}
}

@article{Smith2008,
  author = {Smith, Graeme},
  title = {Private classical capacity with a symmetric side channel and its application to quantum cryptography},
  journal = {Physical Review A},
  volume = {78}, number = {2}, pages = {022306}, year = {2008},
  doi = {10.1103/PhysRevA.78.022306}
}

@article{MW2014,
  author = {Morgan, Ciara and Winter, Andreas},
  title = {{``Pretty strong''} converse for the quantum capacity of degradable channels},
  journal = {IEEE Transactions on Information Theory},
  volume = {60}, number = {1}, pages = {317--333}, year = {2014},
  doi = {10.1109/TIT.2013.2288971},
  url = {https://arxiv.org/abs/1301.4927v3}
}

@article{WTB2017,
  author = {Wilde, Mark M. and Tomamichel, Marco and Berta, Mario},
  title = {Converse bounds for private communication over quantum channels},
  journal = {IEEE Transactions on Information Theory},
  volume = {63}, number = {3}, pages = {1792--1817}, year = {2017},
  doi = {10.1109/TIT.2017.2648825},
  url = {https://arxiv.org/abs/1602.08898}
}

@misc{KBK2026,
  author = {Kondra, Tulja Varun and Brinster, Raphael and Kampermann, Hermann and Bru{\ss}, Dagmar and Wyderka, Nikolai},
  title = {Sharp Quantum Capacity Thresholds: Exponential Strong Converses for Degradable and Antidegradable Channels},
  year = {2026},
  note = {arXiv:2608.01308v1},
  doi = {10.48550/arXiv.2608.01308},
  url = {https://arxiv.org/abs/2608.01308v1}
}

@misc{BKH2025,
  author = {Baghali Khanian, Zahra and Hirche, Christoph},
  title = {On Strong Converse Bounds for the Private and Quantum Capacities of Anti-degradable Channels},
  year = {2025}, note = {arXiv:2507.15661},
  doi = {10.48550/arXiv.2507.15661},
  url = {https://arxiv.org/abs/2507.15661}
}

@article{HHHO2005,
  author = {Horodecki, Karol and Horodecki, Micha{\l} and Horodecki, Pawe{\l} and Oppenheim, Jonathan},
  title = {Secure key from bound entanglement},
  journal = {Physical Review Letters},
  volume = {94}, number = {16}, pages = {160502}, year = {2005},
  doi = {10.1103/PhysRevLett.94.160502}
}

@article{HHHO2009,
  author = {Horodecki, Karol and Horodecki, Micha{\l} and Horodecki, Pawe{\l} and Oppenheim, Jonathan},
  title = {General paradigm for distilling classical key from quantum states},
  journal = {IEEE Transactions on Information Theory},
  volume = {55}, number = {4}, pages = {1898--1929}, year = {2009},
  doi = {10.1109/TIT.2008.2009798},
  url = {https://arxiv.org/abs/quant-ph/0506189}
}

@article{HOW2007,
  author = {Horodecki, Micha{\l} and Oppenheim, Jonathan and Winter, Andreas},
  title = {Quantum state merging and negative information},
  journal = {Communications in Mathematical Physics},
  volume = {269}, pages = {107--136}, year = {2007},
  doi = {10.1007/s00220-006-0118-x}
}

@article{DBWR2014,
  author = {Dupuis, Fr{\'e}d{\'e}ric and Berta, Mario and Wullschleger, J{\"u}rg and Renner, Renato},
  title = {One-shot decoupling},
  journal = {Communications in Mathematical Physics},
  volume = {328}, pages = {251--284}, year = {2014},
  doi = {10.1007/s00220-014-1990-4},
  url = {https://arxiv.org/abs/1012.6044}
}

@article{CKR2009,
  author = {Christandl, Matthias and K{\"o}nig, Robert and Renner, Renato},
  title = {Postselection technique for quantum channels with applications to quantum cryptography},
  journal = {Physical Review Letters},
  volume = {102}, number = {2}, pages = {020504}, year = {2009},
  doi = {10.1103/PhysRevLett.102.020504}
}

@article{TCR2009,
  author = {Tomamichel, Marco and Colbeck, Roger and Renner, Renato},
  title = {A fully quantum asymptotic equipartition property},
  journal = {IEEE Transactions on Information Theory},
  volume = {55}, number = {12}, pages = {5840--5847}, year = {2009},
  doi = {10.1109/TIT.2009.2032797}
}

@phdthesis{Tomamichel2012,
  author = {Tomamichel, Marco},
  title = {A Framework for Non-Asymptotic Quantum Information Theory},
  school = {ETH Zurich}, year = {2012},
  url = {https://arxiv.org/abs/1203.2142},
  note = {In particular, Theorem 6.4 and Corollary 6.5}
}

@article{TCR2010,
  author = {Tomamichel, Marco and Colbeck, Roger and Renner, Renato},
  title = {Duality between smooth min- and max-entropies},
  journal = {IEEE Transactions on Information Theory},
  volume = {56}, number = {9}, pages = {4674--4681}, year = {2010},
  doi = {10.1109/TIT.2010.2054130}
}

@article{deWolf2008,
  author = {de Wolf, Ronald},
  title = {A note on quantum algorithms and the minimal degree of epsilon-error polynomials for symmetric functions},
  journal = {Quantum Information and Computation},
  volume = {8}, number = {10}, pages = {943--950}, year = {2008},
  doi = {10.26421/QIC8.10-4},
  url = {https://arxiv.org/abs/0802.1816}
}

@article{Uhlmann1976,
  author = {Uhlmann, Armin},
  title = {The ``transition probability'' in the state space of a $*$-algebra},
  journal = {Reports on Mathematical Physics},
  volume = {9}, number = {2}, pages = {273--279}, year = {1976},
  doi = {10.1016/0034-4877(76)90060-4}
}

@article{FvdG1999,
  author = {Fuchs, Christopher A. and van de Graaf, Jeroen},
  title = {Cryptographic distinguishability measures for quantum-mechanical states},
  journal = {IEEE Transactions on Information Theory},
  volume = {45}, number = {4}, pages = {1216--1227}, year = {1999},
  doi = {10.1109/18.761271}
}

@article{BDS1997,
  author = {Bennett, Charles H. and DiVincenzo, David P. and Smolin, John A.},
  title = {Capacities of quantum erasure channels},
  journal = {Physical Review Letters},
  volume = {78}, number = {16}, pages = {3217--3220}, year = {1997},
  doi = {10.1103/PhysRevLett.78.3217}
}

@article{Rastegin2002,
  author = {Rastegin, Alexey E.},
  title = {Relative error of state-dependent cloning},
  journal = {Physical Review A},
  volume = {66}, number = {4}, pages = {042304}, year = {2002},
  doi = {10.1103/PhysRevA.66.042304}
}

@misc{Rastegin2006,
  author = {Rastegin, Alexey E.},
  title = {Sine distance for quantum states},
  year = {2006}, note = {arXiv:quant-ph/0602112},
  doi = {10.48550/arXiv.quant-ph/0602112},
  url = {https://arxiv.org/abs/quant-ph/0602112}
}

@article{Stinespring1955,
  author = {Stinespring, W. Forrest},
  title = {Positive functions on {$C^*$}-algebras},
  journal = {Proceedings of the American Mathematical Society},
  volume = {6}, number = {2}, pages = {211--216}, year = {1955},
  doi = {10.1090/S0002-9939-1955-0069403-4}
}

@article{LiebRuskai1973,
  author = {Lieb, Elliott H. and Ruskai, Mary Beth},
  title = {Proof of the strong subadditivity of quantum-mechanical entropy},
  journal = {Journal of Mathematical Physics},
  volume = {14}, number = {12}, pages = {1938--1941}, year = {1973},
  doi = {10.1063/1.1666274}
}

@misc{Wilde2026PureLoss,
  author = {Wilde, Mark M.},
  title = {Strong converse for the quantum capacity of the pure-loss bosonic channel},
  year = {2026},
  note = {arXiv:2609.16608v1},
  doi = {10.48550/arXiv.2609.16608},
  url = {https://arxiv.org/abs/2609.16608v1}
}

@article{HSW2008,
 author = {Hayden, Patrick and Shor, Peter W. and Winter, Andreas},
 title = {Random quantum codes from {Gaussian} ensembles and an uncertainty relation},
 journal = {Open Systems \& Information Dynamics},
 volume = {15}, number = {1}, pages = {71--89}, year = {2008},
 doi = {10.1142/S1230161208000079},
 url = {https://arxiv.org/abs/0712.0975}
}

@article{Wilde2017Position,
  author = {Wilde, Mark M.},
  title = {Position-based coding and convex splitting for private communication over quantum channels},
  journal = {Quantum Information Processing},
  volume = {16}, number = {10}, pages = {264}, month = oct, year = {2017},
  doi = {10.1007/s11128-017-1718-4},
  url = {https://arxiv.org/abs/1703.01733}
}

@article{HHHLO2008QKD,
  author = {Horodecki, Karol and Horodecki, Micha{\l} and Horodecki, Pawe{\l} and Leung, Debbie and Oppenheim, Jonathan},
  title = {Quantum key distribution based on private states: Unconditional security over untrusted channels with zero quantum capacity},
  journal = {IEEE Transactions on Information Theory},
  volume = {54}, number = {6}, pages = {2604--2620}, month = jun, year = {2008},
  doi = {10.1109/TIT.2008.921870},
  url = {https://arxiv.org/abs/quant-ph/0608195}
}

@article{HHHLO2008Privacy,
  author = {Horodecki, Karol and Horodecki, Micha{\l} and Horodecki, Pawe{\l} and Leung, Debbie and Oppenheim, Jonathan},
  title = {Unconditional privacy over channels which cannot convey quantum information},
  journal = {Physical Review Letters},
  volume = {100}, number = {11}, pages = {110502}, month = mar, year = {2008},
  doi = {10.1103/PhysRevLett.100.110502},
  url = {https://arxiv.org/abs/quant-ph/0702077}
}

@article{Wilde2013Sequential,
  author = {Wilde, Mark M.},
  title = {Sequential decoding of a general classical-quantum channel},
  journal = {Proceedings of the Royal Society A},
  volume = {469}, number = {2157}, pages = {20130259}, month = sep, year = {2013},
  doi = {10.1098/rspa.2013.0259},
  url = {https://arxiv.org/abs/1303.0808}
}

@misc{salzmann2022totalinsecurity,
      title={Total insecurity of communication via strong converse for quantum privacy amplification}, 
      author={Robert Salzmann and Nilanjana Datta},
      year={2022},
      eprint={2202.11090},
      archivePrefix={arXiv},
      primaryClass={quant-ph},
      url={https://arxiv.org/abs/2202.11090}, 
}
\endgroup

\appendix

\section{The symmetric dilation of a degradable channel}\label{app:dilation}
This appendix proves Lemma~\ref{lem:factorization}. The construction is due to Morgan and Winter~\cite[Lemmas~1 and~17, and the discussion following Theorem~19]{MW2014}. We include it because equality of the two marginal channels, by itself, does not imply that a particular Stinespring isometry has range in a symmetric subspace. The latter property is the one preserved by the single-filter reduction. Figure~\ref{fig:symmetrization} separates the two constructive steps below.

\subsection{An involutive dilation}
We first make exchange of the two complementary outputs equivalent to a simple involution on a retained auxiliary system. The flag construction in Figure~\ref{fig:symmetrization}(a) enforces this symmetry without changing either complementary marginal.

Choose an isometry $U\colon A\to BE$ for $\cN$ and an isometry $V_0\colon B\to F_0E'$ for a degrading channel, with $E'\cong E$. Define $T\coloneqq(V_0\otimes I_E)U$. The two marginal channels from $A$ to $E$ and to $E'$ both equal $\cN^c$. Introduce a qubit $H$ and the isometry
\begin{equation}\label{eq:flagged-isometry}
 T_1\coloneqq\frac1{\sqrt2}\left(T\otimes\ket0_H
 +(I_{F_0}\otimes\FS_{EE'})T\otimes\ket1_H\right).
\end{equation}
The two terms are orthogonal on $H$, so $T_1^\dagger T_1=I_A$. Tracing out $F_0HE'$ yields $\cN^c$ because each flagged term has that same marginal. Let $F\coloneqq F_0H$ and let $X_F$ act as a Pauli $X$ on $H$ and as the identity on $F_0$. Direct substitution gives
\begin{equation}\label{eq:involution}
 (X_F\otimes\FS_{EE'})T_1=T_1,\qquad X_F^2=I_F.
\end{equation}
The isometries $U$ and $T_1$ dilate the same complementary channel. By Stinespring uniqueness, there is an isometry $V\colon B\to FE'$ such that
\begin{equation}\label{eq:V-exists}
 T_1=(V\otimes I_E)U.
\end{equation}
For this use of uniqueness, one can first restrict $B$ to the support of the channel outputs, making its role as the purifying space of $\cN^c$ minimal. Alternatively, extend the resulting partial isometry on the unused part of $B$, enlarging $F_0$ if necessary. Neither choice changes the channel on its actual outputs. The $E'$ marginal of $T_1$ is also $\cN^c$, so $V$ dilates a degrading channel.

Equation~\eqref{eq:involution} is a symmetry with an additional involution on $F$. We next remove that involution by appending an independent mixed qubit at the receiver.

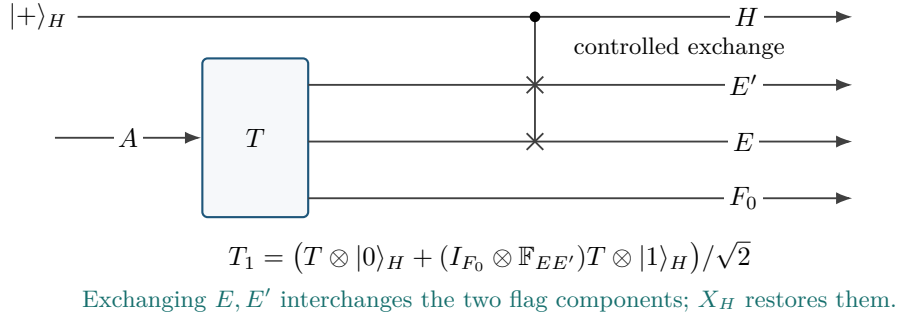
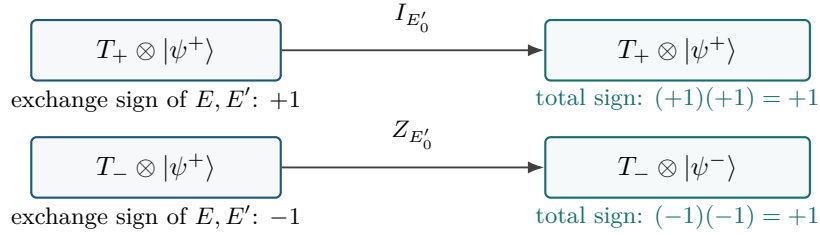
\begin{figure}
\centering
\begin{tikzpicture}[x=1cm,y=1cm,>=Latex,font=\small,
 wire/.style={draw=black!75,line width=.7pt},
 box/.style={draw=figblue,fill=figblue!4,rounded corners=2pt,line width=.8pt},
 lab/.style={fill=white,inner sep=1.7pt}]
 \node[anchor=west,font=\small\bfseries] at (0,4.15)
 {(a) A coherent flag symmetrizes the two terms};
 \node[box,minimum width=1.4cm,minimum height=2.1cm] (T) at (3.0,1.6) {$T$};
 \draw[wire,->] (.35,1.6)--(T.west) node[midway,lab] {$A$};
 \draw[wire,->] (.65,3.2)--(10.9,3.2);
 \node[anchor=east] at (.65,3.2) {$\ket+_H$};
 \node[lab] at (9.5,3.2) {$H$};
 \draw[wire,->] (3.7,2.3)--(10.9,2.3) node[pos=.80,lab] {$E'$};
 \draw[wire,->] (3.7,1.55)--(10.9,1.55) node[pos=.80,lab] {$E$};
 \draw[wire,->] (3.7,.8)--(10.9,.8) node[pos=.80,lab] {$F_0$};
 \draw[wire] (6.7,3.2)--(6.7,1.55);
 \fill (6.7,3.2) circle (2pt);
 \foreach \y in {2.3,1.55}{
  \draw[wire] (6.59,\y-.11)--(6.81,\y+.11);
  \draw[wire] (6.59,\y+.11)--(6.81,\y-.11);
 }
 \node[font=\footnotesize,anchor=west] at (7.08,2.78) {controlled exchange};
 \node at (6.1,.02)
 {$T_1=\bigl(T\otimes\ket0_H+(I_{F_0}\otimes\mathbb F_{EE'})T\otimes\ket1_H\bigr)/\sqrt2$};
 \node[text=figteal,font=\footnotesize] at (6.1,-.54)
 {Exchanging $E,E'$ interchanges the two flag components; $X_H$ restores them.};
 \draw[black!25] (0,-1.08)--(12.75,-1.08);
 \node[anchor=west,font=\small\bfseries] at (0,-1.62)
 {(b) An added Bell pair cancels the exchange sign};
 \node[box,minimum width=3.3cm,minimum height=.8cm] (plusL) at (2.15,-2.6)
 {$T_+\otimes\ket{\psi^+}$};
 \node[box,minimum width=3.3cm,minimum height=.8cm] (minusL) at (2.15,-4.15)
 {$T_-\otimes\ket{\psi^+}$};
 \node[box,draw=figteal,fill=figteal!4,minimum width=3.5cm,minimum height=.8cm] (plusR) at (9.05,-2.6)
 {$T_+\otimes\ket{\psi^+}$};
 \node[box,draw=figteal,fill=figteal!4,minimum width=3.5cm,minimum height=.8cm] (minusR) at (9.05,-4.15)
 {$T_-\otimes\ket{\psi^-}$};
 \draw[wire,->] (plusL.east)--(plusR.west) node[midway,above=3pt,font=\footnotesize] {$I_{E'_0}$};
 \draw[wire,->] (minusL.east)--(minusR.west) node[midway,above=3pt,font=\footnotesize] {$Z_{E'_0}$};
 \node[font=\footnotesize] at (2.15,-3.25) {exchange sign of $E,E'$: $+1$};
 \node[font=\footnotesize] at (2.15,-4.8) {exchange sign of $E,E'$: $-1$};
 \node[text=figteal,font=\footnotesize] at (9.05,-3.25) {total sign: $(+1)(+1)=+1$};
 \node[text=figteal,font=\footnotesize] at (9.05,-4.8) {total sign: $(-1)(-1)=+1$};
 \node[align=center,font=\footnotesize] at (6.15,-5.65)
 {The two sectors remain in coherent superposition.\\
 The combined exchange swaps $EE_0$ with $E'E'_0$.};
\end{tikzpicture}\par
\caption{Constructing the symmetric range in Appendix~\ref{app:dilation}. In (a), a coherent flag $H$ records whether the two complementary outputs are exchanged. Swapping $E,E'$ interchanges the two flag components and is compensated by $X_H$, giving $(X_H\otimes\mathbb F_{EE'})T_1=T_1$. In (b), the $+1$ and $-1$ sectors of this flag involution have, respectively, symmetric and antisymmetric $E,E'$ outputs. A receiver-only controlled phase correlates them with $\ket{\psi^+}$ and $\ket{\psi^-}$ on an added Bell pair. Each term has total exchange sign $+1$. The added receiver qubit has maximally mixed marginal, and its purifying qubit belongs to the environment. Both panels depict isometries and coherent superpositions, not measurements of the flag or its sign.}
\label{fig:symmetrization}
\end{figure}

\subsection{Removing the involution}
The next step correlates the two eigenspaces of the involution with a symmetric or antisymmetric Bell state. The two exchange signs then cancel, giving an exactly symmetric pair of enlarged outputs; see Figure~\ref{fig:symmetrization}(b).

Let $P_+$ and $P_-$ be the two spectral projections of $X_F$, and define
\begin{equation}\label{eq:spectral-involution}
 P_\pm\coloneqq\frac12(I_F\pm X_F),\qquad
 T_\pm\coloneqq(P_\pm\otimes I_{E'E})T_1.
\end{equation}
The operators $T_+$ and $T_-$ are components of the total isometry, not generally isometries on their own. They obey $T_1=T_++T_-$ and $T_+^\dagger T_-=0$. In the construction below they remain in coherent superposition; the sign is not measured.
It follows from~\eqref{eq:involution} that
\begin{equation}\label{eq:sector-symmetry}
 (I_F\otimes\FS_{EE'})T_+=T_+,
 \qquad (I_F\otimes\FS_{EE'})T_-=-T_-.
\end{equation}
Append a maximally mixed qubit $B_0$ to the channel output and purify it by a qubit $E_0$. Thus
\begin{equation}\label{eq:augmented-channel}
 \widetilde\cN(\rho)\coloneqq\cN(\rho)\otimes\pi_{B_0},
 \qquad \pi_{B_0}\coloneqq I_{B_0}/2,
 \qquad \widetilde U\coloneqq U\otimes\ket{\psi^+}_{B_0E_0},
\end{equation}
where $\ket{\psi^\pm}\coloneqq(\ket{01}\pm\ket{10})/\sqrt2$. After applying $V$ to $B$, rename $B_0$ as $E'_0$ and apply the controlled unitary
\begin{equation}\label{eq:controlled-phase}
 P_+\otimes I_{E'_0}+P_-\otimes Z_{E'_0}.
\end{equation}
Let $\widetilde V$ denote this isometry from $BB_0$ to $FE'E'_0$. In a fixed ordering of tensor factors,
\begin{equation}\label{eq:total-symmetric}
 \widetilde V\widetilde U
 =T_+\otimes\ket{\psi^+}_{E'_0E_0}
 +T_-\otimes\ket{\psi^-}_{E'_0E_0}.
\end{equation}
The first term is symmetric under both swaps, and the second term is antisymmetric under each swap separately. Both terms are consequently symmetric under the simultaneous exchange of $EE_0$ with $E'E'_0$:
\begin{equation}\label{eq:full-swap}
 (I_F\otimes\FS_{EE_0:E'E'_0})\widetilde V\widetilde U
 =\widetilde V\widetilde U.
\end{equation}

It remains to check that $\widetilde V$ is a degrading dilation. In the input $V\cN(\rho)V^\dagger\otimes\pi_{E'_0}$ to~\eqref{eq:controlled-phase}, the diagonal $P_+$ and $P_-$ terms leave the mixed qubit invariant. The cross terms vanish after tracing out $F$, since $P_+P_-=0$ and operators acting only on the traced system can be cycled under that partial trace. Hence
\begin{equation}\label{eq:augmented-degrading}
 \Tr_F[\widetilde V\widetilde\cN(\rho)\widetilde V^\dagger]
 =\cN^c(\rho)\otimes\pi_{E'_0}
 =\widetilde\cN^{\,c}(\rho)
\end{equation}
under the identification $E'_0\cong E_0$. This proves the desired degrading property.

Now relabel $BB_0$ as $B$ and $EE_0$ as $E$. Let $G\coloneqq\Sym^2(E)$ and let $W\colon G\to E'E$ be the inclusion. Equation~\eqref{eq:full-swap} implies that the range of $VU$ lies in $F\otimes G$. Define
\begin{equation}\label{eq:J-explicit}
 J\coloneqq(I_F\otimes W^\dagger)VU.
\end{equation}
Since $WW^\dagger$ is the projection onto $G$ and acts as the identity on the range in question,
\begin{equation}\label{eq:J-isometry}
 J^\dagger J
 =U^\dagger V^\dagger(I_F\otimes WW^\dagger)VU
 =I_A.
\end{equation}
Equations~\eqref{eq:J-explicit} and~\eqref{eq:J-isometry} establish the factorization in Lemma~\ref{lem:factorization}.

\subsection{Preservation of coherent information and the joint criterion}
We finish by checking that the auxiliary qubit does not alter either the capacity threshold or the code fidelity. These checks allow us to apply the factorization to the original coding problem without an additional approximation.

For every input state, the appended qubits add one bit to each of the receiver and environment entropies. Therefore
\begin{equation}\label{eq:q-invariance}
 H(\widetilde\cN(\rho))-H(\widetilde\cN^{\,c}(\rho))
 =H(\cN(\rho))-H(\cN^c(\rho)).
\end{equation}
The maximized coherent information is unchanged.

Given a code for $\cN$, let the decoder ignore $B_0^n$ in a code for $\widetilde\cN$. Its classical output state with the environment is exactly $\omega_{K\widehat K E^n}\otimes\pi_{E_0}^{\otimes n}$. For this calculation, abbreviate $\omega\coloneqq\omega_{K\widehat K E^n}$, $\overline\Phi\coloneqq\overline\Phi^M_{K\widehat K}$, and $\pi_0\coloneqq\pi_{E_0}^{\otimes n}$. For an arbitrary $\sigma_{E^nE_0^n}$, data processing of fidelity under the partial trace gives
\begin{equation}\label{eq:fidelity-augmentation-upper}
 F(\omega\otimes\pi_0,\overline\Phi\otimes\sigma_{E^nE_0^n})
 \le F(\omega,\overline\Phi\otimes\sigma_{E^n})
 \le f(\mathscr C_n).
\end{equation}
Conversely, the product of an optimizing original $\sigma_{E^n}$ with $\pi_{E_0}^{\otimes n}$ attains $f(\mathscr C_n)$ by multiplicativity of fidelity. Thus the optimized joint fidelity is exactly preserved, not merely bounded in one direction. The actual environment marginal is $\omega_{E^n}\otimes\pi_{E_0}^{\otimes n}$, so multiplicativity also shows that $f_{\mathrm{marg}}$ is exactly preserved.

\section{Uniform smooth-entropy estimates}\label{app:entropies}
The entropy estimate in Proposition~\ref{prop:overhead} must hold simultaneously for all encoder inputs. This appendix gives the definitions, states the precise standard smooth-entropy inequalities used, and derives the correlated-input estimate. It then explains why the independent-copy asymptotic equipartition estimate can be made uniform in the single-copy state, even when the smoothing parameter is exponentially small.

\subsection{Conventions and elementary inputs}
We fix the smoothing convention and collect the entropy inequalities used to compare a correlated input with a mixture of independent-copy inputs. The distinction between normalized and subnormalized states matters for these comparisons.

Write $\substates(A)\coloneqq\{\rho\ge0:\Tr\rho\le1\}$. For subnormalized states, the generalized fidelity and purified distance are defined as
\begin{equation}\label{eq:generalized-fidelity}
 F_*(\rho,\sigma)\coloneqq
 \left(\trn{\sqrt\rho\sqrt\sigma}
 +\sqrt{(1-\Tr\rho)(1-\Tr\sigma)}\right)^2,
 \qquad P(\rho,\sigma)\coloneqq\sqrt{1-F_*(\rho,\sigma)}.
\end{equation}
This agrees with the sine distance in~\eqref{eq:purified} when both states are normalized, and $F_*=F$ when at least one is normalized. For $\rho_{AB}\in\substates(AB)$, define
\begin{equation}\label{eq:min-definition}
 \hmin(A|B)_\rho\coloneqq
 \sup\left\{\lambda\in\mathbb R:
 \rho_{AB}\le2^{-\lambda}I_A\otimes\sigma_B
 \text{ for some }\sigma_B\in\states(B)\right\}.
\end{equation}
The conditional max-entropy is defined by the fidelity expression
\begin{equation}\label{eq:max-definition}
 \hmax(A|B)_\rho\coloneqq
 \max_{\sigma\in\states(B)}\log F(\rho_{AB},I_A\otimes\sigma_B),
\end{equation}
where $F$ in~\eqref{eq:max-definition} is the algebraic expression~\eqref{eq:fidelity}, not the generalized expression~\eqref{eq:generalized-fidelity}. The second argument in~\eqref{eq:max-definition} need not have trace one.

For $0\le\varepsilon<1$, the smoothed entropies are defined as
\begin{equation}\label{eq:smooth-definitions}
 \begin{aligned}
 \hmin^\varepsilon(A|B)_\rho
 &\coloneqq\sup\left\{\hmin(A|B)_{\widetilde\rho}:
 \widetilde\rho\in\substates(AB),\ P(\widetilde\rho,\rho)\le\varepsilon\right\},\\
 \hmax^\varepsilon(A|B)_\rho
 &\coloneqq\inf\left\{\hmax(A|B)_{\widetilde\rho}:
 \widetilde\rho\in\substates(AB),\ P(\widetilde\rho,\rho)\le\varepsilon\right\}.
 \end{aligned}
\end{equation}
For a purification $\psi_{ABR}$, smooth duality states that~\cite{TCR2010}
\begin{equation}\label{eq:smooth-duality}
 \hmin^\varepsilon(A|R)_\psi=-\hmax^\varepsilon(A|B)_\psi.
\end{equation}
Both quantities are invariant under local isometries. Smoothing makes $\hmin^\varepsilon$ nondecreasing and $\hmax^\varepsilon$ nonincreasing in $\varepsilon$.

We use four established inequalities with normalized center states. They are stated here to fix all parameters. Proofs are given in Ref.~\cite[Lemmas~8--11]{MW2014}, with the min--max comparison also following from the smooth-entropy treatment in Ref.~\cite{Tomamichel2012}. By Ref.~\cite[Lemma~10]{MW2014}, if $\overline\rho=\sum_i p_i(U_i\otimes V_i)\rho(U_i\otimes V_i)^\dagger$ is an average over local unitaries, then
\begin{equation}\label{eq:orbit-entropy}
 \hmax^{\sqrt2\varepsilon}(A|B)_\rho
 \le\hmax^\varepsilon(A|B)_{\overline\rho}
 \quad(0<\varepsilon<1/\sqrt2).
\end{equation}
By Ref.~\cite[Lemma~9]{MW2014}, for $0<\varepsilon<1$,
\begin{equation}\label{eq:max-to-min}
 \hmax^\varepsilon(A|B)_\rho
 \le\hmin^{(1-\varepsilon^2)^{1/4}}(A|B)_\rho
 \le\hmin^{1-\varepsilon^2/4}(A|B)_\rho.
\end{equation}
Ref.~\cite[Lemma~8]{MW2014} gives, for $r,s\ge0$ with $r+s<1$,
\begin{equation}\label{eq:min-to-max}
 \hmin^r(A|B)_\rho
 \le\hmax^s(A|B)_\rho+\log\frac1{1-(r+s)^2}.
\end{equation}
Finally, Ref.~\cite[Lemma~11]{MW2014} shows that, if $\overline\rho=\sum_{i=1}^Jp_i\rho^i$ is a mixture of $J$ states, then
\begin{equation}\label{eq:mixture-max}
 \hmax^\varepsilon(A|B)_{\overline\rho}
 \le\max_{1\le i\le J}\hmax^\varepsilon(A|B)_{\rho^i}+\log J.
\end{equation}
These statements do not assert concavity of the \emph{smoothed} max-entropy. In particular, the smoothing change in~\eqref{eq:orbit-entropy} must be retained. We now use these standard inequalities to prove the needed bound, keeping every smoothing change explicit.

\subsection{Postselection for arbitrary channel inputs}
We now derive the correlated-input bound~\eqref{eq:postselection-main}. The argument first averages over permutations, then uses operator domination by a mixture of tensor powers, and finally transfers the smoothing parameters while paying only logarithmic dimension terms.

Let $\cT\colon A\to FE'$ be an arbitrary fixed channel, let $\rho_{A^n}$ be arbitrary, and define
\begin{equation}\label{eq:permutation-average}
 \overline\rho_{A^n}\coloneqq\frac1{n!}\sum_{\pi\in S_n}
 U_\pi\rho_{A^n}U_\pi^\dagger,
 \qquad \theta\coloneqq\cT^{\otimes n}(\rho),\qquad
 \overline\theta\coloneqq\cT^{\otimes n}(\overline\rho).
\end{equation}
Because $\cT^{\otimes n}$ commutes with tensor-factor permutations, $\overline\theta$ is an average of $\theta$ over unitaries that are local across $F^n:E'^n$. Equation~\eqref{eq:orbit-entropy} with $\varepsilon=u/\sqrt2$ gives
\begin{equation}\label{eq:postselection-step1}
 \hmax^u(F^n|E'^n)_\theta
 \le\hmax^{u/\sqrt2}(F^n|E'^n)_{\overline\theta}
 \le\hmin^{1-u^2/8}(F^n|E'^n)_{\overline\theta},
\end{equation}
where the second inequality is~\eqref{eq:max-to-min}.

We next recall the operator domination supplied by postselection~\cite{CKR2009}. Set $a_A=|A|$ and
\begin{equation}\label{eq:symmetric-dimension}
 D_n\coloneqq\binom{n+a_A^2-1}{a_A^2-1},\qquad
 g_n\coloneqq(n+1)^{a_A^2}.
\end{equation}
There is a probability measure $\mu$ on $\states(A)$ such that
\begin{equation}\label{eq:postselection-domination}
 \overline\rho\le D_n\tau_{A^n}\le g_n\tau_{A^n},
 \qquad \tau_{A^n}\coloneqq\int\sigma^{\otimes n}\,\mathrm d\mu(\sigma).
\end{equation}
Here is a direct explanation of the dimension factor. A canonical purification of a permutation-invariant $\overline\rho$ on $A^nA'^n$, with $A'\cong A$, is invariant under simultaneous permutations of the $n$ pairs $AA'$. It lies in $\Sym^n(A\otimes A')$, whose dimension is $D_n$. A normalized rank-one projection onto that purification is bounded by the projection onto this symmetric subspace. The normalized symmetric projection is the Haar average of $\varphi_{AA'}^{\otimes n}$ over pure $\varphi_{AA'}$. Taking the partial trace over $A'^n$ yields~\eqref{eq:postselection-domination}. Finally, $D_n\le(n+1)^{a_A^2-1}\le g_n$ follows by counting the occupation numbers of the symmetric subspace.

Set $\Theta\coloneqq\cT^{\otimes n}(\tau)$. Complete positivity implies $\overline\theta\le g_n\Theta$. Write $\alpha\coloneqq u^2/8$. Consider an arbitrary subnormalized $\zeta$ with $P(\zeta,\overline\theta)\le1-\alpha$. Since the center is normalized,
\begin{equation}\label{eq:fidelity-ball-lower}
 F(\zeta,\overline\theta)\ge1-(1-\alpha)^2=2\alpha-\alpha^2\ge\alpha.
\end{equation}
The operator inequality $\overline\theta\le g_n\Theta$ and monotonicity of the square root give
\begin{equation}\label{eq:fidelity-domination}
 \begin{aligned}
 \sqrt{F(\zeta,\overline\theta)}
 &=\Tr\sqrt{\sqrt\zeta\,\overline\theta\,\sqrt\zeta}\\
 &\le\sqrt{g_n}\,\Tr\sqrt{\sqrt\zeta\,\Theta\,\sqrt\zeta}
 =\sqrt{g_nF(\zeta,\Theta)}.
 \end{aligned}
\end{equation}
Thus $F(\zeta,\Theta)\ge\alpha/g_n$ and
\begin{equation}\label{eq:smoothing-ball-transfer}
 P(\zeta,\Theta)\le\sqrt{1-\alpha/g_n}
 \le1-\frac{\alpha}{2g_n}.
\end{equation}
Every candidate in the min-entropy smoothing ball about $\overline\theta$ is consequently a candidate in the larger ball about $\Theta$. Define
\begin{equation}\label{eq:beta-v}
 \beta\coloneqq\frac{u^2}{16g_n},\qquad
 v\coloneqq\frac\beta2=\frac{u^2}{32g_n}.
\end{equation}
Equations~\eqref{eq:postselection-step1} and~\eqref{eq:smoothing-ball-transfer} yield
\begin{equation}\label{eq:postselection-step2}
 \hmax^u(F^n|E'^n)_\theta
 \le\hmin^{1-\beta}(F^n|E'^n)_\Theta.
\end{equation}
Apply~\eqref{eq:min-to-max} with $r=1-\beta$ and $s=\beta/2$. Since
\begin{equation}\label{eq:comparison-penalty}
 1-(r+s)^2=1-(1-\beta/2)^2
 =\beta-\beta^2/4\ge\beta/2,
\end{equation}
we obtain
\begin{equation}\label{eq:postselection-step3}
 \hmax^u(F^n|E'^n)_\theta
 \le\hmax^v(F^n|E'^n)_\Theta+\log\frac{32g_n}{u^2}.
\end{equation}

To use~\eqref{eq:mixture-max}, a finite mixture suffices. The Haar-averaged symmetric projection acts on a $D_n$-dimensional space. By the finite-dimensional convex-hull theorem, it can be expressed as a convex combination of at most $D_n^2$ pure tensor-power projections. Tracing out $A'^n$ and applying $\cT^{\otimes n}$ gives
\begin{equation}\label{eq:finite-de-finetti}
 \Theta=\sum_{i=1}^Jp_i\cT(\sigma_i)^{\otimes n},\qquad
 J\le D_n^2\le g_n^2.
\end{equation}
Equation~\eqref{eq:mixture-max} now implies
\begin{equation}\label{eq:postselection-step4}
 \begin{aligned}
 \hmax^u(F^n|E'^n)_\theta
 \le{}&\max_{\sigma\in\states(A)}\hmax^{u^2/(32g_n)}(F^n|E'^n)_{\cT(\sigma)^{\otimes n}}\\
 &+3\log g_n+5+2\log(1/u).
 \end{aligned}
\end{equation}
Replacing $5$ by $6$ gives~\eqref{eq:postselection-main}. All choices and bounds above are independent of $\rho_{A^n}$. No product-state or full-rank restriction has been imposed on the code input.

\subsection{A dimension-uniform independent-copy estimate}
For completeness, we justify the uniform constant and the full range of blocklengths in~\eqref{eq:aep-main}. Let $\theta_{FE'}$ be arbitrary, write $d_F=|F|$, and define
\begin{equation}\label{eq:aep-functions}
 b(v)\coloneqq\log(2/v^2),\qquad
 g(v)\coloneqq-\log\!\left(1-\sqrt{1-v^2}\right).
\end{equation}
The elementary bound $\sqrt{1-v^2}\le1-v^2/2$ gives $g(v)\le b(v)$.

The finite-block asymptotic equipartition theorem~\cite[Theorem~6.4 and Corollary~6.5]{Tomamichel2012}, using min--max duality when necessary, gives
\begin{equation}\label{eq:aep-original}
 \hmax^v(F^n|E'^n)_{\theta^{\otimes n}}
 \le nH(F|E')_\theta+4\log\Upsilon\,\sqrt{ng(v)}
 \quad\text{when }n\ge\tfrac85g(v).
\end{equation}
The convergence parameter can be bounded by
\begin{equation}\label{eq:upsilon-bound}
 \Upsilon\le
 \sqrt{2^{-\hmin(F|B)_\xi}}+
 \sqrt{2^{\hmax(F|B)_\xi}}+1
 \le2\sqrt{d_F}+1.
\end{equation}
Here $\xi_{FB}$ denotes the relevant marginal of a purification in the min-entropy version of the theorem. Its conditioning system need not equal $E'$. The last inequality uses only the normalized-state dimension bounds $-\log d_F\le\hmin(F|B)_\xi$ and $\hmax(F|B)_\xi\le\log d_F$. Thus it applies equally after passing to the purifying system for the max-entropy bound.

If $n\ge\tfrac85b(v)$, then the condition in~\eqref{eq:aep-original} holds, and
\begin{equation}\label{eq:aep-large-n}
 \hmax^v(F^n|E'^n)_{\theta^{\otimes n}}-nH(F|E')_\theta
 \le4\log(2\sqrt{d_F}+1)\sqrt{nb(v)}
 \le c_F\sqrt{nb(v)}.
\end{equation}
If instead $n<\tfrac85b(v)$, use the elementary bounds $\hmax^v(F^n|E'^n)\le n\log d_F$ and $H(F|E')\ge-\log d_F$. They imply
\begin{equation}\label{eq:aep-small-n}
 \begin{aligned}
 \hmax^v(F^n|E'^n)_{\theta^{\otimes n}}-nH(F|E')_\theta
 &\le2n\log d_F\\
 &\le2\sqrt{8/5}\,\log d_F\sqrt{nb(v)}
 \le c_F\sqrt{nb(v)}.
 \end{aligned}
\end{equation}
The last inequality follows from $c_F=8\log(2d_F+1)$. The two cases prove~\eqref{eq:aep-main} for every $n\ge1$ and $0<v<1$.

The small-blocklength case is included because $v$ will itself depend exponentially on $n$. An asymptotic statement with a state- or smoothing-dependent threshold would not be enough without this check. Likewise, an equipartition bound whose constant involves the smallest nonzero eigenvalue would not give the uniform maximum over $\sigma$ required in~\eqref{eq:postselection-main}.

\section{The polynomial approximation input}\label{app:boolean}
The complete matrix construction and averaging argument are given in Section~\ref{subsec:signed-averaging}. This appendix specifies the only external approximation theorem used there and explains its specialization to the Boolean function $\operatorname{NOR}_n$. It also records the normalization of the Fourier transform in terms familiar from quantum information theory. Neither a new approximation theorem nor a new signed-weight construction is claimed.

\subsection{Specializing the approximate-degree theorem}
We first explain why the approximation degree has the form used in~\eqref{eq:approx-degree-main}. For a Boolean function $f\colon\{0,1\}^n\to\{0,1\}$, its $\epsilon$-approximate degree is the minimum degree of a real polynomial whose value differs from $f$ by at most $\epsilon$ at every point of the cube. Values away from the cube are not constrained. A polynomial can be made multilinear by replacing every power $z_j^k$ with $z_j$ for $k\ge1$; this preserves its values on the cube and cannot increase its degree.

Theorem~1 of Ref.~\cite{deWolf2008} states, in particular, that for a nonconstant symmetric Boolean function the $\epsilon$-approximate degree is bounded above by a universal constant times the sum of its constant-error approximate degree and $\sqrt{n\ln(1/\epsilon)}$, for $2^{-n}\le\epsilon\le1/3$. The function $\operatorname{NOR}_n$ is symmetric because it depends only on Hamming weight. Its constant-error approximate degree is $O(\sqrt n)$, also covered by the upper bound in the discussion following that theorem: $\operatorname{NOR}_n$ is constant on all nonzero weights. Absorbing constants, including the base of the logarithm, gives a universal $\Capp\ge1$ and a real multilinear polynomial $p$ satisfying
\begin{equation}\label{eq:approx-degree}
 \deg p\le\Capp\bigl(\sqrt n+\sqrt{n\ln(1/\epsilon)}\bigr),
 \qquad \max_{z\in\{0,1\}^n}|p(z)-\operatorname{NOR}_n(z)|\le\epsilon.
\end{equation}
This is the approximation statement quoted in the proof of Lemma~\ref{lem:weights}. The proof there checks the full allowed error interval for its choice of $\epsilon$, verifies the degree bound, and normalizes $p(0^n)$ without changing the degree.

The small-error dependence is important. When $t$ is a fixed positive fraction of $n$, the chosen $\epsilon=\exp[-t^2/(4\Capp^2n)]$ is exponentially small. A constant-error approximation alone would not yield an exponentially small $\op{W_t-U_0}$. On the other hand, the degree restriction is what forces the support condition $s(x,y)\le t$. Both parts of~\eqref{eq:approx-degree} are therefore used in the converse.

\subsection{The Hadamard transform and the two bases}
For clarity, the label-space transform used in Lemma~\ref{lem:weights} can be written as the familiar unitary
\begin{equation}\label{eq:hadamard-transform}
 H_2\coloneqq\frac1{\sqrt2}\begin{pmatrix}1&1\\1&-1\end{pmatrix},
 \qquad (H_2^{\otimes n})_{xz}=2^{-n/2}(-1)^{x\cdot z}.
\end{equation}
Its columns are the vectors $\chi_z$. If $q$ is viewed as a column vector of its values, then the coefficient vector in~\eqref{eq:walsh} is
$\widehat q=2^{-n/2}H_2^{\otimes n}q$.
Unitarity gives Parseval's identity~\eqref{eq:parseval-main}. The matrix in~\eqref{eq:W-construction} has the spectral representation
\begin{equation}\label{eq:weight-spectral-representation}
 W_t=H_2^{\otimes n}\operatorname{diag}\bigl(((-1)^{|z|}q(z))_z\bigr)H_2^{\otimes n},
 \qquad U_0=(\ket+\!\bra+)^{\otimes n}.
\end{equation}
Here $\ket+\coloneqq(\ket0+\ket1)/\sqrt2$, and $\operatorname{diag}((a_z)_z)$ denotes the diagonal operator with diagonal entries $a_z$. The operator $W_t$ is close to $U_0$ because all its nonuniform Hadamard eigenvectors have small eigenvalues. Its restricted support in the coordinate basis instead comes from the low degree of $q$. This distinction between an entrywise support statement and a spectral approximation statement is essential; neither can be substituted for the other.

\section{Using the actual environment marginal}\label{app:marginal}
The optimized ideal environment state in~\eqref{eq:ideal} need not be the actual marginal of the code output. This appendix treats the alternative requirement that it be exactly $\omega_{E^n}$. The distinction and its relation to separate decoding and secrecy errors were discussed in Ref.~\cite[Appendix~B]{WTB2017}. In particular, that reference already notes that using the actual marginal does not weaken its converse bounds. The same elementary observation applies to our finite-block bound and preserves its exponential rate without any loss.

\subsection{Comparison with the optimized criterion}
We first order the two fidelities and compare their errors near zero. The fidelity ordering, rather than the reverse error estimate, is the fact that preserves the strong-converse exponent.

For a fixed code, abbreviate $f\coloneqq f(\mathscr C_n)$ and $f_{\mathrm{marg}}\coloneqq f_{\mathrm{marg}}(\mathscr C_n)$, with the latter defined in~\eqref{eq:marginal-definition-main}. Both fidelities concern the same state $\omega_{K\widehat K E^n}$ and the same uniform input prior. Only the choice of the ideal environment state differs.

\begin{proposition}[Optimized and actual-marginal fidelities]\label{prop:marginal-comparison}
Every code satisfies
\begin{equation}\label{eq:marginal-order}
 f_{\mathrm{marg}}\le f,\qquad
 \varepsilon\le\varepsilon_{\mathrm{marg}}
 \le\min\{1,4\varepsilon\}.
\end{equation}
Consequently, the two uniform joint criteria have the same vanishing-error capacity. Every finite-block upper bound on $f$ is also an upper bound on $f_{\mathrm{marg}}$.
\end{proposition}
\begin{proof}
The actual marginal $\omega_{E^n}$ is a density operator and is therefore an admissible candidate in the optimization defining $f$. This proves $f_{\mathrm{marg}}\le f$, or equivalently $\varepsilon\le\varepsilon_{\mathrm{marg}}$.

For the reverse error comparison, choose an optimizer $\sigma^*_{E^n}$ in~\eqref{eq:ideal}. By definition,
\begin{equation}\label{eq:optimized-sine}
 P\!\left(\omega_{K\widehat K E^n},
 \overline\Phi^M_{K\widehat K}\otimes\sigma^*_{E^n}\right)
 =\sqrt\varepsilon.
\end{equation}
Tracing out $K\widehat K$ and using monotonicity of the sine distance gives
\begin{equation}\label{eq:marginal-close}
 P(\omega_{E^n},\sigma^*_{E^n})\le\sqrt\varepsilon.
\end{equation}
Multiplicativity of fidelity implies that adjoining the same normalized state to both arguments leaves their sine distance unchanged. The triangle inequality for the sine distance~\cite{Rastegin2006} therefore gives
\begin{equation}\label{eq:marginal-triangle}
 P\!\left(\omega_{K\widehat K E^n},
 \overline\Phi^M_{K\widehat K}\otimes\omega_{E^n}\right)
 \le
 P\!\left(\omega_{K\widehat K E^n},
 \overline\Phi^M_{K\widehat K}\otimes\sigma^*_{E^n}\right)
 +P(\sigma^*_{E^n},\omega_{E^n})
 \le2\sqrt\varepsilon.
\end{equation}
Squaring yields $\varepsilon_{\mathrm{marg}}\le4\varepsilon$. The additional upper bound one follows because fidelity lies in $[0,1]$. Thus one error tends to zero if and only if the other does, proving the capacity statement.
\end{proof}

\begin{corollary}[Exponential strong converse with the actual marginal]\label{cor:marginal}
Let $\cN$ be a finite-dimensional degradable channel and fix $\Delta>0$. The same constants $\gamma,n_0$ as in Theorem~\ref{thm:main} satisfy
\begin{equation}\label{eq:marginal-strong-converse}
 F\!\left(\omega_{K\widehat K E^n},
 \overline\Phi^M_{K\widehat K}\otimes\omega_{E^n}\right)
 \le2^{-\gamma n}
\end{equation}
for every unassisted code with $n\ge n_0$ and $\log M\ge n(\qN+\Delta)$. The corresponding joint infidelity is at least $1-2^{-\gamma n}$. The same transfer applies to the antidegradable-channel and erasure-channel conclusions of Section~\ref{sec:converse}.
\end{corollary}
\begin{proof}
Apply $f_{\mathrm{marg}}\le f$ in~\eqref{eq:marginal-order} to the relevant fidelity upper bound. The inequality $\varepsilon_{\mathrm{marg}}\le4\varepsilon$ is not needed for this converse and causes no change in its exponent.
\end{proof}

There is also a direct way to see compatibility with the proof. In Lemma~\ref{lem:lift}, choose the fixed target $\overline\Phi^M\otimes\omega_{E^n}$ rather than an optimizing target. Uhlmann's theorem then produces a private state $\gamma$ and a test $\Pi^\gamma$ accepted with probability at least $f_{\mathrm{marg}}$. The overlap bound is uniform over private-state shields and twisting unitaries, so the rest of the proof is unchanged. This direct argument is consistent with, but unnecessary for, the simpler ordering argument above.

\subsection{The two finite-error fidelities can be different}
The ordering in~\eqref{eq:marginal-order} must not be replaced by equality. We give an example that is realized by an unassisted code for a degradable channel. Let $M\ge2$, let $0<p<1$, and let $m\oplus1$ denote the next element in the cyclic order on $\{1,\ldots,M\}$. Consider
\begin{equation}\label{eq:marginal-example}
 \begin{aligned}
 \omega_{K\widehat K E}\coloneqq\frac1M\sum_{m=1}^M\bigl[
 &p\ket{mm}\!\bra{mm}_{K\widehat K}\otimes\ket0\!\bra0_E\\
 &+(1-p)\ket{m,m\oplus1}\!\bra{m,m\oplus1}_{K\widehat K}
 \otimes\ket1\!\bra1_E\bigr].
 \end{aligned}
\end{equation}
Its environment marginal is $\omega_E=p\ket0\!\bra0+(1-p)\ket1\!\bra1$. Only the correctly decoded blocks have support in common with the ideal key state. The classical-block fidelity formula therefore gives, for an arbitrary density operator $\sigma_E$,
\begin{equation}\label{eq:marginal-example-fidelity}
 F(\omega_{K\widehat K E},\overline\Phi^M_{K\widehat K}\otimes\sigma_E)
 =p\langle0|\sigma_E|0\rangle.
\end{equation}
Taking $\sigma_E=\ket0\!\bra0_E$ attains the maximum $p$, whereas using $\omega_E$ gives
\begin{equation}\label{eq:marginal-example-values}
 f=p,\qquad f_{\mathrm{marg}}=p^2.
\end{equation}

To realize the example, take $|A|=|B_1|=M$ and the channel
\begin{equation}\label{eq:marginal-example-channel}
 \cN_p(\rho)\coloneqq\rho_{B_1}\otimes
 \bigl(p\ket0\!\bra0+(1-p)\ket1\!\bra1\bigr)_{B_0}.
\end{equation}
A Stinespring isometry transmits $A$ unchanged to $B_1$ and appends
$\sqrt p\ket{00}_{B_0E}+\sqrt{1-p}\ket{11}_{B_0E}$.
The complementary channel is constant, so it is reproduced by discarding the receiver input and preparing $\omega_E$; hence $\cN_p$ is degradable. Encode $m$ as $\ket m$, measure $B_1$ in this basis and $B_0$ in the flag basis, and report $m$ when the flag is zero and $m\oplus1$ when it is one. The resulting state is exactly~\eqref{eq:marginal-example}. The decoder deliberately makes errors on the second flag; its purpose here is to distinguish the criteria, not to optimize communication.

\subsection{Consequence for separate decoding and secrecy errors}
We finally spell out a useful implication for the two-error formulation. This comparison is Proposition~28 of Ref.~\cite{WTB2017}; we give a proof in our notation to specify the environment state and the fidelity convention. Define
\begin{equation}\label{eq:separate-errors}
 p_{\mathrm{err}}\coloneqq\Pr\{K\ne\widehat K\},
 \qquad
 \delta_{\mathrm{sec}}\coloneqq
 \frac12\trn{\omega_{KE^n}-\pi_K\otimes\omega_{E^n}},
 \qquad \pi_K\coloneqq I_K/M.
\end{equation}
Then
\begin{equation}\label{eq:separate-to-joint}
 1-\sqrt{f_{\mathrm{marg}}}
 \le p_{\mathrm{err}}+\delta_{\mathrm{sec}}.
\end{equation}
To prove this, write the classical--classical--quantum state as
\begin{equation}\label{eq:block-output}
 \omega_{K\widehat K E^n}
 =\sum_{m,j}\ket{mj}\!\bra{mj}_{K\widehat K}\otimes A^{mj}_{E^n},
 \qquad A^{mj}_{E^n}\ge0,
\end{equation}
where $\sum_j\Tr A^{mj}=1/M$. Introduce a state that replaces the estimate by a correct copy of $K$, while keeping the $KE^n$ marginal:
\begin{equation}\label{eq:corrected-estimate}
 \zeta_{K\widehat K E^n}
 \coloneqq\sum_m\ket{mm}\!\bra{mm}_{K\widehat K}
 \otimes\sum_j A^{mj}_{E^n}.
\end{equation}
For every incorrect block $m\ne j$, the difference $\omega-\zeta$ has a positive block $A^{mj}$. In the correct block $mm$, it has the negative block $-\sum_{j\ne m}A^{mj}$. The trace norm of a block-diagonal operator is the sum of the trace norms of its blocks. Hence
\begin{equation}\label{eq:corrected-trace}
 \frac12\trn{\omega-\zeta}
 =\sum_{m\ne j}\Tr A^{mj}=p_{\mathrm{err}}.
\end{equation}
The isometry $\ket m_K\mapsto\ket{mm}_{K\widehat K}$ sends $\omega_{KE^n}$ to $\zeta$ and $\pi_K\otimes\omega_{E^n}$ to $\overline\Phi^M\otimes\omega_{E^n}$. Isometric invariance of the trace norm therefore gives
\begin{equation}\label{eq:corrected-security}
 \frac12\trn{\zeta-\overline\Phi^M\otimes\omega_{E^n}}
 =\delta_{\mathrm{sec}}.
\end{equation}
The triangle inequality and the other Fuchs--van de Graaf inequality~\cite{FvdG1999} now imply
\begin{equation}\label{eq:separate-proof}
 1-\sqrt{f_{\mathrm{marg}}}
 \le\frac12\trn{\omega-\overline\Phi^M\otimes\omega_{E^n}}
 \le p_{\mathrm{err}}+\delta_{\mathrm{sec}},
\end{equation}
proving~\eqref{eq:separate-to-joint}. Combining it with Corollary~\ref{cor:marginal} gives
\begin{equation}\label{eq:separate-converse}
 p_{\mathrm{err}}+\delta_{\mathrm{sec}}
 \ge1-2^{-\gamma n/2}
\end{equation}
under the same rate and blocklength conditions. This lower bound concerns the \emph{sum} of the two errors. It does not say that either error separately must converge to one, and it does not identify the separate-error definitions with either joint criterion at a fixed nonzero error.

\end{document}